\documentclass[11pt]{article}
\usepackage{graphicx}
\usepackage{epstopdf}
\usepackage{amssymb,bm}
\usepackage{bbm}
\usepackage{amsmath}
\usepackage{amsfonts}
\usepackage[margin=1.0in]{geometry}
\usepackage{cite}
\usepackage[shortlabels]{enumitem}

\usepackage{color}

\DeclareMathOperator{\E}{\mathbbmss{E}}

\allowdisplaybreaks

\usepackage{amsthm}
\newtheorem{theorem}{Theorem}

\newtheorem{lemma}{Lemma}

\theoremstyle{definition}  
\newtheorem{remark}{Remark}

\newtheorem{example}{Example}

\usepackage{chngcntr}
\counterwithout{theorem}{section}
\counterwithout{definition}{section}
\counterwithout{lemma}{section}
\counterwithout{remark}{section}
\counterwithout{assumption}{section}
\counterwithout{proposition}{section}
\counterwithout{corollary}{section}
\counterwithout{claim}{section}

\begin{document}
\title{On the Optimality of Treating Inter-Cell Interference as Noise in Uplink Cellular Networks
\footnotetext{This work was partially supported by the U.K. Engineering and Physical Sciences Research Council (EPSRC)
under grant EP/N015312/1. This paper was presented in part at the 2018 IEEE
International Symposium on Information Theory \cite{Joudeh2018b}.}}
\author{Hamdi~Joudeh and Bruno~Clerckx \\
\small
Communications and Signal Processing Group, Department of Electrical and Electronic Engineering \\
\small
Imperial College London, London SW7 2AZ, United Kingdom\\
\small
 Email: \{hamdi.joudeh10, b.clerckx\}@imperial.ac.uk}
\date{}
\maketitle
\begin{abstract}
In this paper, we explore the information-theoretic optimality of treating interference as noise (TIN) in
cellular networks.
We focus on uplink scenarios modeled by the Gaussian interfering multiple access channel (IMAC), comprising $K$ mutually interfering multiple access channels (MACs), each formed by an arbitrary number of transmitters communicating independent messages to one receiver.
We define TIN for this setting as a scheme in which each MAC (or cell) performs a power-controlled version of its capacity-achieving
strategy, with Gaussian codebooks and successive decoding,  while treating interference from all other MACs (i.e. inter-cell interference) as noise.
We characterize the generalized degrees-of-freedom (GDoF) region achieved through the proposed TIN scheme,
and then identify conditions under which this achievable region is convex without the need for time-sharing.
We then tighten these convexity conditions and identify a regime in which the proposed TIN scheme achieves the entire GDoF region of the IMAC
and is within a constant gap of the entire capacity region.
\end{abstract}
\newpage
\section{Introduction}
\label{sec:introduction}
Transmitter power control in conjunction with treating interference as noise (TIN) at receivers
is a key  principle for interference management in wireless networks.
Schemes based on TIN are attractive in practice due to their relative simplicity and robustness.
From a theoretical point of view, TIN received considerable research attention
mainly due to its information-theoretic optimality (and near-optimality) in several settings and regimes.
This is best exemplified by the 2-user Gaussian interference channel (IC), for which the capacity region
remains one of the longest standing open problems in network information theory.
The 2-user IC capacity question, while formidable in its generality, has been settled for few special cases; see for example \cite[Ch. 6]{ElGamal2011}.
One such special case is the \emph{noisy interference regime}, in which interference is sufficiently weak such that a simple
TIN scheme, where each transmitter uses its full power, achieves the exact sum-capacity \cite{Annapureddy2009,Shang2009,Motahari2009}.

Beyond the sum-capacity of the 2-user IC, e.g. for the entire capacity region or $K$-user ICs, the problem of identifying regimes
in which TIN is optimal from an exact capacity viewpoint becomes significantly more difficult.
In such cases, power control and time-sharing play a pivotal role in achieving different rate trade-offs; transmitting at full power is
generally not optimal when TIN is in use, while time-sharing between different power control strategies generally \emph{convexifies} (and enlarges) TIN-achievable rate regions \cite{Charafeddine2012}.
This resource allocation problem is known to be hard in general \cite{Luo2008}, giving
rise to intricate TIN-achievable rate regions which are difficult to analyse \cite{Charafeddine2012}, let alone
characterizing regimes in which such regions coincide with corresponding information-theoretic outer bounds.
Nevertheless, it was shown by Geng \emph{et al.} \cite{Geng2015} that the above  challenges  can be circumvented by taking a step away from the exact capacity and instead, pursuing approximate solutions based on the generalized
degrees-of-freedom (GDoF).

Geng \emph{et al.}'s approach to the $K$-user IC TIN-optimality problem rests on three main cornerstones: 1)
relaxing time-sharing for tractability and relying solely on power control to achieve different GDoF trade-offs\footnote{Note that this is a key step in simplifying the Han-Kobayashi achievable region and establishing the ``capacity to one bit" result for the 2-user IC in \cite{Etkin2008}, which also implicitly includes a TIN-optimal characterization.}, 2) focusing on a convex sub-region of the GDoF region achieved through TIN and power control,
referred to as the \emph{polyhedral TIN-achievable GDoF region}, which is explicitly characterized by
eliminating power control variables with the aid of a combinatorial tool named the \emph{potential graph}, and 3) establishing
optimality of the
polyhedral TIN-achievable GDoF region in the regime of interest by matching it to a genie-aided outer bound.
This approach proved very successful, leading to the characterization of a broad regime of channel parameters
in which TIN achieves the entire GDoF region of the general fully-connected, fully-asymmetric $K$-user IC,  and is within a
constant gap of the entire capacity region \cite{Geng2015}.
The success of this GDoF-based TIN-optimality pursuit called for further
investigation, resulting in extensions and generalizations to other settings including:
channels with general message sets (or $X$ channels) \cite{Geng2015a,Gherekhloo2017},  parallel channels
\cite{Sun2016},  multi-state (or compound) channels \cite{Geng2016}
and multi-state channels with opportunistic decoding capabilities \cite{Yi2018}.
Moreover, Yi and Caire gave a fresh combinatorial perspective on the original $K$-user IC TIN-optimality problem considered in \cite{Geng2015} and identified
a new class of partially connected networks for which TIN achieves the entire GDoF region \cite{Yi2016}.
\subsection{TIN in Cellular Networks}
\begin{figure}
\centering
\includegraphics[width = 0.4\textwidth,trim={9cm 6cm 9cm 6cm},clip]{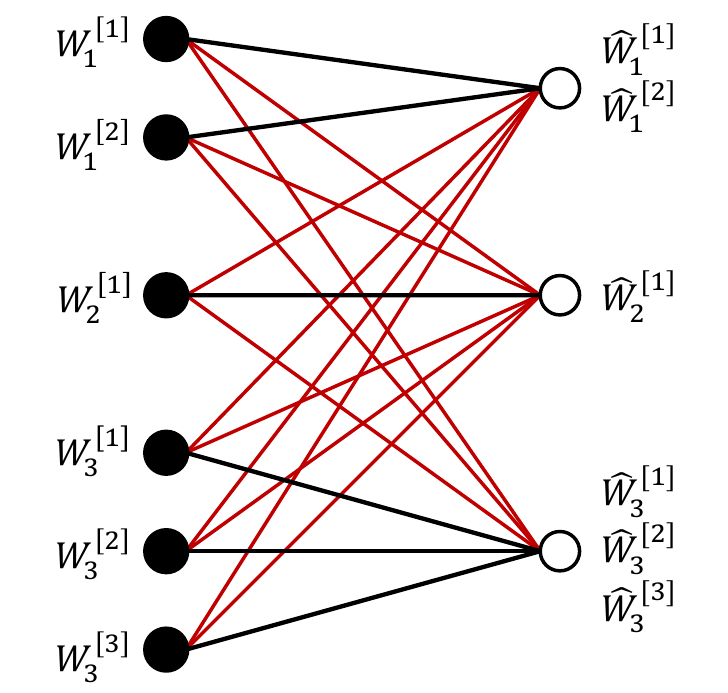}
\caption{\small
An uplink cellular network consisting of  $K = 3$ cells and $L_{1} = 2$, $L_{2} = 1$ and $L_{3} = 3$
users (transmitters) in cells 1, 2 and 3, respectively.
Direct links (between transmitters and their corresponding receivers)
are marked in black and interference links (between transmitters and non-corresponding receivers) are marked in red.
Note that although cell $2$ in the above network is a point-to-point link and not a MAC, we still refer to such network as an IMAC
as the two remaining cells are interfering MACs.}
\label{fig:IMAC_example}
\end{figure}
The optimality of TIN in cellular-like networks has been considered through the lens of the general
$X$ channel  \cite{Geng2015a,Gherekhloo2017}.
In \cite{Geng2015a}, the authors showed that under the TIN-optimality condition identified in \cite{Geng2015},
operating the $X$ channel as a regular IC and treating interference as noise is optimal from a sum-GDoF perspective and achieves the sum-capacity up to a constant gap --- for $M \times N$ channels with arbitrary $M $ and  $N$, some nodes are switched off and a cyclic modification of the condition in \cite{Geng2015} is used.
Building upon this result, the authors of \cite{Gherekhloo2017} considered the sum-GDoF of the $M \times 2$ channel and expanded
the TIN-optimal regime of \cite{Geng2015a} for this special case.
Nevertheless, for the purpose of understanding the optimality of TIN in cellular networks,
the setting and results in \cite{Geng2015a} and \cite{Gherekhloo2017} offer a high degree of generality,  perhaps more than needed, in one aspect and less generality in another.
Specifically, on one hand, the $X$ channel allows each transmitter to communicate an independent message to each receiver,
e.g. a cellular scenario where all users transmit independent steams to all base stations.
On the other hand, restricting the analysis to the sum-GDoF (and sum-capacity) gives limited insights into the different
trade-offs that can be achieved and reveals little about special cases of the $X$ channel that resemble more realistic
settings, e.g. a classical cellular scenario where each user associates with the closest base station.
In this paper, we make progress towards a comprehensive and crystalized understanding of TIN in cellular networks
by constraining the former of the two above aspects and relaxing the latter.

We consider a cellular network in the classical sense, consisting of an arbitrary number of cells, where each cell is formed by one base station
and an arbitrary number of users.
We further focus on uplink scenarios, in which each user wishes to communicate an independent message to the corresponding base station.
Such uplink cellular scenarios are captured by the Gaussian interfering multiple access channel (IMAC) \cite{Suh2008}, as illustrated in
Fig. \ref{fig:IMAC_example}.
Moreover, we seek a general TIN strategy for the IMAC that achieves the entire GDoF region (and capacity region up to a constant gap), as opposed to
the sum-GDoF only, under specific TIN-optimality conditions.
The optimality of TIN for a special case of this channel, named the PIMAC and consisting of
a point-to-point link and a 2-user multiple access channel (MAC) that mutually interfere,
was studied by Gherekhloo \emph{et al.} in \cite{Gherekhloo2016}.
In particular, Gherekhloo \emph{et al.} identified regimes for the PIMAC in which a simple time-sharing-TIN scheme is sum-GDoF optimal and achieves the sum-capacity within a constant gap.
However, the specificity of the results and analysis in \cite{Gherekhloo2016} to the sum-GDoF of this 2-cell, 3-user network hinders their extension to more general IMAC scenarios.

A crucial initial step before commencing the pursuit of TIN-optimality results for the IMAC is establishing
an adequate definition of TIN for such channel.
By viewing the $K$-user IC as a cellular network with one user in each cell,
TIN can be interpreted as the employment of a single-cell, capacity-achieving transmission strategy in
each cell, while  treating all inter-cell interference as noise.
This definition of TIN naturally extends to the cellular setting at hand.
More importantly, a TIN strategy for the IMAC, in accordance with the above definition, satisfies the requirement for
robustness, as capacity-achieving strategies for the MAC do not demand channel knowledge at the transmitters
beyond the coarse level assumed to be available in known TIN schemes.

Next, we move on to presenting
an overview of this work's main results and contributions.
A detailed exposition of such results, with insights and examples, is given in Section \ref{sec:main_results}.
\subsection{Main Results and Contributions}
\subsubsection{TIN-achievable GDoF region for the IMAC}
We propose a TIN scheme for the IMAC in which a MAC-type, capacity-achieving strategy, with Gaussian codebooks and successive decoding,
is employed in each cell while treating all inter-cell interference as noise\footnote{We focus on TIN schemes that employ unstructured random Gaussian codes throughout this work. This excludes schemes that use codes with (some) structure,
e.g. the TIN scheme with mixed inputs in \cite{Dytso2016}.}.
This scheme is complemented with power control to manage inter-cell interference and achieve different GDoF (or rate) trade-offs.
We follow the tradition of disallowing time-sharing for the sake of tractability
\cite{Etkin2008,Geng2015,Geng2015a,Gherekhloo2017,Geng2016,Sun2016,Yi2018,Yi2016}.
The resulting achievable GDoF region is therefore obtained by considering all feasible power control strategies
and successive decoding orders in each cell.
To distinguish this region from different restricted versions that appear throughout the work, we refer to it as the
\emph{general TIN-achievable GDoF region} henceforth.

To obtain an explicit characterization of the general TIN-achievable GDoF region, that does not depend on power control variables, we focus on sub-regions achieved with fixed decoding orders.
We then seek to characterize restricted (smaller) achievable sub-regions, known as  \emph{polyhedral TIN-achievable regions},
through a non-trivial application of the potential graph approach in \cite{Geng2015,Geng2015a}.
Polyhedral TIN-achievable regions are then employed as building blocks in characterizing
the general TIN-achievable GDoF region, which in turn is shown to be a finite union of polyhedra.

One major challenge in applying the potential graph approach compared to \cite{Geng2015,Geng2015a} is identifying and eliminating
redundant directed circuits (and their corresponding GDoF inequalities), which arise due to the special structure of the IMAC under the proposed TIN scheme (see Section \ref{subsec:TIN_region_proof}).
This step proves crucial for matching the achievable region with the outer bound derived later on to
establish the GDoF optimality of the proposed TIN scheme in the regime of interest.
\subsubsection{Conditions for TIN-Convexity}
After establishing a TIN-achievable GDoF region for the IMAC, the natural question to ask next is:
under what conditions is this achievable region optimal?
As an intermediate step towards answering this question,
we consider the closely related issue of determining conditions under
which this TIN-achievable GDoF region is a polyhedron\footnote{By a polyhedron,
we are referring to a convex set given by the intersection
of a finite number of half spaces. Since we are dealing with GDoF regions,
we only encounter bounded polyhedra which are therefore convex polytopes.}, and hence convex, in its own right without requiring time-sharing.
We identify a regime for which this holds that
is characterized by two conditions referred to as the \emph{TIN-convexity conditions} (see Theorem \ref{theorem:polyhedrality} in Section \ref{subsec:TIN_polyhedrality}).
The first of the two conditions guarantees that one successive decoding order for each cell dominates all others in a GDoF sense, such that it is sufficient to only consider this decoding order to achieve the entire general TIN-achievable GDoF region.
The second condition is essentially the TIN-convexity condition of the $K$-user IC, identified by Yi and Caire in \cite{Yi2016},
applied to all possible $K$-user IC subnetworks of the considered IMAC.
This condition guarantees that the TIN region achieved through the dominant decoding order is in itself convex.
\subsubsection{Conditions for TIN-Optimality}
We further strengthen the TIN-convexity conditions and
obtain a set of \emph{TIN-optimality conditions} under which the general TIN-achievable GDoF
region is also optimal (see Theorem \ref{theorem:TIN_optimality} in Section \ref{subsec:TIN_optimality}).
The TIN-optimality conditions are merely stronger versions of the two aforementioned TIN-convexity conditions and include the
TIN-optimality condition of Geng \emph{et al.} \cite{Geng2015}, applied to all possible $K$-user IC subnetworks of the considered
IMAC.

We prove the TIN-optimality result by deriving an outer bound which coincides with the general TIN-achievable GDoF
region in the regime of interest.
We recall that in the converse used to establish the TIN-optimality condition for the $K$-user IC in \cite{Geng2015},
each bound featuring more than one user is obtained by first reducing the channel to a cyclic (sub)network and then
directly applying the Etkin-Tse-Wang (ETW) genie \cite{Etkin2008} (see also \cite{Zhou2013} where the cyclic IC is considered).
We follow in the same general footsteps by first considering cyclic (sub)networks, where cyclicity
is taken with respect to participating cells. We then use a non-trivial genie-aided argument which extends
the ETW genie to cope with the multi-user per-cell setting at hand.
In particularly, the genie signal given for each cell is taken as a noisy linear combination of
in-cell signals, where the weights of such linear combinations (i.e. the genie channels) are carefully designed to yield the desired bounds
in the regime of interest (see Section \ref{subsec:proof_theorem_outer_bound} for details).
This outer bound also directly lends itself to showing that under the identified TIN-optimality conditions, the proposed TIN scheme achieves the entire capacity region of the IMAC up to a constant gap.
\subsection{Notation}
For any positive integers $z_{1}$ and $z_{2}$, where $z_{1} \leq z_{2}$, the sets $\{1,2,\ldots,z_{1}\}$ and $\{z_{1},z_{1}+1,\ldots,z_{2}\}$ are denoted by
$\langle z_{1} \rangle$ and $\langle z_{1}:z_{2}\rangle$ respectively.
For any real number $a$, $(a)^{+} = \max\{0,a\}$.
Bold lowercase symbols  denote tuples, e.g. $\mathbf{a} = (a_{1},\ldots,a_{Z})$.
For $\mathcal{A} = \{\mathbf{a}_{1},\ldots,\mathbf{a}_{K}\}$, $\Sigma(\mathcal{A})$
is the set of all cyclic sequences formed by any number of elements in $\mathcal{A}$ without repetitions,
e.g.
\begin{equation}
\nonumber
\Sigma\big(\{\mathbf{a}_{1},\mathbf{a}_{2},\mathbf{a}_{3}\}\big) =
\big\{(\mathbf{a}_{1}),(\mathbf{a}_{2}),(\mathbf{a}_{3}),(\mathbf{a}_{1},\mathbf{a}_{2}),(\mathbf{a}_{1},\mathbf{a}_{3}),(\mathbf{a}_{2},\mathbf{a}_{3}),
(\mathbf{a}_{1},\mathbf{a}_{2},\mathbf{a}_{3}),(\mathbf{a}_{1},\mathbf{a}_{3},\mathbf{a}_{2})  \big\}.
\end{equation}
The complement of set $\mathcal{A}$ is denoted by $\overline{\mathcal{A}}$. The cardinality of set
$\mathcal{A}$ is denoted by $|\mathcal{A}|$,  where $|\emptyset| = 0$.
The indicator function on set $\mathcal{A}$ is defined as
\begin{equation}
\nonumber
\mathbbm{1}_{\mathcal{A}}( \mathbf{a} ) =
\begin{cases}
1, \ \text{if}  \ \mathbf{a}  \in \mathcal{A}\\
0, \ \text{if}  \ \mathbf{a}  \notin \mathcal{A}.
\end{cases}
\end{equation}
We sometimes use the alternative definition of the indicator function given by
\begin{equation}
\nonumber
\mathbbm{1}( \text{statement} ) =
\begin{cases}
1, \ \text{if statement is true}\\
0, \ \text{otherwise}.
\end{cases}
\end{equation}
\section{System Model and Preliminaries}
\label{sec:system model}
Consider a $K$-receiver Gaussian IMAC in which each receiver $k$, $k \in \langle K \rangle$, is associated with $L_{k}$ transmitters.
The $k$-th receiver is denoted by Rx-$k$ and the $l_{k}$-th transmitter, $l_{k} \in \langle L_{K} \rangle$,
associated with this receiver is denoted by
Tx-$(l_{k},k)$.
Using the terminology of cellular networks, a receiver and its associated
transmitters are referred to as a cell, operating in the uplink mode.
The set of tuples corresponding to transmitters (or users) in cell $k$ is given by
$\mathcal{K}_{k} \triangleq \left\{(l_{k},k) : l_{k} \in \langle L_{K} \rangle  \right\}$, $k \in \langle K \rangle$,
and the set of all users in the network is given by
$\mathcal{K} \triangleq \bigcup_{k \in \langle K \rangle} \mathcal{K}_{k}$.

The input-output relationship at the $t$-th use of the channel, $t \in \mathbb{N}$, is described as
\begin{equation}
\label{eq:system model}
Y_{i}(t) = \sum_{k = 1}^{K} \sum_{l_{k} = 1}^{L_{k}}  h_{ki}^{[l_{k}]} \tilde{X}_{k}^{[l_{k}]}(t)
+ Z_{i}(t), \ \forall i \in
\langle K \rangle,
\end{equation}
where $h_{ki}^{[l_{k}]}$ is the channel coefficient from Tx-$(l_{k},k)$ to Rx-$i$,
$\tilde{X}_{k}^{[l_{k}]}(t)$ is the transmitted symbol of Tx-$(l_{k},k)$ and  $Z_{i}(t) \sim \mathcal{N}_{\mathbb{C}}(0,1)$
is the additive white Gaussian noise (AWGN) at Rx-$i$, which is i.i.d over
channel uses (time).
All symbols are complex and each transmitter $(l_{k},k)$ is subject to the power constraint
\begin{equation}
\frac{1}{n}\sum_{t=1}^{n}\E \Big[\big|\tilde{X}_{k}^{[l_{k}]}(t)\big|^{2}\Big] \leq P_{k}^{[l_{k}]}.
\end{equation}
Note that receivers are indexed by the subscript, transmitters are indexed by the superscript in square parentheses and channel
uses are indexed by the argument in the round parentheses.

Following the standard reformulation in \cite{Geng2015}, the channel model in \eqref{eq:system model} is transformed into
\begin{equation}
\label{eq:system model 2}
Y_{i}(t) = \sum_{k = 1}^{K} \sum_{l_{k} = 1}^{L_{k}}  \sqrt{ P^{\alpha_{ki}^{[l_{k}]}} } e^{j \theta_{ki}^{[l_{k}]}} X_{k}^{[l_{k}]}(t)
+ Z_{i}(t), \ i \in \langle K \rangle,
\end{equation}
where $P>0$ is a nominal power value and $X_{k}^{[l_{k}]}(t) \triangleq {\tilde{X}_{k}^{[l_{k}]}(t)}/{\sqrt{P_{k}^{[l_{k}]}}}$ is the normalized transmit symbol of Tx-$(l_{k},k)$ with power constraint
\begin{equation}
\frac{1}{n}\sum_{t=1}^{n}\E \Big[\big|X_{k}^{[l_{k}]}(t)\big|^{2}\Big] \leq 1.
\end{equation}
In this equivalent channel, $\sqrt{ P^{\alpha_{ki}^{[l_{k}]}} }$ and $\theta_{ki}^{[l_{k}]}$ are the magnitude and phase of the link between
 Tx-$(l_{k},k)$  and Rx-$i$. The exponent $\alpha_{ki}^{[l_{k}]}$, known as the channel strength level, is defined as
\begin{equation}
\alpha_{ki}^{[l_{k}]} \triangleq \frac{\log \left( \max \big\{ 1, | h_{ki}^{[l_{k}]} |^{2} P_{k}^{[l_{k}]} \big\} \right) }{\log P}
, \ \forall (l_{k},k) \in \mathcal{K}, \
i\in \langle K \rangle.
\end{equation}
As shown in \cite{Geng2015}, avoiding negative channel strength levels has no impact on the GDoF or the constant gap results.
Therefore, we focus on the equivalent channel model in \eqref{eq:system model 2}
henceforth.
Furthermore, without loss of generality, we assume the following order of direct link strength levels
\begin{equation}
\label{eq:strength_order}
\alpha_{kk}^{[1]}  \leq  \alpha_{kk}^{[2]} \leq \cdots \leq \alpha_{kk}^{[L_{k}]}, \ \forall k \in \langle K \rangle.
\end{equation}
\subsection{Messages, Rates, Capacity and GDoF}
\label{subsec:msg_rates_capacity_GDoF}
Tx-$(1,k)$$,\ldots,$Tx-$(L_{k},k)$ have the messages $W_{k}^{[1]},\ldots,W_{k}^{[L_{k}]}$, respectively,
intended to Rx-$k$.
All messages are independent and $|W_{k}^{[l_{k}]}|$ denotes the size of the corresponding message set.
For codewords spanning $n$ channel uses, the rates $R_{k}^{[l_{k}]} = \frac{\log |W_{k}^{[l_{k}]}|}{n}$, $\forall (l_{k},k) \in \mathcal{K}$,
are achievable if  all messages can be decoded simultaneously with arbitrarily small error probability as $n$ grows sufficiently large.
A rate tuple is denoted by
$\mathbf{R}= \big(R_{1}^{[1]},\ldots,R_{1}^{[L_{1}]},\ldots,R_{K}^{[1]},\ldots,R_{K}^{[L_{K}]}\big)$
and the channel capacity region $\mathcal{C}$ is the closure of the set of all achievable rate tuples.
A GDoF tuple is denoted by $\mathbf{d} = \big(d_{1}^{[1]},\ldots,d_{1}^{[L_{1}]},\ldots,d_{K}^{[1]},\ldots,d_{K}^{[L_{K}]}\big)$
and the GDoF region is defined as
\begin{equation}
\mathcal{D} \triangleq \left\{ \mathbf{d} : \  d_{k}^{[l_{k}]} = \lim_{P \rightarrow \infty} \frac{R_{k}^{[l_{k}]}}{\log P}, \
\forall (l_{k},k) \in \mathcal{K}, \ \mathbf{R} \in \mathcal{C}  \right\}.
\end{equation}
\subsection{Treating (Inter-cell) Interference as Noise}
\label{subsec:TIN_scheme}
For a single Gaussian MAC, it is well known that the capacity region is a
polyhedron, where the corner points are achieved using independent Gaussian
codebooks with successive decoding at the receiver, while the remaining points are achieved by further incorporating time-sharing
\cite{Cover2012}.
In the TIN scheme proposed for the IMAC, a MAC-type capacity-achieving strategy with Gaussian codebooks and successive decoding
is employed in each cell, while all inter-cell interference is treated as noise.
Furthermore, power control is employed by transmitters to manage inter-cell interference levels and achieve various tradeoffs\footnote{Note that
such power control is not required to achieve the capacity region for a single Gaussian MAC \cite{Cover2012}.}.
Nevertheless, keeping to the tradition followed in
\cite{Etkin2008,Geng2015,Geng2015a,Gherekhloo2017,Geng2016,Sun2016,Yi2018,Yi2016}, we prohibit time-sharing.
Although this restriction is mainly motivated by tractability,
it remarkably has no influence on the results in the regimes of interest as explained in detail further on.

To formalize the above TIN scheme, let $P^{r_{k}^{[l_{k}]}}$  be the (controlled) transmit power of
Tx-$(l_{k},k)$, where $r_{k}^{[l_{k}]} \leq 0$ denotes the transmit power exponent (or power allocation variable).
The tuple of all power allocation variables is given by
$\mathbf{r} = \big(r_{1}^{[1]},\ldots,r_{1}^{[L_{1}]},\ldots,r_{K}^{[1]},\ldots,r_{K}^{[L_{K}]}\big)$.
On the other hand, the order in which Rx-$k$ successively decodes its
in-cell signals  is given by the permutation function $\pi_{k}:\langle L_{k} \rangle \rightarrow \langle L_{k} \rangle$,
such that $X_{k}^{[\pi_{k}(L_{k})]}$ is decoded and cancelled before decoding $X_{k}^{[\pi_{k}(L_{k}-1)]}$ and so on.
The decoding order across the network is given by the tuple
$\bm{\pi} \triangleq \left( \pi_{1},\ldots,  \pi_{K} \right)$, which is drawn from the set $\Pi$ comprising
all possible $\prod_{i=1}^{K}(L_{i}!)$ network decoding orders.

For a decoding order $\bm{\pi}$ and a power allocation $\mathbf{r}$,
Tx-$\big(\pi_{k}(l_{k}),k\big)$ achieves any rate satisfying
\begin{equation}
\label{eq:rate per user}
0 \leq R_{k}^{[\pi_{k}(l_{k})]}  \leq \log \Biggl( 1+ \frac{ P^{ r_{k}^{[\pi_{k}(l_{k})]} + \alpha_{kk}^{[\pi_{k}(l_{k})]}  }  }
{1  + \sum_{l_{k}' = 1}^{l_{k}-1}P^{ r_{k}^{[\pi_{k}(l_{k}')]} + \alpha_{kk}^{[\pi_{k}(l_{k}')]} }  +
\sum_{j\neq k} \sum_{l_{j} = 1}^{L_{j}} P^{ r_{j}^{[l_{j}]} + \alpha_{jk}^{[l_{j}]} }   } \Biggr).
\end{equation}
In the GDoF sense, the achievable rate in \eqref{eq:rate per user} translates to
\begin{multline}
\label{eq:GDoF per user}
0 \leq d_{k}^{[\pi_{k}(l_{k})]}  \leq \\
\max \biggl\{  0, r_{k}^{[\pi_{k}(l_{k})]}  +   \alpha_{kk}^{[\pi_{k}(l_{k})]}
  -  \Bigl(\max \Bigl\{ \max_{l_{k}' < l_{k}} \{ r_{k}^{[\pi_{k}(l_{k}')]}  +  \alpha_{kk}^{[\pi_{k}(l_{k}')]} \}  ,\max_{j \neq k}
  \max_{l_{j}} \{ r_{j}^{[l_{j}]}  +  \alpha_{jk}^{[l_{j}]} \}   \Bigr\} \Bigr)^{+} \biggr\}.
\end{multline}
For a fixed $\bm{\pi} \in \Pi$, the \emph{TIN-achievable GDoF region}, denoted by $\mathcal{P}_{\bm{\pi}}^{\star}$,
is the set of all GDoF tuples $\mathbf{d}$ with components satisfying \eqref{eq:GDoF per user}
for some feasible power allocation vector $\mathbf{r} \leq \mathbf{0}$.
The \emph{general TIN-achievable GDoF region}, denoted by $\mathcal{P}^{\star}$, is obtained by taking the union over all possible decoding
orders in $\Pi$ and is defined as
\begin{equation}
\label{eq:TIN_GDoF_region}
\mathcal{P}^{\star} \triangleq \bigcup_{\bm{\pi} \in \Pi} \mathcal{P}_{\bm{\pi}}^{\star}.
\end{equation}
Note that since time-sharing is not allowed,
each GDoF tuple $\mathbf{d} \in \mathcal{P}^{\star}$ is achieved through a strategy identified by a decoding order and a power
allocation tuple, i.e. $(\bm{\pi},\mathbf{r})$.

Before we proceed, we highlight that we often work with the identity order $\bm{\pi} = \bm{\mathrm{id}}$ in the following sections,
where $\bm{\mathrm{id}} \triangleq \left( \mathrm{id}_{1},\ldots,\mathrm{id}_{K} \right)$ and $\mathrm{id}_{i}(l_{i}) = l_{i}$,
$\forall(l_{i},i) \in \mathcal{K}$.
\subsection{Polyhedral TIN-Achievable GDoF Regions}
\label{subsec:polyhedral_TIN}
In this part we introduce a \emph{polyhedral TIN scheme} for the IMAC from which we obtain
\emph{polyhedral TIN-achievable GDoF regions}, which form the main building blocks of GDoF characterizations obtained in this work.
For any decoding order $\bm{\pi} \in \Pi$, the polyhedral TIN scheme is a restricted
version of the TIN scheme described in Section \ref{subsec:TIN_scheme} in which $\mathbf{r}$ is chosen such that
the second argument of the outmost $\max\{0,\cdot\}$ in \eqref{eq:GDoF per user} is non-negative.
The resulting polyhedral TIN region, denoted by $\mathcal{P}_{\bm{\pi}}$, is hence described by
all GDoF tuples $\mathbf{d}$ that satisfy
\begin{align}
\label{eq:polyhedral_region_1}
r_{k}^{[\pi_{k}(l_{k})]} & \leq 0,  \  \forall (l_{k},k) \in \mathcal{K}  \\
\label{eq:polyhedral_region_2}
d_{k}^{[\pi_{k}(l_{k})]} & \geq 0,  \  \forall (l_{k},k) \in \mathcal{K}  \\
\nonumber
d_{k}^{[\pi_{k}(l_{k})]} & \leq  r_{k}^{[\pi_{k}(l_{k})]}  +   \alpha_{kk}^{[\pi_{k}(l_{k})]}
  -  \Bigl(\max \Bigl\{ \max_{l_{k}' < l_{k}} \{ r_{k}^{[\pi_{k}(l_{k}')]}  +  \alpha_{kk}^{[\pi_{k}(l_{k}')]} \}  ,\max_{j \neq k}
  \max_{l_{j}} \{ r_{j}^{[l_{j}]}  +  \alpha_{jk}^{[l_{j}]} \}  \Bigr\} \Bigr)^{+}, \\
\label{eq:polyhedral_region_3}
  & \quad    \forall (l_{k},k) \in \mathcal{K},
\end{align}
where it can be  seen from \eqref{eq:polyhedral_region_3} that the outmost $\max\{0,\cdot\}$ in \eqref{eq:GDoF per user} has been dropped.
It follows from this restriction that $\mathcal{P}_{\bm{\pi}} \subseteq \mathcal{P}_{\bm{\pi}}^{\star}$ and therefore we have
$\bigcup_{\bm{\pi}\in \Pi} \mathcal{P}_{\bm{\pi}} \subseteq \mathcal{P}^{\star}$.
This inner bound of $\mathcal{P}^{\star}$ can be further tightened in general by following along the lines of \cite[Th. 5]{Geng2015}, i.e. taking the union of polyhedral TIN-achievable regions that correspond to all subnetworks of the original IMAC.

To facilitate the above, we define the more general collection of polyhedral TIN-achievable regions that
correspond to subnetworks of the IMAC.
For instance, consider a subnetwork comprising the subset of users $\mathcal{S} \subseteq \mathcal{K}$, where
$\overline{\mathcal{S}} \triangleq \mathcal{K} \setminus \mathcal{S}$ is the set  of all remaining users in the original IMAC.
We apply the polyhedral TIN scheme to the subnetwork $\mathcal{S}$ while deactivating all users in  $\overline{\mathcal{S}}$, i.e. by setting $r_{i}^{[l_{i}]} = - \infty$, $\forall (l_{i},i) \in \overline{\mathcal{S}}$, from which we obtain $d_{i}^{[l_{i}]} = 0$, $\forall (l_{i},i) \in \overline{\mathcal{S}}$.
The corresponding polyhedral TIN region for decoding order $\bm{\pi} \in \Pi$ is denoted by $\mathcal{P}_{\bm{\pi}}(\mathcal{S})$.
Note that the polyhedral TIN region described in \eqref{eq:polyhedral_region_1}--\eqref{eq:polyhedral_region_3}
is obtained by activating all users, i.e. $\mathcal{P}_{\bm{\pi}} = \mathcal{P}_{\bm{\pi}}(\mathcal{K})$.
On the other hand,  by deactivating all users we obtain $ \mathcal{P}_{\bm{\pi}}(\emptyset) =
\mathbf{0}$.

It is easily seen that $\mathcal{P}_{\bm{\pi}}(\mathcal{S}) \subseteq \mathcal{P}_{\bm{\pi}}^{\star}$,
$\forall \mathcal{S} \subseteq \mathcal{K}$, as the polyhedral TIN scheme over any subnetwork $\mathcal{S}$
is a special case of the original TIN scheme with the same decoding order.
By taking the union over all possible decoding orders, we establish an inner bound on $\mathcal{P}^{\star}$ given by
\begin{equation}
\label{eq:TIN_region_inner_1}
\mathcal{P}^{\star} \supseteq \bigcup_{\bm{\pi} \in \Pi} \bigcup_{\mathcal{S} \subseteq \mathcal{K}}
\mathcal{P}_{\bm{\pi}}(\mathcal{S}).
\end{equation}
By swapping the order of the union operators in the above inner bound, we reveal redundancies in its representation
as shown through the following remark.
\begin{remark}
\label{remark:inner_redundancies}
Consider a subset of users $\mathcal{\mathcal{S}} \subseteq \mathcal{K}$ and the corresponding family of
polyhedral TIN-achievable GDoF regions given by $\big\{ \mathcal{P}_{\bm{\pi}}(\mathcal{S}) : \bm{\pi} \in \Pi \big\}$.
Some decoding orders $\bm{\pi} \in \Pi$ are redundant, in the sense that they yield the same
polyhedral TIN regions, since varying the order of users in $\overline{\mathcal{S}}$, which are inactive, has
no influence on $\mathcal{P}_{\bm{\pi}}(\mathcal{S})$.
This type of redundancy is eliminated by considering the set of decoding orders for subnetwork $\mathcal{S}$ only,
which we denote by
$\Pi(\mathcal{S})$, and slightly modifying the definition of $\mathcal{P}_{\bm{\pi}}(\mathcal{S})$
into $\mathcal{P}_{\bm{\pi}'}(\mathcal{S})$, where $\bm{\pi}' \in \Pi(\mathcal{K})$, in which the order of users in $\overline{\mathcal{S}}$
is irrelevant\footnote{\label{footnote:decoding_order_subset}Suppose
that $\mathcal{S} = \cup_{i \in \mathcal{M}} \mathcal{S}_{i}$ for some $\mathcal{M} \subseteq \langle K \rangle$ and
$\mathcal{S}_{i} \subseteq \mathcal{K}_{i}$, $i \in \mathcal{M}$. Each decoding order $\bm{\pi}' \in \Pi(\mathcal{S})$
is given by $(\pi_{i}': i \in \mathcal{M})$, where $\pi_{i}' : \langle |\mathcal{S}_{i}| \rangle \rightarrow \mathcal{S}_{i}$
maps the order $s_{i} \in \langle |\mathcal{S}_{i}| \rangle$ to user $\pi_{i}'(s_{i}) \in \mathcal{S}_{i}$.
By definition, we have $\Pi = \Pi(\mathcal{K})$.}.
By employing these definitions, we can then easily show that \eqref{eq:TIN_region_inner_1} is equivalent to
\begin{equation}
\label{eq:TIN_region_inner_2}
\mathcal{P}^{\star} \supseteq \bigcup_{\mathcal{S} \subseteq \mathcal{K}}
\bigcup_{\bm{\pi}' \in \Pi(\mathcal{S})} \mathcal{P}_{\bm{\pi}'}(\mathcal{S}).
\end{equation}
\hfill $\lozenge$
\end{remark}
\section{Main Results and Insights}
\label{sec:main_results}
In this section, we present the primary results of this work with insights and illustrative examples. The proofs are
deferred to subsequent sections.
\subsection{Characterization of Polyhedral TIN-Achievable GDoF Regions}
We start by characterizing the polyhedral TIN-achievable GDoF region $\mathcal{P}_{\bm{\pi}}$ for any $\bm{\pi} \in \Pi$.
This polyhedral characterization is at the heart of all subsequent GDoF characterizations.
\begin{theorem}
\label{theorem:TIN_region}
For the IMAC described in Section \ref{sec:system model},
the achievable GDoF region through polyhedral TIN with decoding order $\bm{\pi} \in \Pi$, denoted by $\mathcal{P}_{\bm{\pi}}$,
is given by all tuples $\mathbf{d} \in \mathbb{R}_{+}^{|\mathcal{K}|}$ that satisfy
\begin{align}
\label{eq:IMAC_TIN_region_1}
\sum_{s_{i} = 1}^{l_{i}} d_{i}^{[\pi_{i}(s_{i})]} & \leq \alpha_{ii}^{[\pi_{i}(l_{i})]}, \ \forall (l_{i},i) \in \mathcal{K} \\
\nonumber
\sum_{j =1 }^{m} \sum_{s_{i_{j}} =1 }^{l_{i_{j}} } d_{i_{j}}^{[\pi_{i_{j}}(s_{i_{j}})]} & \leq
\sum_{j =1 }^{m} \alpha_{i_{j}i_{j}}^{[\pi_{i_{j}}(l_{i_{j}})]} - \alpha_{i_{j}i_{j-1}}^{[\pi_{i_{j}}(l_{i_{j}})]}, \\
\label{eq:IMAC_TIN_region_3}
 \forall l_{i_{j}} \in \langle L_{i_{j}}  \rangle, \  (i_{1},\ldots & ,i_{m}) \in \Sigma\big(\langle K \rangle\big),
 m \in \langle 2:K \rangle,
\end{align}
where  $\Sigma \big(\langle K \rangle \big)$ is the set of all possible cyclic sequences of all subsets\footnote{See the definition and example in the notation part.} of $\langle K \rangle $ and a modulo-$m$ operation is implicitly used on cell indices when dealing with cyclic sequences, i.e. $i_{0} = i_{m}$.
\end{theorem}
The characterization in the above theorem, which is given in terms of the channel strength levels only,
is obtained by eliminating the power control variables in \eqref{eq:polyhedral_region_1}--\eqref{eq:polyhedral_region_3}.
This elimination in turn is accomplished through the potential theorem \cite{Schrijver2002} and builds upon the arguments employed
in \cite{Geng2015,Geng2016}.
Full details of this procedure are presented in Section \ref{subsec:TIN_region_proof}.
From  the characterization in Theorem \ref{theorem:TIN_region}, it is evident that $\mathcal{P}_{\bm{\pi}}$ is
a polyhedron, which hence justifies the name of the polyhedral TIN scheme and the corresponding regions.
For the MAC special case, recovered by setting $K = 1$, this characterization
reduces to a MAC achievable GDoF region under a decoding order $\bm{\pi} \in \Pi$.
On the other hand, for the $K$-user  IC special case recovered when $L_{i} = 1$, $\forall i \in \langle K \rangle$,
the characterization reduces to the polyhedral TIN-achievable region in \cite[Th. 2]{Geng2015}.
In general, the characterization  in Theorem \ref{theorem:TIN_region} inherits features from both the
MAC and IC special cases which are further elaborated in the following remarks.
\begin{remark}
\label{remark:MAC_type_inequalities}
From the characterization of  $\mathcal{P}_{\bm{\pi}}$ in Theorem \ref{theorem:TIN_region}, it can be seen that
any GDoF inequality that includes $d_{i}^{[\pi_{i}(l_{i})]}$ also includes $d_{i}^{[\pi_{i}(l_{i}')]}$, for all $l_{i}' < l_{i}$.
This is due to the MAC-type successive decoding in which Rx-$i$ decodes
the signal of Tx-$\big(\pi_{i}(l_{i}),i\big)$
before decoding the signals of Tx-$\big(\pi_{i}(l_{i}'),i\big)$,
for all $l_{i}' < l_{i}$.
This in turn bounds the maximum achievable sum-GDoF of such users, i.e.
$\sum_{s_{i} \leq l_{i}}d_{i}^{[\pi_{i}(s_{i})]}$, by Tx-$\big(\pi_{i}(l_{i}),i\big)$'s  maximum achievable GDoF.
\hfill $\lozenge$
\end{remark}
\begin{remark}
\label{remark:IC_type_inequalities}
The cyclic feature exhibited in \eqref{eq:IMAC_TIN_region_3} is a product of the power-controlled TIN strategy
and has its roots in the regular IC \cite{Geng2015}.
From the observation in Remark \ref{remark:MAC_type_inequalities}, we may treat
the sum-GDoF
$\hat{d}_{i}^{[\pi_{i}(l_{i})]}  = \sum_{s_{i} \leq l_{i}}d_{i}^{[\pi_{i}(s_{i})]}$ as the GDoF of a single user Tx-$\big(\pi_{i}(l_{i}),i\big)$.
With this treatment in mind, consider a subnetwork of the IMAC which constitutes a
$K$-user IC. Such subnetwork must consist of one transmitter from each cell,
e.g. Tx-$\big(\pi_{i}(l_{i}),i\big)$ for all $i \in \langle K \rangle$.
From  \cite[Th. 1]{Geng2015}, the polyhedral TIN-achievable GDoF region of this $K$-user IC has the following cyclic inequalities
\begin{equation}
\nonumber
\sum_{j =1 }^{m}  \hat{d}_{i_{j}}^{[\pi_{i_{j}}(l_{i_{j}})]} \leq
\sum_{j =1 }^{m} \alpha_{i_{j}i_{j}}^{[\pi_{i_{j}}(l_{i_{j}})]} - \alpha_{i_{j}i_{j-1}}^{[\pi_{i_{j}}(l_{i_{j}})]}, \  (i_{1},\ldots  ,i_{m}) \in \Sigma\big(\langle K \rangle\big), m \in \langle 2:K \rangle
\end{equation}
which are included in \eqref{eq:IMAC_TIN_region_3}.
In fact, it can be seen that \eqref{eq:IMAC_TIN_region_3} consists of all cyclic GDoF inequalities resulting from the
polyhedral TIN regions of all possible subnetworks of the IMAC that constitute $K$-user ICs,
while retaining the above GDoF treatment of $\hat{d}_{i}^{[\pi_{i}(l_{i})]}  = \sum_{s_{i} \leq l_{i}}d_{i}^{[\pi_{i}(s_{i})]}$.
\hfill $\lozenge$
\end{remark}
Next, we turn our attention to the role of the decoding order $\bm{\pi}$.
First, it can be easily checked that the optimal GDoF region of the MAC special case is obtained by fixing the decoding order\footnote{Note that
this is in contrast to the MAC capacity region, which requires changing
the successive decoding order to achieve different corner points in general \cite{Cover2012}.
This difference is highlighted in \cite[Fig. 4]{Avestimehr2011} for the 2-user MAC through the linear deterministic
model, which shares many features with the GDoF model.}  to  $\bm{\pi} = \bm{\mathrm{id}}$.
In this case, the signal of a stronger user Tx-$(l_{i},i)$ is received by Rx-$i$ at a higher power level compared to
the signal of a weaker user Tx-$(l_{i}',i)$, $l_{i}' < l_{i}$.
Therefore, it is preferable from a GDoF perspective to decode the signal from Tx-$(l_{i},i)$ first while
treating all signals from Tx-$(l_{i}',i)$, $l_{i}' < l_{i}$,  as noise.
Contrary to the MAC special case however, the decoding order $\bm{\mathrm{id}}$
does not always yield the largest polyhedral TIN-achievable GDoF region for the IMAC,
i.e. $\mathcal{P}_{\bm{\pi}} \subseteq \mathcal{P}_{\bm{\mathrm{id}}}$ does not hold in general for all $\bm{\pi} \in \Pi$.
For example, a stronger user Tx-$(l_{i},i)$ may also have stronger cross links compared to a weaker
user Tx-$(l_{i}',i)$, $l_{i}' < l_{i}$, causing significantly more inter-cell interference.
Tx-$(l_{i},i)$ may be required to control its power to an extent that its signal is now received by
 Rx-$i$ at a lower power level compared to the signal
 of Tx-$(l_{i}',i)$. In this case, some GDoF points may only be
 achieved through a decoding order in which the signal of Tx-$(l_{i}',i)$ is decoded before that of Tx-$(l_{i},i)$.
To further illustrate the influence of $\bm{\pi}$, we consider the following  simple example.
\begin{example}
\label{example:3_user_order}
Consider a network of $K = 2$ cells, where cell $1$ and cell $2$ comprise
$L_{1} = 2$ and $L_{2} = 1$  users, respectively. This 2-cell, 3-user network, referred to as the
PIMAC in \cite{Gherekhloo2016}, is used as a running example throughout this section
as it captures some of the IMAC's main features and allows for GDoF regions that can be visualized.
According to Theorem \ref{theorem:TIN_region}, the region $\mathcal{P}_{\bm{\mathrm{id}}}$ for this network
is the set of all tuples $\big(d_{1}^{[1]},d_{1}^{[2]},d_{2}^{[1]}\big) \in \mathbb{R}_{+}^{3}$  that satisfy
\begin{align}
\label{eq:Poly_3_user_id_1}
d_{1}^{[1]}  & \leq  \alpha_{11}^{[1]} \\
d_{1}^{[2]} + d_{1}^{[1]}  & \leq  \alpha_{11}^{[2]} \\
d_{2}^{[1]}  & \leq  \alpha_{22}^{[1]} \\
d_{1}^{[1]} + d_{2}^{[1]}  & \leq  \alpha_{11}^{[1]} - \alpha_{12}^{[1]} + \alpha_{22}^{[1]} - \alpha_{21}^{[1]}  \\
\label{eq:Poly_3_user_id_last}
d_{1}^{[2]} + d_{1}^{[1]} + d_{2}^{[1]}  & \leq  \alpha_{11}^{[2]} - \alpha_{12}^{[2]}  + \alpha_{22}^{[1]} - \alpha_{21}^{[1]}.
\end{align}
In addition to the decoding order $\bm{\mathrm{id}}$, we have one more decoding order denoted by $\overline{\bm{\mathrm{id}}}$
for which Rx-$1$ decodes the signal of Tx-$(1,1)$ before decoding the signal of Tx-$(2,1)$.
From Theorem \ref{theorem:TIN_region}, the corresponding polyhedral region $\mathcal{P}_{\overline{\bm{\mathrm{id}}}}$ is the set of all $\big(d_{1}^{[1]},d_{1}^{[2]},d_{2}^{[1]}\big) \in \mathbb{R}_{+}^{3}$  satisfying
\begin{align}
\label{eq:Poly_3_user_id_bar_1}
d_{1}^{[2]}  & \leq  \alpha_{11}^{[2]} \\
d_{1}^{[1]} + d_{1}^{[2]}  & \leq  \alpha_{11}^{[1]} \\
d_{2}^{[1]}  & \leq  \alpha_{22}^{[1]} \\
d_{1}^{[2]} + d_{2}^{[1]}  & \leq  \alpha_{11}^{[2]} - \alpha_{12}^{[2]}  + \alpha_{22}^{[1]} - \alpha_{21}^{[1]}  \\
\label{eq:Poly_3_user_id_bar_last}
d_{1}^{[2]} + d_{1}^{[1]} + d_{2}^{[1]}  & \leq  \alpha_{11}^{[1]} - \alpha_{12}^{[1]} + \alpha_{22}^{[1]} - \alpha_{21}^{[1]}
\end{align}
where inequality \eqref{eq:Poly_3_user_id_bar_1} is clearly redundant.
Now let us assume that the following condition holds
\begin{align}
\label{eq:condition_example_order}
\alpha_{21}^{[1]} & \leq \alpha_{11}^{[2]} - \alpha_{12}^{[2]}  < \alpha_{11}^{[1]} - \alpha_{12}^{[1]}.
\end{align}
It can be easily verified that under the condition in \eqref{eq:condition_example_order}, the GDoF tuple given by
\begin{equation}
\label{eq:GDoF_tuple_example_order}
\big(d_{1}^{[1]},d_{1}^{[2]},d_{2}^{[1]}\big)  = \big( (\alpha_{11}^{[1]} - \alpha_{12}^{[1]}) - (\alpha_{11}^{[2]} - \alpha_{12}^{[2]} )  , (\alpha_{11}^{[2]} - \alpha_{12}^{[2]})  - \alpha_{21}^{[1]}  , \alpha_{22}^{[1]} \big)
\end{equation}
lies in the region $\mathcal{P}_{\overline{\bm{\mathrm{id}}}}$.
For this tuple, Tx-$(1,2)$ of cell $2$ achieves its full interference-free GDoF of $\alpha_{22}^{[1]}$,
and hence Tx-$(1,1)$ and Tx-$(2,1)$  of cell $1$ have to lower their transmit powers hence  limiting their
sum-GDoF to $(\alpha_{11}^{[1]} - \alpha_{12}^{[1]}) - \alpha_{21}^{[1]}$.
It can also be checked that the GDoF tuple in \eqref{eq:GDoF_tuple_example_order} is not in the  region $\mathcal{P}_{\bm{\mathrm{id}}}$,
as the inequality \eqref{eq:Poly_3_user_id_last} is violated under condition \eqref{eq:condition_example_order}.
In particular, for decoding order $\bm{\mathrm{id}}$, the sum-GDoF of cell $1$ is bounded by
$(\alpha_{11}^{[2]} - \alpha_{12}^{[2]}) - \alpha_{21}^{[1]} < (\alpha_{11}^{[1]} - \alpha_{12}^{[1]}) - \alpha_{21}^{[1]}$
when cell $2$ achieves its interference-free GDoF of $\alpha_{22}^{[1]}$.
An illustration of $\mathcal{P}_{\bm{\mathrm{id}}}$ and $\mathcal{P}_{\overline{\bm{\mathrm{id}}}}$
for an instance of the above network that satisfies \eqref{eq:condition_example_order}
is shown in Fig. \ref{fig:GDoF_regions}(c).
\hfill $\lozenge$
\end{example}
\begin{figure}[t]
\centering
\includegraphics[width = 1.0\textwidth,trim={0cm 0cm 0cm 5.9cm},clip]{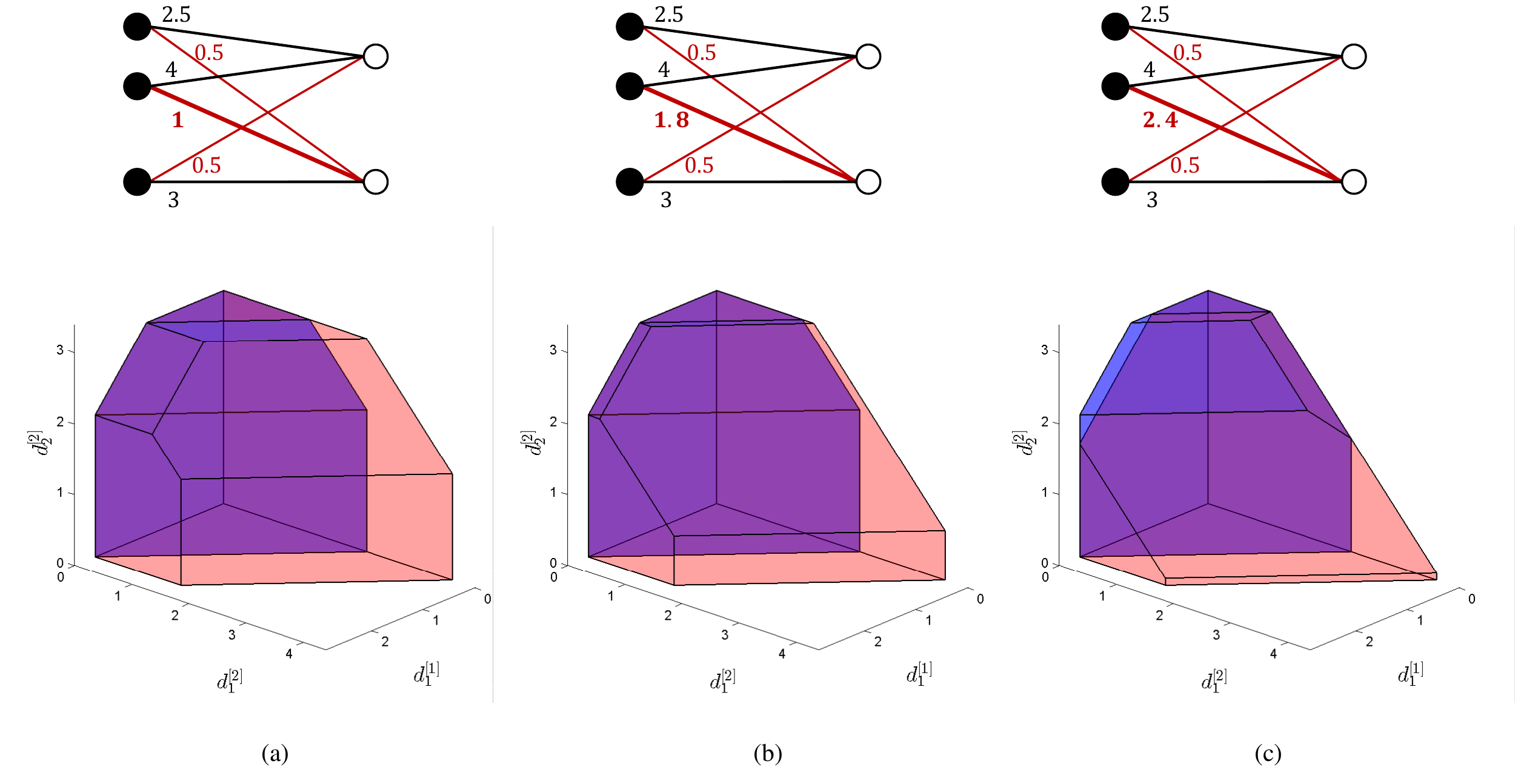}
\caption{\small
Polyhedral TIN-achievable GDoF regions for the 2-cell, 3-user network from the running example.
The regions $\mathcal{P}_{\bm{\mathrm{id}}}$ and
$\mathcal{P}_{\overline{\bm{\mathrm{id}}}}$ are illustrated in red and blue, respectively. For the instances in (a) and (b), we have
$\mathcal{P}_{\overline{\bm{\mathrm{id}}}} \subseteq \mathcal{P}_{\bm{\mathrm{id}}}$. For (c), we have
$\mathcal{P}_{\overline{\bm{\mathrm{id}}}} \nsubseteq \mathcal{P}_{\bm{\mathrm{id}}}$ and
$\mathcal{P}_{\bm{\mathrm{id}}}  \nsubseteq \mathcal{P}_{\overline{\bm{\mathrm{id}}}}$.}
\label{fig:GDoF_regions}
\end{figure}
The result in Theorem \ref{theorem:TIN_region} can be easily adapted to characterize the general polyhedral
TIN region for any subnetwork $\mathcal{S} \subseteq \mathcal{K}$ as shown in the following remark.
\begin{remark}
\label{remark:General_poly_TIN_region}
For any subnetwork $\mathcal{S} = \cup_{i \in \mathcal{M}} \mathcal{S}_{i}$, where $\mathcal{M} \subseteq \langle K \rangle$ and
$\mathcal{S}_{i} \subseteq \mathcal{K}_{i}$, $i \in \mathcal{M}$, the polyhedral TIN region
$\mathcal{P}_{\bm{\pi}'}(\mathcal{S})$, where $\bm{\pi}' \in \Pi(\mathcal{S})$, is described by
all tuples $\mathbf{d} \in \mathbb{R}_{+}^{|\mathcal{K}|}$ that satisfy
\begin{align}
d_{j}^{[l_{j}]}   &  = 0,   \ \forall (l_{j},j) \in \overline{\mathcal{S}} \\
\sum_{s_{i} = 1}^{ l_{i}}
d_{i}^{[\pi'_{i}(s_{i})]} & \leq \alpha_{ii}^{[\pi'_{i}(l_{i})]}, \
\forall l_{i}\in \langle |\mathcal{S}_{i}| \rangle, \; i \in \mathcal{M} \\
\nonumber
\sum_{j =1 }^{m} \sum_{s_{i_{j}} = 1}^{l_{i_{j}} } d_{i_{j}}^{[\pi'_{i_{j}}(s_{i_{j}})]} & \leq
\sum_{j =1 }^{m} \alpha_{i_{j}i_{j}}^{[\pi'_{i_{j}}(l_{i_{j}})]} - \alpha_{i_{j}i_{j-1}}^{[\pi'_{i_{j}}(l_{i_{j}})]}, \\
 \forall l_{i_{j}} \in \langle |\mathcal{S}_{i}| \rangle ,
 (i_{1} & ,\ldots  ,i_{m}) \in \Sigma\big(\mathcal{M}\big), m \in \langle 2:|\mathcal{M}| \rangle.
\end{align}
Note that the definitions of  $\bm{\pi}'$ and  $ \Pi(\mathcal{S})$
 are given in Remark \ref{remark:inner_redundancies} (see also footnote
 \ref{footnote:decoding_order_subset}).
\hfill $\lozenge$
\end{remark}
\subsection{Characterization of General TIN-Achievable GDoF Region}
Following the characterization of polyhedral TIN-achievable GDoF regions,
the natural question to ask next is whether we can characterize the general TIN-achievable GDoF region $\mathcal{P}^{\star}$.
This is settled in the following theorem which makes use of the results
in Theorem \ref{theorem:TIN_region} and Remark \ref{remark:General_poly_TIN_region}.
\begin{theorem}
\label{theorem:TIN_region_general}
For the IMAC described in Section \ref{sec:system model}, the general TIN-achievable region is equal to
\begin{equation}
\label{eg:general_TIN_region_thrm}
\mathcal{P}^{\star} = \bigcup_{\bm{\pi} \in \Pi} \bigcup_{\mathcal{S} \subseteq \mathcal{K}}
\mathcal{P}_{\bm{\pi}}(\mathcal{S}).
\end{equation}
\end{theorem}
The above theorem is proved by essentially showing that the inclusion in \eqref{eq:TIN_region_inner_1}
also holds in the opposite direction.
Full details are given in Section \ref{subsec:proofs_general_TIN_region_1}.

The general TIN-achievable region $\mathcal{P}^{\star}$, as seen from \eqref{eg:general_TIN_region_thrm},
is a finite union of polyhedra.
While the order of the two union operators in \eqref{eg:general_TIN_region_thrm} is set in this manner by construction
(see \eqref{eq:TIN_region_inner_1}), and also used in this fashion in the proof (see Section \ref{subsec:proofs_general_TIN_region_1}),
we may swap the order of the operators to eliminate redundancies as suggested by
Remark \ref{remark:inner_redundancies}.
It follows that \eqref{eg:general_TIN_region_thrm} is equivalent to
\begin{equation}
\label{eg:general_TIN_region_thrm_2}
\mathcal{P}^{\star} = \bigcup_{\mathcal{S} \subseteq \mathcal{K}}  \bigcup_{\bm{\pi}' \in \Pi(\mathcal{S})}
\mathcal{P}_{\bm{\pi}'}(\mathcal{S}).
\end{equation}
We observe that
there is a total of $2^{|\mathcal{K}|} - 1$ non-empty subnetworks of
$\mathcal{K}$ (including $\mathcal{K}$ itself)
and each such subnetwork may be expressed as $\mathcal{S} = \cup_{i \in \mathcal{K}} \mathcal{S}_{i} =
\cup_{i \in \mathcal{M}} \mathcal{S}_{i}$, where
$\mathcal{M} \subseteq \langle K \rangle$, $\mathcal{S}_{i} \subseteq \mathcal{K}_{i}$ and $\mathcal{S}_{i}
= \emptyset$ for all $i \in \mathcal{K} \setminus \mathcal{M}$.
Therefore, $\mathcal{S}$
admits $|\Pi(\mathcal{S})| = \prod_{i\in \mathcal{K}}(|\mathcal{S}_{i}|!)$ different decoding orders\footnote{Note that
we use the conventions $|\emptyset| = 0$ and $0! = 1$.}.
It follows from the representation in \eqref{eg:general_TIN_region_thrm_2}  that
$\mathcal{P}^{\star}$  is the union of
$\sum_{\mathcal{S}_{1} \subseteq \mathcal{K}_{1}} \! \!   \cdots \!  \sum_{\mathcal{S}_{K} \subseteq
\mathcal{K}_{K}}  \! \prod_{i\in \mathcal{K}}(|\mathcal{S}_{i}|!)$
polyhedral TIN-achievable regions in general.

From the above characterizations, we conclude that when time-sharing is not allowed, the GDoF region $\mathcal{P}^{\star}$, which is
achieved through power control and TIN, is not convex in general as it is given by a finite union of polyhedra.
This is further illustrated by revisiting our running example.
\begin{example}
Consider the 2-cell, 3-user network from Example \ref{example:3_user_order}.
For each of the instances of this network given in Fig. \ref{fig:GDoF_regions}, it can easily checked that the
polyhedral TIN-achievable GDoF regions for all subnetworks are included in
the 3-user polyhedral regions
$\mathcal{P}_{\bm{\mathrm{id}}}$ and $\mathcal{P}_{\overline{\bm{\mathrm{id}}}}$.
Therefore, it follows that
$\mathcal{P}^{\star}$ coincides with $\mathcal{P}_{\bm{\mathrm{id}}} \cup \mathcal{P}_{\overline{\bm{\mathrm{id}}}}$ for the examples in
Fig. \ref{fig:GDoF_regions}, from which we observe that $\mathcal{P}^{\star}$ is convex for the instances in (a) and (b),
and non-convex for the instance in (c).
\hfill $\lozenge$
\end{example}
The observation that $\mathcal{P}^{\star}$ is non-convex in general  is key in guiding the path towards
establishing conditions under which $\mathcal{P}^{\star}$ is optimal as we show in the following parts of this section.
\subsection{Conditions for TIN-Convexity}
\label{subsec:TIN_polyhedrality}
Capacity regions of synchronous channels, and therefore their GDoF counterparts, are known to be convex by virtue of
time-sharing \cite{ElGamal2011}.
Hence, for the general TIN-achievable GDoF region $\mathcal{P}^{\star}$ to be optimal, it must necessarily also be convex.
Here we identify conditions under which the latter holds,
i.e. where $\mathcal{P}^{\star}$ is convex in its own right without the need for
time-sharing.
This serves as a first step towards establishing conditions under which $\mathcal{P}^{\star}$ is also optimal.

From the characterization in \eqref{eg:general_TIN_region_thrm}, it can be seen that the polyhedrality (and hence convexity) of $\mathcal{P}^{\star}$ is guaranteed when at least one of the polyhedral TIN regions in the union contains all others, i.e.
if there exists $\bm{\pi}^{\star} \in \Pi$ and  $\mathcal{S}^{\star} \subseteq \mathcal{K}$ such that the following holds:
\begin{equation}
\mathcal{P}_{\bm{\pi}^{\star}}(\mathcal{S}^{\star}) \supseteq
\mathcal{P}_{\bm{\pi}}(\mathcal{S}), \ \forall \bm{\pi} \in \Pi \text{ and } \mathcal{S} \subseteq \mathcal{K}.
\end{equation}
In the following result, we identify conditions under which such $\bm{\pi}^{\star}$ and $\mathcal{S}^{\star}$ exist.
\begin{theorem}
\label{theorem:polyhedrality}
For the IMAC described in Section \ref{sec:system model}, if the following conditions are satisfied
\begin{align}
\label{eq:Poly_condition_2}
\alpha_{ii}^{[l_{i}]}  & \geq  \alpha_{ii}^{[l_{i}']} + \max_{j:j\neq i} \left\{ \alpha_{ij}^{[l_{i}]} - \alpha_{ij}^{[l_{i}']} \right\}, \
\forall i \in \langle K \rangle, \; l_{i}',l_{i} \in \langle L_{i} \rangle, \; l_{i}' < l_{i} \\
\label{eq:Poly_condition_1}
\alpha_{ii}^{[l_{i}]} & \geq \max_{j,(l_{k},k):j \neq i,k \neq i} \left\{\alpha_{ij}^{[l_{i}]} + \alpha_{ki}^{[l_{k}]} -
\alpha_{kj}^{[l_{k}]} \mathbbm{1}\big( k \neq j\big) \right\}, \ \forall(l_{i},i) \in \mathcal{K},
\end{align}
then the general TIN-achievable GDoF region $\mathcal{P}^{\star}$ is a polyhedron and it is given by
$\mathcal{P}^{\star} = \mathcal{P}_{\bm{\mathrm{id}}}$.
This region, achieved with a fixed decoding order $\bm{\pi} = \bm{\mathrm{id}}$, is
described by \eqref{eq:IMAC_TIN_region_1}  and \eqref{eq:IMAC_TIN_region_3} in Theorem \ref{theorem:TIN_region}
while setting $\pi_{i}(l_{i}) = l_{i}$ for all $(l_{i},i) \in \mathcal{K}$.
\end{theorem}
We refer to the conditions identified in Theorem \ref{theorem:polyhedrality} as the \emph{TIN-convexity conditions}.
To gain insight into these conditions, we first consider  \eqref{eq:Poly_condition_2} which is equivalently expressed as
\begin{equation}
\label{eq:Poly_condition_2_2}
\alpha_{ii}^{[l_{i}]} - \alpha_{ij}^{[l_{i}]}  \geq \alpha_{ii}^{[l_{i}']} - \alpha_{ij}^{[l_{i}']}, \
\forall i,j \in \langle K \rangle, \; i \neq j, \; l_{i}',l_{i} \in \langle L_{i} \rangle, \; l_{i}' < l_{i}.
\end{equation}
Now consider users Tx-$(l_{i},i)$ and Tx-$(l_{i}',i)$ from cell $i$
with the former being the stronger MAC user, i.e. $\alpha_{ii}^{[l_{i}]} \geq \alpha_{ii}^{[l_{i}']}$.
Moreover, we focus on the interference caused by these two users to some cell $j$.
The condition in  \eqref{eq:Poly_condition_2_2} implies that even after attenuating the powers of Tx-$(l_{i},i)$ and Tx-$(l_{i}',i)$
such that they cause no interference to cell $j$ above noise level, i.e.  $r_{i}^{[l_{i}]} = - \alpha_{ij}^{[l_{i}]}$ and
$r_{i}^{[l_{i}']} = - \alpha_{ij}^{[l_{i}']}$,
Tx-$(l_{i},i)$ remains stronger compared to Tx-$(l_{i}',i)$ in the sense that its signal is still received by Rx-$i$ at a higher power level.
This extends to all users such that the MAC order of users in each cell is preserved under the constraint of reducing inter-cell interference
caused to any subset of cells to noise level.
As a wider implication, we see through the proof of Theorem \ref{theorem:polyhedrality}
in Section \ref{sec:proof_polyhedrality} that the condition in \eqref{eq:Poly_condition_2_2}
is sufficient to guarantee that
$\bm{\mathrm{id}}$ is the dominant order, i.e. for any subnetwork $\mathcal{S} \subseteq \mathcal{K}$, we have
$\mathcal{P}_{\bm{\pi}}(\mathcal{S}) \subseteq \mathcal{P}_{\bm{\mathrm{id}}}(\mathcal{S})$ for all
$\bm{\pi} \in \Pi$.

In addition to order preservation within each MAC, the following step in
establishing Theorem \ref{theorem:polyhedrality} is to show that
$\mathcal{P}_{\bm{\mathrm{id}}}(\mathcal{S}) \subseteq \mathcal{P}_{\bm{\mathrm{id}}}(\mathcal{K})$ holds
for all subnetworks $\mathcal{S} \subseteq \mathcal{K}$.
To this end, we note that \eqref{eq:Poly_condition_1} is essentially the TIN-convexity condition for the $K$-user IC,
identified by Yi and Caire in \cite[Th. 4]{Yi2016}, applied to all possible $K$-user IC subnetworks of the considered IMAC.
This condition in conjunction with the one in \eqref{eq:Poly_condition_2} are sufficient to guarantee a monotonic behaviour
of $\mathcal{P}_{\bm{\mathrm{id}}}(\mathcal{S})$ in $\mathcal{S}$, i.e.
$\mathcal{S}' \subseteq \mathcal{S} \subseteq \mathcal{K}$ implies
$\mathcal{P}_{\bm{\mathrm{id}}}(\mathcal{S}') \subseteq \mathcal{P}_{\bm{\mathrm{id}}}(\mathcal{S})
\subseteq \mathcal{P}_{\bm{\mathrm{id}}}(\mathcal{K})$.
Full details are relegated to  Section \ref{sec:proof_polyhedrality}.
\begin{example}
\label{example:3_polyhedrality}
Consider the 2-cell, 3-user network from the running example.
The TIN-convexity  condition in \eqref{eq:Poly_condition_2} of Theorem \ref{theorem:polyhedrality}
is expressed for this network as
\begin{equation}
\label{eq:Poly_condition_2_3_user_example}
\alpha_{11}^{[2]} - \alpha_{11}^{[1]} \geq \alpha_{12}^{[2]} - \alpha_{12}^{[1]}.
\end{equation}
On the other hand, the IC-type TIN-convexity condition in  \eqref{eq:Poly_condition_1}
is given for this network by the following set of inequalities\footnote{Note that this is equivalent
to the TIN-optimality condition of Geng \emph{et al.} \cite{Geng2015} applied
to each of the 2-user IC subnetworks of the IMAC in Fig. \ref{fig:GDoF_regions}.
This holds since the IC-type TIN-optimality and TIN-convexity conditions are
identical
in 2-cell networks (see Theorem \ref{theorem:TIN_optimality} and Example \ref{example:3_user_TIN} in the following part).}:
\begin{align}
\label{eq:TIN_condition_3_user_example_1_1}
\alpha_{11}^{[1]} & \geq \alpha_{12}^{[1]} + \alpha_{21}^{[1]} \\
\alpha_{22}^{[1]}       & \geq \alpha_{12}^{[1]} + \alpha_{21}^{[1]} \\
\alpha_{11}^{[2]} & \geq \alpha_{12}^{[2]} + \alpha_{21}^{[1]} \\
\label{eq:TIN_condition_3_user_example_1_4}
\alpha_{22}^{[1]}       & \geq \alpha_{12}^{[2]} + \alpha_{21}^{[1]}.
\end{align}
It can be verified that the instances of this network given in (a) and (b) of Fig. \ref{fig:GDoF_regions}
satisfy the above TIN-convexity conditions. This in turn leads to
$\mathcal{P}_{\overline{\bm{\mathrm{id}}}} \subseteq \mathcal{P}_{\bm{\mathrm{id}}}$
and $\mathcal{P}^{\star}$ being a polyhedron, and hence convex, as seen in the illustrations.
On the other hand, the instance given in (c) of Fig. \ref{fig:GDoF_regions}
violates these conditions and has a region $\mathcal{P}^{\star}$ which is non-convex.
This example, however, is far from enough for proving that the
TIN-convexity conditions identified in Theorem \ref{theorem:polyhedrality} are also necessarily for
the convexity of $\mathcal{P}^{\star}$.
This issue of necessity and sufficiency of TIN conditions in this work, and in the related literature,
is revisited in Remark \ref{remark:sufficiency_of_conditions} presented at the end of the section.
\hfill $\lozenge$
\end{example}
Knowing that the convexity of an achievable GDoF region is a necessary condition for it to be optimal,
the main question that comes to mind at this point is whether the TIN-convexity conditions identified in
Theorem \ref{theorem:polyhedrality}, under which the TIN region $\mathcal{P}^{\star}$ is convex,
also imply the optimality of  $\mathcal{P}^{\star}$.
This issue is further explored in Remark \ref{remark:TIN_and_IA}, after presenting the final main result of this work next.
\subsection{Conditions for TIN-Optimality}
\label{subsec:TIN_optimality}
In the following theorem, we obtain \emph{TIN-optimality conditions} under which the TIN scheme proposed in
Section \ref{subsec:TIN_scheme}, with power control, successive decoding and no time-sharing,
achieves the entire GDoF region of the IMAC which we denote by $\mathcal{D}$.
\begin{theorem}
\label{theorem:TIN_optimality}
For the IMAC described in Section \ref{sec:system model}, if the following conditions are satisfied
\begin{align}
\label{eq:TIN_condition_2}
\alpha_{ii}^{[l_{i}]}  & \geq \alpha_{ii}^{[l_{i}']}  + \max_{j:j\neq i} \left\{ \min\left\{\alpha_{ij}^{[l_{i}]},2\alpha_{ij}^{[l_{i}]} - \alpha_{ij}^{[l_{i}']} \right\} \right\}, \
\forall i \in \langle K \rangle, \; l_{i}',l_{i} \in \langle L_{i} \rangle, \; l_{i}' < l_{i} \\
\label{eq:TIN_condition_1}
\alpha_{ii}^{[l_{i}]} & \geq \max_{j:j \neq i} \left\{\alpha_{ij}^{[l_{i}]} \right\} +
\max_{(l_{k},k):k \neq i} \left\{\alpha_{ki}^{[l_{k}]} \right\}, \ \forall(l_{i},i)\in
\mathcal{K},
\end{align}
then the optimal GDoF region is given by $\mathcal{D} = \mathcal{P}^{\star} = \mathcal{P}_{\bm{\mathrm{id}}}$.
This region, achieved with a fixed decoding order $\bm{\pi} = \bm{\mathrm{id}}$, is
described by \eqref{eq:IMAC_TIN_region_1} and \eqref{eq:IMAC_TIN_region_3} in Theorem \ref{theorem:TIN_region}
while setting $\pi_{i}(l_{i}) = l_{i}$ for all $(l_{i},i) \in \mathcal{K}$.
\end{theorem}
In the proof of Theorem \ref{theorem:TIN_optimality}, we derive an outer bound for the capacity region $\mathcal{C}$
under the assumption that the conditions in \eqref{eq:TIN_condition_2} and \eqref{eq:TIN_condition_1} hold (see Theorem \ref{theorem:outer_bound} in
Section \ref{sec:TIN-Optimality}).
It turns out that the corresponding GDoF region outer bound, obtained from the capacity region outer bound,
coincides with the polyhedral TIN-achievable region $\mathcal{P}_{\bm{\mathrm{id}}}$
when \eqref{eq:TIN_condition_2} and \eqref{eq:TIN_condition_1}  hold, from which optimality is established.
The full details of the proof are given in Section \ref{sec:TIN-Optimality}.

We turn our attention now to understanding the TIN-optimality conditions in Theorem \ref{theorem:TIN_optimality}.
It can be seen that the condition \eqref{eq:TIN_condition_2} is equivalently expressed by the following inequalities:
\begin{equation}
\label{eq:TIN_condition_2_2}
\begin{aligned}
& \alpha_{ii}^{[l_{i}]} - \alpha_{ij}^{[l_{i}]}  \geq \alpha_{ii}^{[l_{i}']}  \ \ \text{or}  \\
& \alpha_{ii}^{[l_{i}]} - \alpha_{ij}^{[l_{i}]}  \geq \alpha_{ii}^{[l_{i}']} - \alpha_{ij}^{[l_{i}']}  + \alpha_{ij}^{[l_{i}]}
\end{aligned}
 \ \ , \ \forall i,j \in \langle K \rangle, \; i \neq j, \; l_{i}',l_{i} \in \langle L_{i} \rangle, \; l_{i}' < l_{i}.
\end{equation}
The TIN-optimality condition in \eqref{eq:TIN_condition_2} is reminiscent of the
TIN-convexity condition in \eqref{eq:Poly_condition_2} in the sense that
it provisions the power level gains of stronger users against weaker users in each MAC.
On the other hand, the IC-type TIN-optimality condition in \eqref{eq:TIN_condition_1} is the condition identified by
Geng \emph{et al.} in \cite{Geng2015}, applied to all possible $K$-user IC subnetworks of the IMAC.
By comparing \eqref{eq:Poly_condition_2} and \eqref{eq:TIN_condition_2}
(for instance through \eqref{eq:Poly_condition_2_2} and \eqref{eq:TIN_condition_2_2} respectively), we
can see that the latter is stricter than the former.
Moreover, we know from \cite[Rem. 4]{Yi2016} that \eqref{eq:TIN_condition_1} is a stricter version of \eqref{eq:Poly_condition_1}.
These observations lead to the following remark.
\begin{remark}
\label{remark:TIN_poly_includes_TIN_opt}
The TIN-convexity conditions in Theorem \ref{theorem:polyhedrality} are a
relaxed version of the TIN-optimality conditions in Theorem \ref{theorem:TIN_optimality}. Therefore, if the TIN-optimality conditions
\eqref{eq:TIN_condition_2} and \eqref{eq:TIN_condition_1} hold then the TIN-convexity conditions \eqref{eq:Poly_condition_2} and
\eqref{eq:Poly_condition_1} are automatically satisfied.
\hfill $\lozenge$
\end{remark}
\begin{figure}
\centering
\includegraphics[width = 0.6\textwidth]{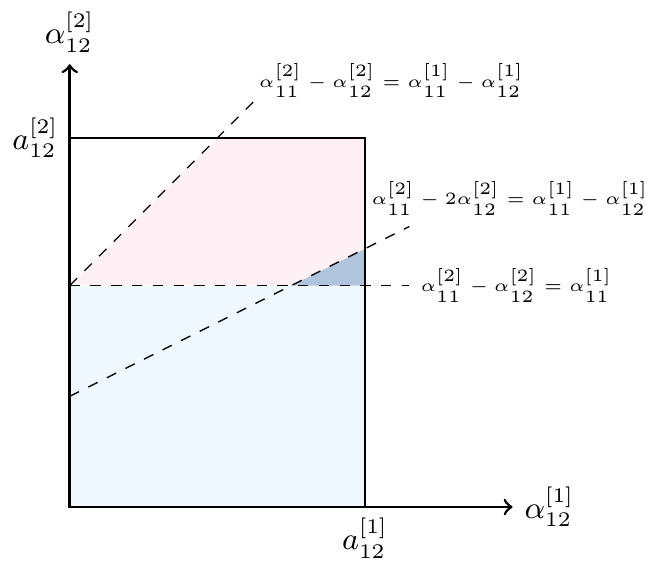}
\caption{\small
TIN-optimality and TIN-convexity regimes for the 2-cell, 3-user network of Fig. \ref{fig:GDoF_regions}
in terms of cross link strengths $\big(\alpha_{12}^{[1]},\alpha_{12}^{[2]}\big)$, as explained in Example \ref{example:3_user_TIN}.
The constants $a_{12}^{[1]}$ and $a_{12}^{[2]}$ are given by
$a_{12}^{[l]} =  \min\big\{\alpha_{22}^{[1]},\alpha_{11}^{[l]}\big\} - \alpha_{21}^{[1]}$, $l \in \{1,2\}$.
Different regimes are highlighted as follows: $\mathcal{A}_{\mathrm{o}}'$ in light blue, $\mathcal{A}_{\mathrm{o}}''\setminus\mathcal{A}_{\mathrm{o}}'$ in dark blue,
$\mathcal{A}_{\mathrm{p}}\setminus\mathcal{A}_{\mathrm{o}}$ in light red and $\mathcal{A}$ is the entire box.}
\label{fig:TIN_conditions}
\end{figure}
To gain more insights, we revisit our running example in the light of the newly established conditions.
This is followed by further remarks and observations.
\begin{example}
\label{example:3_user_TIN}
For the 2-cell, 3-user network in previous examples, the TIN-optimality condition in \eqref{eq:TIN_condition_2}
is given by the inequalities:
\begin{subequations}
\label{eq:TIN_condition_3_user_example_2}
\begin{align}
\label{eq:TIN_condition_3_user_example_2_a}
\alpha_{11}^{[2]} - \alpha_{11}^{[1]}  & \geq \alpha_{12}^{[2]} \ \text{or} \\
\label{eq:TIN_condition_3_user_example_2_b}
\alpha_{11}^{[2]} - \alpha_{11}^{[1]}  & \geq 2\alpha_{12}^{[2]}  - \alpha_{12}^{[1]}.
\end{align}
\end{subequations}
Moreover, it can be easily checked that the IC-type TIN-optimality condition in \eqref{eq:TIN_condition_1} for this network is equivalent
to the IC-type TIN-convexity condition in \eqref{eq:TIN_condition_3_user_example_1_1}--\eqref{eq:TIN_condition_3_user_example_1_4}
of Example \ref{example:3_polyhedrality}.

Next, we look at the regimes of channel strength levels described by the conditions in Theorem \ref{theorem:TIN_optimality} and
Theorem \ref{theorem:polyhedrality}.
To facilitate this, we fix all direct
link strength levels $ \alpha_{11}^{[1]} $,  $\alpha_{11}^{[2]}$ and
$\alpha_{22}^{[1]} $ and the cross link strength level $\alpha_{21}^{[1]}$
(i.e. interference caused to cell $1$), such that
$ \alpha_{11}^{[1]} $,  $\alpha_{11}^{[2]} , \alpha_{22}^{[1]} > \alpha_{21}^{[1]}$ is satisfied.
We consider the influence of varying the cross link strengths
$\alpha_{12}^{[1]}$ and $\alpha_{12}^{[2]}$ (i.e. interference caused to cell $2$),
while assuming that the IC-type TIN-optimality conditions in \eqref{eq:TIN_condition_3_user_example_1_1}--\eqref{eq:TIN_condition_3_user_example_1_4} hold.
It can be seen that \eqref{eq:TIN_condition_3_user_example_1_1}--\eqref{eq:TIN_condition_3_user_example_1_4}
confine the set of allowed strengths $\big( \alpha_{12}^{[1]},\alpha_{12}^{[2]} \big) \in \mathbb{R}^{2}_{+}$ to  the box
given by
\begin{align}
 0 \leq \alpha_{12}^{[1]} \leq \min \big\{ \alpha_{22}^{[1]}, \alpha_{11}^{[1]}   \big\}
 - \alpha_{21}^{[1]} \\
  0 \leq \alpha_{12}^{[2]} \leq \min \big\{ \alpha_{22}^{[1]}, \alpha_{11}^{[2]}   \big\}
 - \alpha_{21}^{[1]}
\end{align}
which we denote by $\mathcal{A}$ (see Fig. \ref{fig:TIN_conditions}).
We further define the following sub-regimes of $\mathcal{A}$:
\begin{itemize}
\item $\mathcal{A}_{\mathrm{o}}'$ and $\mathcal{A}_{\mathrm{o}}''$ are given by
the intersection of $\mathcal{A}$ with
\eqref{eq:TIN_condition_3_user_example_2_a} and \eqref{eq:TIN_condition_3_user_example_2_b} respectively.
\item $\mathcal{A}_{\mathrm{p}}$  is given by the intersection of $\mathcal{A}$ with \eqref{eq:Poly_condition_2_3_user_example}.
\end{itemize}
The above sub-regimes are all illustrated in Fig. \ref{fig:TIN_conditions}.
It is readily seen that $\mathcal{A}_{\mathrm{o}} = \mathcal{A}_{\mathrm{o}}' \cup \mathcal{A}_{\mathrm{o}}''$ is the TIN-optimality regime
identified in Theorem \ref{theorem:TIN_optimality}, while $\mathcal{A}_{\mathrm{p}}$ is the TIN-convexity regime identified in
Theorem \ref{theorem:polyhedrality}.
Furthermore, it can be easily verified that the instances of the 2-cell, 3-user network in Fig. \ref{fig:GDoF_regions}(a), (b) and (c)
are in the regimes $\mathcal{A}_{\mathrm{o}}$, $\mathcal{A}_{\mathrm{p}} \setminus \mathcal{A}_{\mathrm{o}}$
and $\mathcal{A} \setminus \mathcal{A}_{\mathrm{p}}$ respectively.
\hfill $\lozenge$
\end{example}
Beyond the 2-cell, 3-user network considered above, to gain further insights into
the broadness of the TIN-convexity and TIN-optimality regimes in cellular settings with more cells and users, we resort to
numerical simulations. These results are presented in Appendix \ref{appendix:numerical_evaluations}.

Next, in the light of Example \ref{example:3_user_TIN}, we explore the relationship between the TIN-optimality regime in Theorem \ref{theorem:TIN_optimality} and the regime identified in \cite{Gherekhloo2016} for the 2-cell, 3-user network (i.e. PIMAC).
\begin{remark}
To make the connection between Example \ref{example:3_user_TIN} and the results in \cite{Gherekhloo2016} more apparent,
we express the regime $\mathcal{A}$ in terms of the notation and sub-regimes in \cite{Gherekhloo2016}.
In particular, $\mathcal{A}$ here corresponds to the union of sub-regimes (2A), (2B), (2C), (3C) and
 $(\alpha_{d3} - \alpha_{c3} = \alpha_{d1} - \alpha_{c1})$ in \cite{Gherekhloo2016}, while imposing an additional
  order constraint of $\alpha_{d3} \geq \alpha_{d1}$ (see  \cite[Fig. 8]{Gherekhloo2016}).
It follows that the TIN-optimality regime here, i.e. $\mathcal{A}_{\mathrm{o}}$, corresponds to the union of (2A), (2B) and part of (2C).

Through a direct comparison, it is evident that we arrive at a smaller
TIN-optimality regime compared to the one in \cite{Gherekhloo2016}.
This is not surprising, since we consider the entire GDoF region as opposed to only the sum-GDoF
considered in \cite{Gherekhloo2016}.
From the more restrictive GDoF region perspective,
the TIN-optimal regime specified here requires each transmitter-receiver pair to
satisfy the IC-type TIN optimality conditions, i.e.
\eqref{eq:TIN_condition_3_user_example_1_1}--\eqref{eq:TIN_condition_3_user_example_1_4},
known to be necessary for the 2-user IC and conjectured to be necessary for the $K$-user IC\footnote{Except for a set of channel
gain values of measure zero.} \cite{Geng2015}.
On the other hand, the sum-GDoF TIN-optimal regime in
\cite{Gherekhloo2016} allows for some of the  IC-type TIN conditions  in
\eqref{eq:TIN_condition_3_user_example_1_1}--\eqref{eq:TIN_condition_3_user_example_1_4}
to be violated.
For example, in parts of sub-regime (1B), the weaker MAC user, i.e. Tx-$(1,1)$,
may be causing significant interference to Rx-$2$ such that $ \alpha_{12}^{[1]} > \alpha_{22}^{[1]}  - \alpha_{21}^{[1]}$
(i.e. $ \alpha_{c3} > \alpha_{d2}  - \alpha_{c2}$ ), yet TIN is still sum-GDoF optimal.
In this scenario, which is not in $\mathcal{A}_{\mathrm{o}}$ or $\mathcal{A}$,
the optimal sum-GDoF is attained by switching off Tx-$(1,1)$, hence operating the network as a  TIN-optimal
 2-user IC. This luxury of excluding \emph{bad} transmitters cannot be afforded when
considering the entire GDoF region.
\hfill $\lozenge$
\end{remark}
\begin{remark}
\label{remark:TIN_and_IA}
For instances of the 2-cell, 3-user network that fall within the regime $\mathcal{A}$ identified in
Example \ref{example:3_user_TIN},
the sum-GDoF achieved through the proposed TIN scheme is bounded above as
\begin{equation}
\label{eq:sum_GDoF_upperbound_TIN_2_cell}
d_{1}^{[1]} + d_{1}^{[2]} + d_{2}^{[1]} \leq \max\{\alpha_{11}^{[1]} - \alpha_{12}^{[1]},\alpha_{11}^{[2]} - \alpha_{12}^{[2]} \}
+ \alpha_{22}^{[1]} - \alpha_{21}^{[1]}, \; \forall \big(d_{1}^{[1]},d_{1}^{[2]} ,d_{2}^{[1]}\big) \in \mathcal{P}^{\star}.
\end{equation}
This holds since \eqref{eq:TIN_condition_3_user_example_1_1}--\eqref{eq:TIN_condition_3_user_example_1_4}
and $\mathcal{P}^{\star} = \mathcal{P}_{\bm{\mathrm{id}}} \cup \mathcal{P}_{\overline{\bm{\mathrm{id}}}}$  hold throughout
$\mathcal{A}$.
Interestingly, it has been demonstrated by Gherekhloo \emph{et al.} in \cite{Gherekhloo2016} that for
the sub-regime $\mathcal{A} \setminus \mathcal{A}_{\mathrm{o}}$,
the TIN-achievable sum-GDoF upper bound in  \eqref{eq:sum_GDoF_upperbound_TIN_2_cell} can be strictly surpassed, almost surely,
using schemes that employ interference alignment with common and private signalling\footnote{In the notation and sub-regimes
of \cite{Gherekhloo2016},  $\mathcal{A} \setminus \mathcal{A}_{\mathrm{o}}$ defined
here corresponds to the intersection of $\alpha_{d3} \geq \alpha_{d1}$ with the union of sub-regimes
(3C) and  ($\alpha_{d3} - \alpha_{c3} = \alpha_{d1} - \alpha_{c1}$). It is noteworthy that for
 $\alpha_{d3} - \alpha_{c3} = \alpha_{d1} - \alpha_{c1}$
 (i.e. $\alpha_{11}^{[2]} - \alpha_{12}^{[2]} = \alpha_{11}^{[1]} - \alpha_{12}^{[1]}$ here), the
strict superiority of interference alignment holds except
for a subset of channel coefficients of measure zero.
For details, readers are referred to \cite[Corollaries 6 and 7]{Gherekhloo2016} and their proofs.
}.
Since $\mathcal{A}_{\mathrm{p}} \setminus \mathcal{A}_{\mathrm{o}}$ is contained both in the TIN-convexity regime and in
$\mathcal{A} \setminus \mathcal{A}_{\mathrm{o}}$, the above observation confirms that
the convexity of the TIN region $\mathcal{P}^{\star}$
does not necessarily imply its optimality.
\hfill $\lozenge$
\end{remark}
Finally, we conclude this section with the two  further general remarks.
\begin{remark}
\label{remark:sufficiency_of_conditions}
As pointed out in \cite[Rem. 1]{Geng2016}, whether we look through the lens of the GDoF or the exact capacity, existing TIN-optimality results are
\emph{``primarily in the form of sufficient conditions''} and the necessity of such conditions
\emph{``remains undetermined in most cases''}.
The TIN-optimality result in Theorem \ref{theorem:TIN_optimality} is no exception to most existing results in that regards.
Similarly, the TIN-convexity conditions in Theorem \ref{theorem:polyhedrality} are also sufficient and there is no claim of necessity.
\hfill $\lozenge$
\end{remark}
\begin{remark}
\label{remark:constant_gap}
Assuming that the TIN-optimality conditions in Theorem \ref{theorem:TIN_optimality} hold, then it is not difficult to show that
the TIN scheme proposed in this paper achieves the whole capacity region of the IMAC to within a constant gap of
$\Delta + \log( | \mathcal{K}|)$ bits at any finite SNR, where $\Delta > 0$ is fixed.
This can be shown using the capacity outer bound obtained in
Theorem \ref{theorem:outer_bound} of Section \ref{sec:TIN-Optimality}
in conjunction with the rate bounding techniques in \cite{Geng2015, Geng2015a}.
Moreover, the constant $\Delta$  can be explicitly calculated, e.g. see  \cite[Th. 4]{Joudeh2018}
where $\Delta$ is characterized for an IMAC with $K$ cells and $L_{1}, \ldots, L_{K} = 2$ users per cell.
This calculation can be easily extended to arbitrary numbers of users in different cells.
\hfill $\lozenge$
\end{remark}
\section{Proofs of Achievability}
In this section, we provide proofs for the achievability results, i.e. Theorem \ref{theorem:TIN_region}
and Theorem \ref{theorem:TIN_region_general}.
\subsection{Proof of Theorem \ref{theorem:TIN_region}}
\label{subsec:TIN_region_proof}
We prove Theorem \ref{theorem:TIN_region} by constructing a potential graph \cite{Geng2015,Geng2016} for the considered IMAC
and invoking the potential theorem \cite{Schrijver2002}.
To avoid cumbersome notation, we work with $\mathcal{P}_{\bm{\mathrm{id}}}$.
All derivations extend to $\mathcal{P}_{\bm{\pi}}$ by replacing each superscript $l_{k}$ with the corresponding
$\pi_{k}(l_{k})$.

The first step towards applying the potential theorem is to derive the conditions of feasible power allocation.
To this end, we rewrite \eqref{eq:polyhedral_region_3} as
\begin{equation}
\label{eq:GDoF_polyhedral_2}
d_{k}^{[l_{k}]} \leq  \min\biggl\{ r_{k}^{[l_{k}]} + \alpha_{kk}^{[l_{k}]},
\min_{j \neq k} \min_{l_{j}} \bigl\{
r_{k}^{[l_{k}]} - r_{j}^{[l_{j}]} + \alpha_{kk}^{[l_{k}]} - \alpha_{jk}^{[l_{j}]}
\bigr\} ,
\min_{l_{k}' < l_{k}} \bigl\{ r_{k}^{[l_{k}]} - r_{k}^{[l_{k}']} + \alpha_{kk}^{[l_{k}]} - \alpha_{kk}^{[l_{k}']} \bigl\}   \biggr\}
\end{equation}
where the three terms inside the outmost minimization incorporate no
interference, inter-cell interference and intra-cell interference, respectively.
From  \eqref{eq:GDoF_polyhedral_2}, it follows that the polyhedral TIN region $\mathcal{P}_{\bm{\mathrm{id}}}$, described by the inequalities in \eqref{eq:polyhedral_region_1}--\eqref{eq:polyhedral_region_3} while setting $\bm{\pi} = \bm{\mathrm{id}}$, is equivalently described by the following inequalities
\begin{align}
\label{eq:polyhedral_region_21}
r_{k}^{[l_{k}]} & \leq 0,  \  \forall (l_{k},k) \in \mathcal{K}  \\
\label{eq:polyhedral_region_22}
d_{k}^{[l_{k}]} & \geq 0,  \  \forall (l_{k},k) \in \mathcal{K}  \\
\label{eq:polyhedral_region_23}
d_{k}^{[l_{k}]} & \leq \alpha_{kk}^{[l_{k}]} + r_{k}^{[l_{k}]},
\ \forall (l_{k},k) \in \mathcal{K}  \\
\label{eq:polyhedral_region_24}
d_{k}^{[l_{k}]} & \leq  r_{k}^{[l_{k}]}  -  r_{j}^{[l_{j}]} + \alpha_{kk}^{[l_{k}]} -  \alpha_{jk}^{[l_{j}]} , \;
\forall (l_{k},k), (l_{j},j) \in \mathcal{K}, \ j \neq k \\
\label{eq:polyhedral_region_25}
d_{k}^{[l_{k}]} & \leq r_{k}^{[l_{k}]} -   r_{k}^{[l_{k}']} + \alpha_{kk}^{[l_{k}]} - \alpha_{kk}^{[l_{k}']}, \; \forall (l_{k},k) \in \mathcal{K},\; l_{k}' \in \langle L_{k} \rangle, \; l_{k}' < l_{k}.
\end{align}
After rearranging, the inequalities in \eqref{eq:polyhedral_region_21}--\eqref{eq:polyhedral_region_25} are rewritten as
\begin{align}
\label{eq:polyhedral_region_31}
d_{k}^{[l_{k}]} & \geq 0,  \  \forall (l_{k},k) \in \mathcal{K}  \\
\label{eq:polyhedral_region_32}
r_{k}^{[l_{k}]} & \leq 0,  \  \forall (l_{k},k) \in \mathcal{K}  \\
\label{eq:polyhedral_region_33}
-r_{k}^{[l_{k}]} & \leq \alpha_{kk}^{[l_{k}]} - d_{k}^{[l_{k}]},
\ \forall (l_{k},k) \in \mathcal{K}  \\
\label{eq:polyhedral_region_34}
r_{j}^{[l_{j}]} - r_{k}^{[l_{k}]}  & \leq \alpha_{kk}^{[l_{k}]} -  \alpha_{jk}^{[l_{j}]} - d_{k}^{[l_{k}]}, \;
\forall (l_{k},k), (l_{j},j) \in \mathcal{K}, \ j \neq k \\
\label{eq:polyhedral_region_35}
r_{k}^{[l_{k}']}   - r_{k}^{[l_{k}]} & \leq \alpha_{kk}^{[l_{k}]} - \alpha_{kk}^{[l_{k}']} - d_{k}^{[l_{k}]},
\; \forall (l_{k},k) \in \mathcal{K},\; l_{k}' \in \langle L_{k} \rangle, \; l_{k}' < l_{k}.
\end{align}
Hence, a GDoF tuple $\mathbf{d} \in \mathbb{R}_{+}^{|\mathcal{K}|}$ is in the polyhedral TIN region
$\mathcal{P}_{\bm{\mathrm{id}}}$ if and only if there exists a power allocation tuple $\mathbf{r}\in \mathbb{R}^{|\mathcal{K}|}$ such that \eqref{eq:polyhedral_region_32}--\eqref{eq:polyhedral_region_35} hold.
\subsubsection{Potential Graph and Potential Theorem}
We construct a directed graph (digraph) $\mathcal{G}_{\mathrm{p}} = (\mathcal{V}, \mathcal{E} )$
with vertices and directed edges given by
\begin{align}
\label{eq:potential_graph_V}
\mathcal{V} & = \left\{ v_{0}^{[0]} \right\} \cup \left\{ v_{k}^{[l_{k}]} : (l_{k},k) \in \mathcal{K} \right\} \\
\label{eq:potential_graph_E}
\mathcal{E} & = \mathcal{E}_{1}' \cup \mathcal{E}_{1}'' \cup  \mathcal{E}_{2} \cup \mathcal{E}_{3}' \cup \mathcal{E}_{3}'' \\
\label{eq:potential_graph_E_1}
\mathcal{E}_{1}' & = \left\{\big(v_{k}^{[l_{k}']} , v_{k}^{[l_{k}]} \big) :
k \in \langle K \rangle, \; l_{k}',l_{k} \in \langle L_{k} \rangle, \; l_{k}' < l_{k}  \right\} \\
\mathcal{E}_{1}'' & = \left\{\big(v_{k}^{[l_{k}]} , v_{k}^{[l_{k}']} \big) :
k \in \langle K \rangle, \; l_{k}',l_{k} \in \langle L_{k} \rangle, \; l_{k}' < l_{k}   \right\} \\
\mathcal{E}_{2} & = \left\{\big(v_{k}^{[l_{k}]} , v_{j}^{[l_{j}]} \big) : (l_{k},k), (l_{j},j) \in \mathcal{K}, \ k \neq j \right\} \\
\mathcal{E}_{3}' & = \left\{ \big(v_{0}^{[0]} ,v_{k}^{[l_{k}]} \big) : (l_{k},k) \in \mathcal{K} \right\}\\
\label{eq:potential_graph_E_last}
\mathcal{E}_{3}'' & = \left\{\big(v_{k}^{[l_{k}]} , v_{0}^{[0]} \big) : (l_{k},k) \in \mathcal{K} \right\}.
\end{align}
The above digraph, known as the potential graph, consists of $|\mathcal{V}| = 1 + |\mathcal{K}|$ vertices: a ground node $v_{0}^{[0]}$ and one node $v_{i}^{[l_{i}]}$ for each user (or message) indexed by $(l_{i},i) \in \mathcal{K}$.
An example is given in Fig. \ref{fig:potential_graph}.
Each pair of distinct vertices is connected by a pair of edges, and different edges are assigned different lengths, capturing desired and interfering signal power levels as we see next.

\begin{figure}[t]
\centering
\includegraphics[width = 0.9\textwidth,trim={0.5cm 11.5cm 8cm 0.5cm},clip]{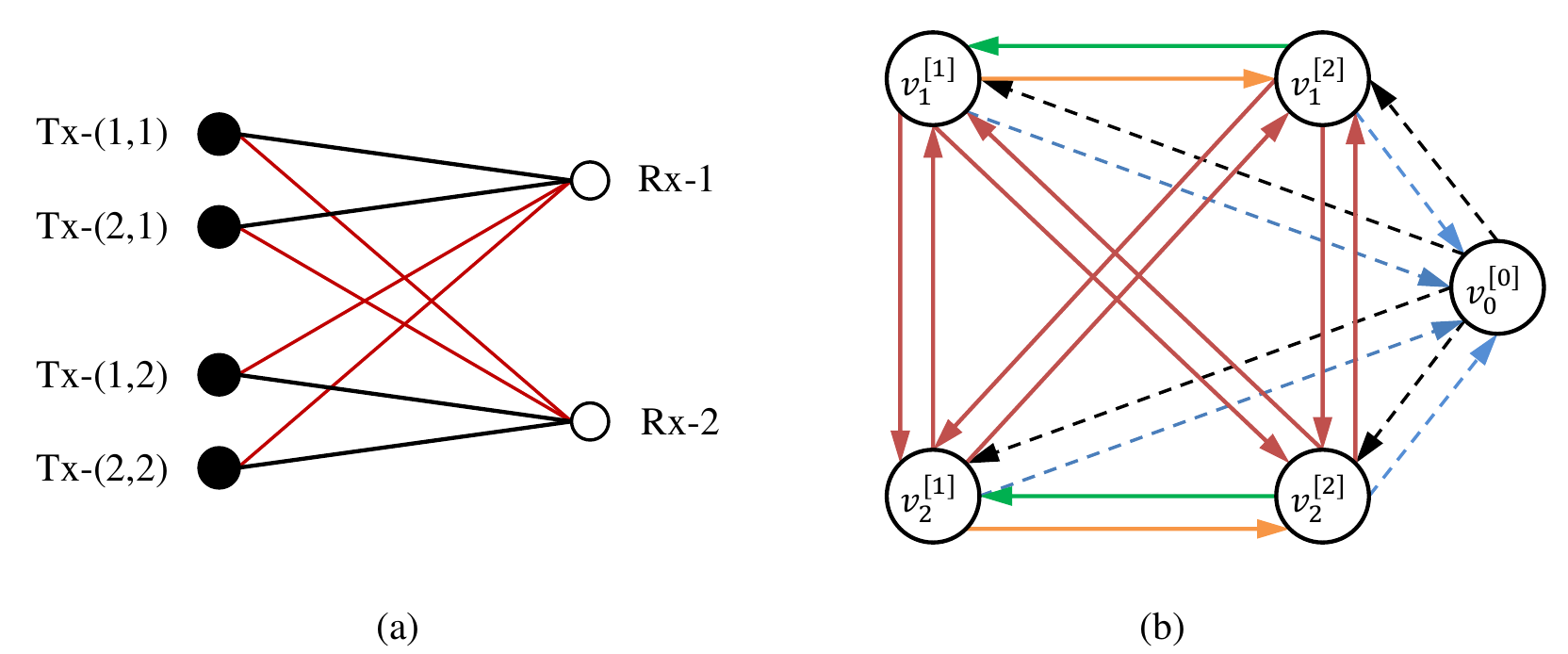}
\caption{\small
A $2$-cell IMAC with $2$ users in each cell (a) and its potential graph (b). The potential graph
consists of $5$ vertices, one for each user and a ground node, in addition to
the sets of edges given by: $\mathcal{E}_{1}'$ in orange, $\mathcal{E}_{1}''$ in green,  $\mathcal{E}_{2}$ in red,
$\mathcal{E}_{3}'$ in dashed black and  $\mathcal{E}_{3}''$  in dashed blue.}
\label{fig:potential_graph}
\end{figure}
We define the functions $d:\mathcal{E} \rightarrow \mathbb{R}_{+}$, $\alpha:\mathcal{E} \rightarrow \mathbb{R}_{+}$ and
$w:\mathcal{E} \rightarrow \mathbb{R}_{+}$ such that for any edge $\big(v_{i}^{[l_{i}]} , v_{j}^{[l_{j}]} \big) \in \mathcal{E}$, these functions take the following values
\begin{align}
\label{eq:edges_GDoF}
d\big(v_{i}^{[l_{i}]} , v_{j}^{[l_{j}]} \big)  & = d_{i}^{[l_{i}]} \\
\label{eq:edges_alpha}
\alpha \big(v_{i}^{[l_{i}]} , v_{j}^{[l_{j}]} \big)  & = \alpha_{ii}^{[l_{i}]} \\
\label{eq:edges_w}
w \big(v_{i}^{[l_{i}]} , v_{j}^{[l_{j}]} \big)  & = \alpha_{ji}^{[l_{j}]}
\mathbbm{1}_{\overline{\mathcal{E}_{1}'}}\big( ( v_{i}^{[l_{i}]} , v_{j}^{[l_{j}]} ) \big)
\end{align}
where $d_{0}^{[0]},\alpha_{00}^{[0]} , \alpha_{j0}^{[l_{j}]}, \alpha_{0j}^{[0]} = 0 $,
$\forall (l_{j},j) \in \mathcal{K}$,
while $\mathbbm{1}_{\overline{\mathcal{E}_{1}'}}\big((v_{i}^{[l_{i}]} , v_{j}^{[l_{j}]})\big) = 0$,  $ \forall i = j$ and $l_{i} < l_{j}$,
 and $1$ otherwise.
The length function is defined as $l:\mathcal{E} \rightarrow \mathbb{R}$,
such that the lengths assigned to different edges of $\mathcal{G}_{\mathrm{p}}$ are given by
\begin{equation}
l(e) = \alpha(e) - w(e) - d(e), \ e \in \mathcal{E}.
\end{equation}
Such lengths are explicitly expressed, for each subset of edges in \eqref{eq:potential_graph_E_1}--\eqref{eq:potential_graph_E_last}, as
\begin{align}
\label{eq:lengths_1}
l\big(v_{k}^{[l_{k}']} , v_{k}^{[l_{k}]} \big) & = \alpha_{kk}^{[l_{k}']} - d_{k}^{[l_{k}']}, \; \forall
k \in \langle K \rangle, \; l_{k}',l_{k} \in \langle L_{k} \rangle, \; l_{k}' < l_{k}  \\
l\big(v_{k}^{[l_{k}]} , v_{k}^{[l_{k}']} \big) & = \alpha_{kk}^{[l_{k}]} - \alpha_{kk}^{[l_{k}']} - d_{k}^{[l_{k}]}, \; \forall
k \in \langle K \rangle, \; l_{k}',l_{k} \in \langle L_{k} \rangle, \; l_{k}' < l_{k}  \\
l\big(v_{k}^{[l_{k}]} , v_{j}^{[l_{j}]} \big) & =  \alpha_{kk}^{[l_{k}]} -  \alpha_{jk}^{[l_{j}]} - d_{k}^{[l_{k}]}, \
\forall (l_{k},k), (l_{j},j) \in \mathcal{K}, \ k \neq j \\
l\big(v_{k}^{[l_{k}]} , v_{0}^{[0]} \big) & =  \alpha_{kk}^{[l_{k}]}  - d_{k}^{[l_{k}]}, \
\forall (l_{k},k) \in \mathcal{K} \\
\label{eq:lengths_5}
l\big( v_{0}^{[0]} , v_{k}^{[l_{k}]} \big) & =  0, \
\forall (l_{k},k) \in \mathcal{K}.
\end{align}
From the above assignment of lengths and the potential theorem we obtain the following result.
\begin{lemma}
\label{lemma:non_negative_circuit_length}
The GDoF tuple $\mathbf{d} \in \mathbb{R}_{+}^{|\mathcal{K}|}$ is in
the polyhedral region $\mathcal{P}_{\bm{\mathrm{id}}}$ if and only if the length of
each directed circuit in the potential graph $\mathcal{G}_{\mathrm{p}}$ is non-negative.
\end{lemma}
\begin{proof}
Note that the length of a directed circuit is given by the sum of the lengths of
its traversed  edges.
By definition \cite{Schrijver2002}, the function $p: \mathcal{V} \rightarrow \mathbb{R}$ is called a potential if for any pair of vertices
$a,b \in \mathcal{V}$ such that $(a,b) \in \mathcal{E}$, we have $l(a,b) \geq p(b) - p(a)$.
These conditions depend only on the difference between potential function values. Therefore, if there exists a valid potential function,
we may assume without loss of generality that the ground node has zero potential, i.e. $p\big(v_{0}^{[0]}\big) = 0$.
Moreover,  the \emph{potential theorem} (see \cite[Th. 8.2]{Schrijver2002}) states that:
\emph{there exists a potential function for a digraph $\mathcal{G}_{\mathrm{p}}$ if and only if each directed circuit
in $\mathcal{G}_{\mathrm{p}}$ has a non-negative length.}

Now for the digraph $\mathcal{G}_{\mathrm{p}}$,
we set the value of the potential function as $p\big(v_{k}^{[l_{k}]}\big) = r_{k}^{[l_{k}]}$, $(l_{k},k) \in \mathcal{K}$. By definition, the potential
function values should satisfy
\begin{align}
\label{eq:potential_polyhedral_region_1}
r_{k}^{[l_{k}]}  - r_{k}^{[l_{k}']} & \leq
\alpha_{kk}^{[l_{k}']} - d_{k}^{[l_{k}']}, \; \forall
k \in \langle K \rangle, \; l_{k}',l_{k} \in \langle L_{k} \rangle, \; l_{k}' < l_{k}   \\
\label{eq:potential_polyhedral_region_2}
r_{k}^{[l_{k}']}   - r_{k}^{[l_{k}]} & \leq
\alpha_{kk}^{[l_{k}]} - \alpha_{kk}^{[l_{k}']} - d_{k}^{[l_{k}]}, \; \forall
k \in \langle K \rangle, \; l_{k}',l_{k} \in \langle L_{k} \rangle, \; l_{k}' < l_{k}  \\
\label{eq:potential_polyhedral_region_3}
r_{j}^{[l_{j}]} - r_{k}^{[l_{k}]}  & \leq
\alpha_{kk}^{[l_{k}]} -  \alpha_{jk}^{[l_{j}]} - d_{k}^{[l_{k}]}, \
\forall (l_{k},k), (l_{j},j) \in \mathcal{K}, \ j \neq k \\
\label{eq:potential_polyhedral_region_4}
-r_{k}^{[l_{k}]} & \leq
\alpha_{kk}^{[l_{k}]} - d_{k}^{[l_{k}]}, \ \forall (l_{k},k) \in \mathcal{K} \\
\label{eq:potential_polyhedral_region_5}
r_{k}^{[l_{k}]} & \leq 0, \ \forall (l_{k},k) \in \mathcal{K}.
\end{align}
The inequalities in \eqref{eq:potential_polyhedral_region_2}--\eqref{eq:potential_polyhedral_region_5}
are equivalent to the ones in \eqref{eq:polyhedral_region_32}--\eqref{eq:polyhedral_region_35}.
Moreover, the inequality in \eqref{eq:potential_polyhedral_region_1} is redundant as it is obtained by adding the inequalities in
\eqref{eq:polyhedral_region_32} and \eqref{eq:polyhedral_region_33}.
It follows that  \emph{$\mathbf{d} \in \mathbb{R}_{+}^{|\mathcal{K}|}$ is in $\mathcal{P}_{\bm{\mathrm{id}}}$
if and only if there exists a valid potential function for $\mathcal{G}_{\mathrm{p}}$}.
Combining this with the potential theorem stated above, we conclude that the tuple $\mathbf{d} \in \mathbb{R}_{+}^{|\mathcal{K}|}$ is in
$\mathcal{P}_{\bm{\mathrm{id}}}$ if and only if the length of each directed circuit in $\mathcal{G}_{\mathrm{p}}$ is non-negative.
\end{proof}

Equipped with Lemma \ref{lemma:non_negative_circuit_length},
it remains to interpret the non-negative length conditions as GDoF inequalities.
In particular, each directed circuit in $\mathcal{G}_{\mathrm{p}}$ is identified by a sequence of vertices
$\big(v_{i_{1}}^{[l_{1}]},\ldots,v_{i_{n}}^{[l_{n}]},v_{i_{n+1}}^{[l_{n+1}]}\big)$, where\footnote{In a slight abuse of notation,
we use $n$ here as the length of directed circuits.
This should not be confused with the number of channel uses $n$ defined in Section
 \ref{subsec:msg_rates_capacity_GDoF} and used later on in the converse.} $\big\{v_{i_{1}}^{[l_{1}]},\ldots,v_{i_{n}}^{[l_{n}]}\big\} \subseteq \mathcal{V}$, $(l_{n+1},i_{n+1}) = (l_{1},i_{1})$ and $n \geq 2$.
Alternatively, we may express a directed circuit in terms of its traversed edges as $(e_{1},\ldots,e_{n})$, where $e_{j} = \big(v_{i_{j}}^{[l_{j}]},v_{i_{j+1}}^{[l_{j+1}]}\big)$, $j \in \langle n \rangle$.
For each such circuit, the non-negative length condition of Lemma  \ref{lemma:non_negative_circuit_length}
 yields a GDoF inequality given by
\begin{equation}
\label{eq:non-negative_length}
\sum_{j=1}^{n}l(e_{j}) \geq 0 \Leftrightarrow
\sum_{j = 1}^{n}d(e_{j}) \leq \sum_{j = 1}^{n}\big[ \alpha(e_{j}) - w(e_{j}) \big].
\end{equation}
Next, we closely examine all directed circuits of $\mathcal{G}_{\mathrm{p}}$ to obtain an explicit
characterization of the GDoF inequalities describing $\mathcal{P}_{\bm{\mathrm{id}}}$ while eliminating circuits which are necessarily redundant.
We often refer to a vertex of the type $v_{k}^{[l_{k}]}$, $(l_{k},k) \in \mathcal{K}$, as a user in what follows.
\subsubsection{From Directed Circuits to GDoF Inequalities: A Simple Example}
\label{subsec:GDoF_example}
We start with the simple example in Fig. \ref{fig:potential_graph} and derive insights which prove useful for addressing the general case.
We categorize directed circuits of $\mathcal{G}_{\mathrm{p}}$ in Fig. \ref{fig:potential_graph} into the following classes:
\begin{itemize}
\item Single-user circuits: such circuits take the simple form of $\big(v_{0}^{[0]},v_{i}^{[l_{i}]},v_{0}^{[0]}  \big)$, $(l_{i},i) \in \{1,2\}^{2}$.
From the non-negative length condition in \eqref{eq:non-negative_length}, each circuit in the class  yields an inequality given by $l\big(v_{0}^{[0]},v_{i}^{[l_{i}]} \big) + l\big(v_{i}^{[l_{i}]},v_{0}^{[0]}  \big) = \alpha_{ii}^{[l_{i}]} - d_{i}^{[l_{i}]} \geq 0$. This is rewritten as
    \begin{equation}
    \label{eq:1_user_ineq}
    d_{i}^{[l_{i}]}  \leq \alpha_{ii}^{[l_{i}]}, \ (l_{i},i) \in \{1,2\}^{2}.
    \end{equation}
    It is easily seen that we obtain $4$ inequalities from this class of circuits.
\item Multi-user circuits traversing $v_{0}^{[0]}$: these are given by $\big(v_{0}^{[0]},v_{i_{1}}^{[l_{1}]},\ldots,v_{i_{n}}^{[l_{n}]},v_{0}^{[0]} \big)$, $(l_{j},i_{j}) \in \{1,2\}^{2}$,
    $j \in \langle n \rangle $ and $n \geq 2$.
    GDoF inequalities obtained from this class of circuits are all redundant.
    This follows by comparing the inequality obtained from $\big(v_{0}^{[0]},v_{i_{1}}^{[l_{1}]},\ldots,v_{i_{n}}^{[l_{n}]},v_{0}^{[0]} \big)$, using the non-negative length condition (as shown above), to the inequality obtained from the corresponding circuits given by $\big(v_{i_{1}}^{[l_{1}]},\ldots,v_{i_{n}}^{[l_{n}]},v_{i_{1}}^{[l_{1}]} \big)$.
    Both inequalities bound the sum-GDoF of the same set of users. However, the latter is tighter since it has an extra (negative) interference term $w\big( v_{i_{n}}^{[l_{n}]},v_{i_{1}}^{[l_{1}]} \big)$ on its right-hand-side compared to $w\big(v_{i_{n}}^{[l_{n}]},v_{0}^{[0]}\big) = 0$ for the former.
    Henceforth, we only consider multi-user circuits that do not traverse $v_{0}^{[0]}$.
\item 2-user circuits (same cell): these circuits take the form $\big(v_{i}^{[1]},v_{i}^{[2]},v_{i}^{[1]}  \big)$, $i \in \{1,2\}$.
 From the non-negative length condition, we obtain 2 inequalities described as
 \begin{equation}
 \label{eq:2_user_ineq_same_cell}
    d_{i}^{[2]} + d_{i}^{[1]}  \leq \alpha_{ii}^{[2]}, \ i \in \{1,2\}.
 \end{equation}
Note that given $i \in \{1,2\}$, the single-user inequality in \eqref{eq:1_user_ineq} with $l_{i} = 2$, i.e. $d_{i}^{[2]}  \leq \alpha_{ii}^{[2]}$,  is implied by \eqref{eq:2_user_ineq_same_cell}, hence making the former redundant.
\item 2-user circuits (different cells): such circuits have the form $\big(v_{1}^{[l_{1}]},v_{2}^{[l_{2}]},v_{1}^{[l_{1}]}  \big)$, $(l_{1},l_{2}) \in \{1,2\}^{2}$.
    From the non-negative length condition, we obtain 4 different inequalities given by
\begin{equation}
 \label{eq:2_user_ineq_diff_cell}
    d_{1}^{[l_{1}]} + d_{2}^{[l_{2}]}  \leq \alpha_{11}^{[l_{1}]} - \alpha_{21}^{[l_{2}]} + \alpha_{22}^{[l_{2}]} - \alpha_{12}^{[l_{1}]}, \ (l_{1},l_{2}) \in \{1,2\}^{2}.
\end{equation}
\item 3-user circuits (2 users from cell 1): these take the form $\big(v_{1}^{[l_{1}^{1}]},v_{1}^{[l_{1}^{2}]},v_{2}^{[l_{2}]},v_{1}^{[l_{1}^{1}]}\big)$,
 $(l_{1}^{1},l_{1}^{2},l_{2}) \in \{1,2\}^{3}$  and $l_{1}^{1} \neq l_{1}^{2}$. We start with the case where $(l_{1}^{1},l_{1}^{2}) = (2,1)$.
From the non-negative length condition, we obtain 2 inequalities (one for each $l_{2}$) given by
\begin{equation}
\label{eq:3_user_ineq_1}
d_{1}^{[1]} + d_{1}^{[2]} + d_{2}^{[l_{2}]}  \leq  \alpha_{11}^{[2]} - \alpha_{21}^{[l_{2}]}  +  \alpha_{22}^{[l_{2}]} - \alpha_{12}^{[2]}, \ l_{2} \in \{1,2\}.
\end{equation}
Note that from the right-hand-side of \eqref{eq:3_user_ineq_1}, users $v_{1}^{[1]}$ and $v_{1}^{[2]}$ appear as a single user with desired signal strength $\alpha_{11}^{[2]}$ and received interference $\alpha_{21}^{[l_{2}]}$.
This is because $v_{1}^{[2]}$ precedes user $v_{1}^{[1]}$ in the cyclic order, from which we have
$l\big(v_{1}^{[2]},v_{1}^{[1]} \big) + l\big(v_{1}^{[1]},v_{2}^{[l_{2}]} \big) = \alpha_{11}^{[2]} - \alpha_{21}^{[l_{2}]}  - d_{1}^{[2]} - d_{1}^{[1]}$.
Moreover, the 2-user inequality in \eqref{eq:2_user_ineq_diff_cell} for $l_{1} = 2$, i.e. $d_{1}^{[2]} + d_{2}^{[l_{2}]}  \leq  \alpha_{11}^{[2]} - \alpha_{21}^{[l_{2}]}  +  \alpha_{22}^{[l_{2}]} - \alpha_{12}^{[2]}$, is implied by \eqref{eq:3_user_ineq_1}, which in turn makes the former redundant.

We move on to the case $(l_{1}^{1},l_{1}^{2}) = (1,2)$, for which we obtain 2 more inequalities given by
\begin{equation}
\label{eq:3_user_ineq_1_2}
d_{1}^{[1]} + d_{1}^{[2]} + d_{2}^{[l_{2}]}  \leq  \alpha_{11}^{[1]} + \alpha_{11}^{[2]} - \alpha_{21}^{[l_{2}]}  +  \alpha_{22}^{[l_{2}]} - \alpha_{12}^{[1]}, \ l_{2} \in \{1,2\}.
\end{equation}
Note that since $v_{1}^{[1]}$ precedes $v_{1}^{[2]}$ for this case, the $2$ users do not appear as a single user as seen from the right-hand side of \eqref{eq:3_user_ineq_1_2}.
In fact, it turns out that for each $l_{2}$, \eqref{eq:3_user_ineq_1_2} is redundant since it is obtained by adding
$d_{1}^{[2]} \leq \alpha_{11}^{[2]}$ and $d_{1}^{[1]} + d_{2}^{[l_{2}]} \leq \alpha_{11}^{[1]} - \alpha_{21}^{[l_{2}]}  +  \alpha_{22}^{[l_{2}]} - \alpha_{12}^{[1]} $, obtained from \eqref{eq:1_user_ineq} and \eqref{eq:2_user_ineq_diff_cell} respectively.
\item 3-user circuits (2 users from cell 2): these take the form
$\big(v_{2}^{[l_{2}^{1}]},v_{2}^{[l_{2}^{2}]},v_{1}^{[l_{1}]}, v_{2}^{[l_{2}^{1}]}\big)$,
 $(l_{2}^{1},l_{2}^{2},l_{1}) \in \{1,2\}^{3}$  and $l_{2}^{1} \neq l_{2}^{2}$. The inequalities are obtained as for the previous class while swapping the cell subscripts. For the case where $(l_{2}^{1},l_{2}^{2}) = (2,1)$, we obtain
\begin{equation}
\label{eq:3_user_ineq_2}
d_{1}^{[l_{1}]} + d_{2}^{[1]} + d_{2}^{[2]}  \leq  \alpha_{11}^{[l_{1}]} - \alpha_{21}^{[2]}  +  \alpha_{22}^{[2]} - \alpha_{12}^{[l_{1}]}, \ l_{1} \in \{1,2\}
\end{equation}
from which we conclude that \eqref{eq:2_user_ineq_diff_cell}, with $l_{2} = 2$, is redundant.
For the second case where $(l_{2}^{1},l_{2}^{2}) = (1,2)$, the resulting inequalities are redundant as shown for \eqref{eq:3_user_ineq_1_2}.
\item 4-user circuits (adjacent same-cell users): these take the form $\big(v_{1}^{[l_{1}^{1}]},v_{1}^{[l_{1}^{2}]},v_{2}^{[l_{2}^{1}]},v_{2}^{[l_{2}^{2}]},v_{1}^{[l_{1}^{1}]}  \big)$, $(l_{1}^{1},l_{1}^{2},l_{2}^{1},l_{2}^{2}) \in \{1,2\}^{4}$, $l_{1}^{1} \neq l_{1}^{2}$ and $l_{2}^{1} \neq l_{2}^{2}$,
    from which we obtain 4 inequalities.
    We start with the case where $(l_{1}^{1},l_{1}^{2},l_{2}^{1},l_{2}^{2}) = (2,1,2,1)$, from which we obtain
    \begin{align}
    \label{eq:4_user_ineq_1}
    d_{1}^{[1]} +d_{1}^{[2]}+ d_{2}^{[1]} + d_{2}^{[2]} & \leq  \alpha_{11}^{[2]} - \alpha_{21}^{[2]}  +  \alpha_{22}^{[2]} - \alpha_{12}^{[2]}.
    \end{align}
    Now consider the 3 remaining circuits obtained from $(l_{1}^{1},l_{1}^{2},l_{2}^{1},l_{2}^{2})$ given by $(1,2,2,1)$, $(2,1,1,2)$ and
    $(1,2,1,2)$. The resulting inequalities are given by
    \begin{align}
    \label{eq:4_user_ineq_1_2}
     d_{1}^{[1]} +d_{1}^{[2]}+ d_{2}^{[1]} + d_{2}^{[2]} & \leq  \alpha_{11}^{[1]} +  \alpha_{11}^{[2]} - \alpha_{21}^{[2]}  +  \alpha_{22}^{[2]} - \alpha_{12}^{[1]} \\
     \label{eq:4_user_ineq_1_3}
     d_{1}^{[1]} +d_{1}^{[2]}+ d_{2}^{[1]} + d_{2}^{[2]} & \leq   \alpha_{11}^{[2]} - \alpha_{21}^{[1]}  +  \alpha_{22}^{[1]} +  \alpha_{22}^{[2]} - \alpha_{12}^{[2]} \\
     \label{eq:4_user_ineq_1_4}
     d_{1}^{[1]} +d_{1}^{[2]}+ d_{2}^{[1]} + d_{2}^{[2]} & \leq  \alpha_{11}^{[1]} +  \alpha_{11}^{[2]} - \alpha_{21}^{[1]}  +  \alpha_{22}^{[1]} + \alpha_{22}^{[2]} - \alpha_{12}^{[1]}.
     \end{align}
     As in the 3-user case, the inequalities in \eqref{eq:4_user_ineq_1_2}--\eqref{eq:4_user_ineq_1_4} have additional signal strength terms on their right-hand-sides, compared to  \eqref{eq:4_user_ineq_1}, since $v_{i}^{[1]}$ precedes $v_{i}^{[2]}$ for at least one $i \in \{1,2\}$.
     Hence, the redundancy of \eqref{eq:4_user_ineq_1_2}--\eqref{eq:4_user_ineq_1_4} can be easily shown by following the same argument used for \eqref{eq:3_user_ineq_1_2}.
     Note that \eqref{eq:4_user_ineq_1_4} is implied by 3 inequalities:
     $d_{1}^{[2]} \leq \alpha_{11}^{[2]} $ and $d_{2}^{[2]} \leq \alpha_{22}^{[2]}$, obtained from \eqref{eq:1_user_ineq}, and
     $d_{1}^{[1]} + d_{2}^{[1]} \leq \alpha_{11}^{[1]} - \alpha_{21}^{[1]} +  \alpha_{22}^{[1]} - \alpha_{12}^{[1]}$, obtained from
     \eqref{eq:2_user_ineq_diff_cell}.
\item 4-user circuits (non-adjacent same-cell users): these take the form $\big(v_{1}^{[1]},v_{2}^{[l_{2}^{1}]},v_{1}^{[2]},v_{2}^{[l_{2}^{2}]},v_{1}^{[1]}  \big)$, $(l_{2}^{1},l_{2}^{2}) \in \{1,2\}^{2}$ and $l_{2}^{1} \neq l_{2}^{2}$.
    We obtain 2 inequalities, each given by
    \begin{equation}
    d_{1}^{[1]} +d_{1}^{[2]}+ d_{2}^{[1]} + d_{2}^{[2]}  \leq \alpha_{11}^{[1]} - \alpha_{21}^{[l_{2}^{1}]} + \alpha_{22}^{[l_{2}^{1}]} - \alpha_{12}^{[2]} + \alpha_{11}^{[2]} - \alpha_{21}^{[l_{2}^{2}]} +  \alpha_{22}^{[l_{2}^{2}]} - \alpha_{12}^{[1]}.
    \end{equation}
The above is redundant since it is implied by $d_{1}^{[1]} + d_{2}^{[l_{2}^{1}]} \leq \alpha_{11}^{[1]} - \alpha_{21}^{[l_{2}^{1}]} + \alpha_{22}^{[l_{2}^{1}]} - \alpha_{12}^{[1]} $
and
$d_{1}^{[2]} + d_{2}^{[l_{2}^{2}]} \leq \alpha_{11}^{[2]} - \alpha_{21}^{[l_{2}^{2}]} + \alpha_{22}^{[l_{2}^{2}]} - \alpha_{12}^{[2]} $, where both are obtained from \eqref{eq:2_user_ineq_diff_cell}.
\end{itemize}
After removing all redundant inequalities identified above, we are left with
\begin{align}
d_{1}^{[1]}  & \leq  \alpha_{11}^{[1]} \\
d_{2}^{[1]}  & \leq  \alpha_{22}^{[1]} \\
d_{1}^{[2]} + d_{1}^{[1]}  & \leq  \alpha_{11}^{[2]} \\
d_{2}^{[2]} + d_{2}^{[1]}  & \leq  \alpha_{22}^{[2]} \\
d_{1}^{[1]} + d_{2}^{[1]}  & \leq  \alpha_{11}^{[1]} - \alpha_{21}^{[1]} + \alpha_{22}^{[1]} - \alpha_{12}^{[1]}  \\
d_{1}^{[2]} + d_{1}^{[1]} + d_{2}^{[1]}  & \leq  \alpha_{11}^{[2]} - \alpha_{21}^{[1]} + \alpha_{22}^{[1]} - \alpha_{12}^{[2]}  \\
d_{1}^{[1]} + d_{2}^{[2]} +  d_{2}^{[1]}  & \leq  \alpha_{11}^{[1]} - \alpha_{21}^{[2]} + \alpha_{22}^{[2]} - \alpha_{12}^{[1]}  \\
d_{1}^{[2]} + d_{1}^{[1]} + d_{2}^{[2]} +  d_{2}^{[1]}  & \leq  \alpha_{11}^{[2]} - \alpha_{21}^{[2]} + \alpha_{22}^{[2]} - \alpha_{12}^{[2]}.
\end{align}
By further including $d_{i}^{[l_{i}]} \geq 0$, $(l_{i},i) \in \{1,2\}^{2}$, we obtain the polyhedral TIN region $\mathcal{P}_{\bm{\mathrm{id}}}$
for the example in Fig. \ref{fig:potential_graph}, which coincides with the characterization in Theorem \ref{theorem:TIN_region}.

To summarize, from the above procedure, the following (sub)classes of circuits give rise to redundant inequalities:
$\bullet$ single-user circuits involving a stronger MAC user (i.e. $v_{k}^{[2]}$), $\bullet$ all multi-user circuits traversing the ground node, $\bullet$
2-user circuits (different cells) which involve any of the stronger MAC users, $\bullet$
3-user circuits (2 users from cell $k$) in which the weaker MAC user from cell $k$
precedes the stronger MAC user from the same cell in the cyclic order, or the participating user from cell $(3-k)$ is the stronger MAC user,
$\bullet$  all 4-user circuits (adjacent same-cell users), except for the circuit in which the stronger MAC user precedes the weaker MAC user from the same cell in the cyclic order, $\bullet$ all 4-user circuits (non-adjacent same-cell users).

One may also translate the above findings into more succinct and general principles, which are given as follows: 1)
multi-user circuits that traverse the ground node $v_{0}^{[0]}$ are not useful,
2) users belonging to the same cell must be cyclicly adjacent in 4-user circuits
(this holds automatically for 3-user circuits), 3) a circuit traversing
$v_{i}^{[2]}$ must also traverse $v_{i}^{[1]}$, where $v_{i}^{[2]}$ should precede $v_{i}^{[1]}$ in the cyclic order.
Next, we carry out redundancy elimination for the general case by building upon, and further generalizing, the above principles.
\subsubsection{From Directed Circuits to GDoF Inequalities: The General Case}
We start by introducing some notation employed in showing the general case, particularly in the proof of
Lemma \ref{lemma:redundant_circuits} given in the appendix.
We denote each vertex $v_{k}^{[l_{k}]} \in \mathcal{V}$ in this part by its index tuple $(l_{k},k)$ to avoid cumbersome subscript-superscript notation.
Let $\mathbf{s}^{n} \in \Sigma(\mathcal{K})$ be a cyclicly ordered sequence of $n$ distinct users.
$\mathbf{s}^{n}$ can be partitioned into $m$ single-cell subsequences as
\begin{equation}
\label{eq:cyclic_sequence_partition}
\mathbf{s}^{n} = \big(\mathbf{s}_{1}^{n_{1}},\ldots,\mathbf{s}_{m}^{n_{m}} \big), \ \text{such that}
\end{equation}
\begin{equation}
\label{eq:cyclic_sequence_partition_2}
\mathbf{s}_{j}^{n_{j}}  = \big((l_{j}^{1},i_{j}),\ldots,(l_{j}^{n_{j}},i_{j}) \big) \in \Sigma(\mathcal{K}_{i_{j}}),
\; i_{j} \in \langle K \rangle,  \;
i_{j}  \neq i_{j+1},  \;  \forall j \in \langle m  \rangle
\end{equation}
where a modulo arithmetic is implicitly used on cell indices, i.e. $i_{m + 1} = i_{1}$.
It is readily seen that $n_{j}  \leq L_{i_{j}}$ and $\sum_{j=1}^{m} n_{j} = n$.
Moreover, while two cyclicly adjacent single-cell subsequences in \eqref{eq:cyclic_sequence_partition} cannot have the same cell index,
this is not necessary for nonadjacent subsequence.
Note that the partition in \eqref{eq:cyclic_sequence_partition}--\eqref{eq:cyclic_sequence_partition_2} is cyclicly unique (i.e. unique up to a cyclic shift). Therefore, we always assume that $\mathbf{s}^{n} \in \Sigma(\mathcal{K})$
is given in terms of its cyclicly unique single-cell partition.

Sequences in $\Sigma(\mathcal{K})$  map into two types of circuits in $\mathcal{G}_{\mathrm{p}}$.
The first type is given by
\begin{equation}
\label{eq:circuit_fn_1}
\mathbf{c}(\mathbf{s}^{n})  \triangleq \big(e_{1}^{1},\ldots,e_{1}^{n_{1}},\ldots,e_{m}^{1},\ldots,e_{m}^{n_{m}} \big), \
  \mathbf{s}^{n} = \big(\mathbf{s}_{1}^{n_{1}},\ldots,\mathbf{s}_{m}^{n_{m}} \big) \in \Sigma(\mathcal{K}), n \geq 2
\end{equation}
where each edge $e_{j}^{s} \in \mathcal{E}_{1}' \cup \mathcal{E}_{1}''\cup \mathcal{E}_{2}$ connects a pair of cyclicly consecutive users such that
\begin{equation}
\label{eq:edges_in_circuits}
e_{j}^{s} \triangleq
\begin{cases}
 \big((l_{j}^{s},i_{j}),(l_{j}^{s+1},i_{j}) \big), \; s \in \langle n_{j} -1 \rangle, \; j \in \langle m \rangle \\
 \big((l_{j}^{n_{j}},i_{j}),(l_{j+1}^{1},i_{j+1}) \big), \;  s = n_{j}, \; j \in \langle m \rangle
\end{cases}
\end{equation}
and $(l_{m+1}^{1},i_{m+1}) = (l_{1}^{1},i_{1})$ is implicitly assumed.
The second type of directed circuits is defined as
\begin{equation}
\label{eq:circuit_fn_2}
\mathbf{c}_{0}(\mathbf{s}^{n})  \triangleq \big(e_{0}',e_{1}^{1},\ldots,e_{1}^{n_{1}},\ldots,e_{m}^{1},\ldots,e_{m}^{n_{m}-1},e_{0}''\big), \
  \mathbf{s}^{n} = \big(\mathbf{s}_{1}^{n_{1}},\ldots,\mathbf{s}_{m}^{n_{m}} \big) \in \Sigma(\mathcal{K})
\end{equation}
where $e_{0}' \triangleq  \big(v_{0}^{[0]},(l_{1}^{1},i_{1}) \big) \in \mathcal{E}_{3}'$  and
$e_{0}'' \triangleq \big((l_{m}^{n_{m}},i_{m}), v_{0}^{[0]}\big) \in \mathcal{E}_{3}''$, while the remaining edges are as in \eqref{eq:edges_in_circuits}.
We further categorize circuits in \eqref{eq:circuit_fn_1} and \eqref{eq:circuit_fn_2} as follows:
\begin{itemize}
\item  Single-user circuits: these circuits take the form $\mathbf{c}_{0}\big((l_{i},i)\big) = (e_{0}',e_{0}'')$, $(l_{i},i) \in \mathcal{K}$.
As in the example of Section \ref{subsec:GDoF_example}, from non-negative length condition we obtain
\begin{equation}
\label{eq:GDoF_single_circuit_1}
d_{i}^{[l_{i}]} \leq \alpha_{ii}^{[l_{i}]}, \ (l_{i},i) \in \mathcal{K}.
\end{equation}
\item Multi-user circuits traversing $v_{0}^{[0]}$: such circuits take the form $\mathbf{c}_{0}(\mathbf{s}^{n})$, $\mathbf{s}^{n} \in \Sigma(\mathcal{K})$ and  $n \geq 2$.
    As in the simple example, it can be easily shown that these circuits are redundant since circuits of the type
    $\mathbf{c}(\mathbf{s}^{n})$ yield tighter GDoF inequalities in general.
\item Multi-user circuits not traversing $v_{0}^{[0]}$: these are the remaining circuits which take the form
$\mathbf{c}(\mathbf{s}^{n})$, $\mathbf{s}^{n} \in \Sigma(\mathcal{K})$ and  $n \geq 2$.
Some of these circuits turn out to be redundant as shown through the following
lemma, which proof is given in Appendix \ref{appendix:proof_lemma:redundant_circuits}.
\begin{lemma}
\label{lemma:redundant_circuits}
For any multi-user circuit $\mathbf{c}(\mathbf{s}^{n})$,
where $\mathbf{s}^{n} = \big(\mathbf{s}_{1}^{n_{1}},\ldots,\mathbf{s}_{m}^{n_{m}} \big) \in \Sigma(\mathcal{K})$ and $n\geq2$,
the corresponding GDoF inequality obtained from the non-negative length condition is necessarily redundant if the circuit fails to satisfy the
two following conditions:
\begin{enumerate}[label=C.{\arabic*}]
\item The sequence of cells $(i_{1},\ldots,i_{m})$, associated with the cyclicly unique single-cell partition $\big(\mathbf{s}_{1}^{n_{1}},\ldots,\mathbf{s}_{m}^{n_{m}}\big)$, should include no repetitions.
\label{cond:redundant_2}
\item The sequence of users associated with each single-cell subsequence
$\mathbf{s}_{j}^{n_{j}}  =  \big((l_{j}^{1},i_{j}),\ldots,(l_{j}^{n_{j}},i_{j}) \big) \in \Sigma(\mathcal{K}_{i_{j}})$, for all $j \in \langle m \rangle$,
must take the descending  form $(l_{j}^{1},\ldots,l_{j}^{n_{j}}) = (n_{j},n_{j}-1,\ldots,1)$.
\label{cond:redundant_3}
\end{enumerate}
\end{lemma}
Equipped with Lemma \ref{lemma:redundant_circuits}, it follows that each non-necessarily redundant
multi-user circuit of the form $\mathbf{c}(\mathbf{s}^{n})$
 is uniquely identified (up to a cyclic order) by two sequences:
\begin{enumerate}
\item $(i_{1},\ldots,i_{m}) \! \in \! \Sigma\big(\langle K \rangle \big)$, $ m \! \in \! \langle K \rangle$, which
identifies participating cells and their cyclic order.
\item $(l_{i_{1}},\ldots,l_{i_{m}}) \in \langle L_{i_{1}} \rangle \times \cdots \times \langle L_{i_{m}} \rangle$,
which identifies the number (and identity due to the order in \ref{cond:redundant_3}) of
participating users in each of the participating cells.
\end{enumerate}
Taking all possible such sequences and specializing the non-negative length condition in \eqref{eq:non-negative_length},
we obtain the GDoF inequalities given by
\begin{equation}
\label{eq:GDoF_bounds_single_cell_multi_user}
\sum_{s_{i} \in \langle l_{i} \rangle} d_{i}^{[s_{i}]}  \leq \alpha_{ii}^{[l_{i}]}, \; l_{i} \geq 2, \; i \in \langle K \rangle
\end{equation}
whenever the cycle $\mathbf{c}(\mathbf{s}^{n})$  traverses only $m = 1$ cell, while for $m \geq 2$ cells we obtain
\begin{align}
\nonumber
\sum_{j \in \langle m \rangle}\sum_{s_{i_{j}} \in \langle l _{i_{j}} \rangle} d_{i_{j}}^{[s_{i_{j}}]} & \leq
\sum_{j \in \langle m \rangle}\alpha_{i_{j}i_{j}}^{[l_{i_{j}}]} - \alpha_{i_{j+1}i_{j}}^{[l_{i_{j+1}}]}
\stackrel{\text{(a)}}{=}
\sum_{j \in \langle m \rangle}\alpha_{i_{j}i_{j}}^{[l_{i_{j}}]} - \alpha_{i_{j}i_{j-1}}^{[l_{i_{j}}]}, \\
\label{eq:GDoF_bounds_multi_cell}
\forall l_{i_{j}} \in \langle L_{i_{j}} \rangle, \ & (i_{1},\ldots,i_{m}) \in \Sigma\big(\langle K \rangle\big), m \in \langle 2:K \rangle
\end{align}
where a modulo arithmetic is implicitly used on cell indices, i.e.  $i_{m + 1} = i_{1}$ and $i_{0} = i_{m}$, and (a) follows by rearranging the terms while exploiting the cyclic ordering.
\end{itemize}
It is notable that the single-user GDoF inequalities in \eqref{eq:GDoF_single_circuit_1} with $l_{i} \geq 2$ are redundant as they are
included in the single-cell multi-user inequalities in \eqref{eq:GDoF_bounds_single_cell_multi_user}.
After removing these redundancies, the remaining inequalities, in addition to the non-negativity constraints
$d_{i}^{[l_{i}]} \geq 0$, $\forall (l_{i},i) \in \mathcal{K}$, describe the polyhedral TIN region $\mathcal{P}_{\bm{\mathrm{id}}}$
and coincide with the characterization in Theorem \ref{theorem:TIN_region}.
\subsection{Proof of Theorem \ref{theorem:TIN_region_general}}
\label{subsec:proofs_general_TIN_region_1}
In this part, we turn to the characterization of the general TIN-achievable GDoF region.
To prove the equality in \eqref{eg:general_TIN_region_thrm}, it is sufficient to show that
\begin{equation}
 \mathcal{P}^{\star} = \bigcup_{\bm{\pi}  \in \Pi}  \mathcal{P}^{\star}_{\bm{\pi} }
 \subseteq \bigcup_{\bm{\pi} \in \Pi } \bigcup_{\mathcal{S} \subseteq \mathcal{K}}
 \mathcal{P}_{\bm{\pi}} (\mathcal{S})
\end{equation}
since inclusion in the other direction is given in \eqref{eq:TIN_region_inner_1}.
In turn, the above is shown by proving that for any decoding order  $\bm{\pi} \in \Pi $,
the inclusion given by $\mathcal{P}^{\star}_{\bm{\pi} } \subseteq \cup_{\mathcal{S} \subseteq \mathcal{K}} \mathcal{P}_{\bm{\pi}} (\mathcal{S})$
holds.
Therefore, we focus on a fixed arbitrary decoding order $\bm{\pi} \in \Pi $ henceforth.

Consider an arbitrary GDoF tuple $\mathbf{d}$ in the TIN-achievable GDoF region $\mathcal{P}^{\star}_{\bm{\pi}}$.
We wish to show that there exists $\mathcal{S} \subseteq \mathcal{K}$
such that $\mathbf{d}$ is also in $\mathcal{P}_{\bm{\pi}}(\mathcal{S})$.
By definition,  there exists a feasible power allocation $\mathbf{r} \leq \mathbf{0}$ such that the components
of $\mathbf{d}$ satisfy \eqref{eq:GDoF per user}.
For such tuple $\mathbf{d}$, we may partition $\mathcal{K}$ into
$\mathcal{S}$ and $\overline{\mathcal{S}} = \mathcal{K} \setminus \mathcal{S}$, such that
$d_{k}^{[\pi_{k}(l_{k})]} > 0$ for all $\big( \pi_{k}(l_{k}),k\big) \in \mathcal{S}$ and  $d_{j}^{[\pi_{j}(l_{j})]} = 0$
for all $\big(\pi_{j}(l_{j}),j\big) \in \overline{\mathcal{S}}$.
It follows that for all $\big( \pi_{k}(l_{k}),k\big) \in \mathcal{S}$, we must have
\begin{equation}
\label{eq:general_TIN_proof_d_ineq_1}
0 < d_{k}^{[\pi_{k}(l_{k})]}  \leq \! r_{k}^{[\pi_{k}(l_{k})]}  \!  +   \!  \alpha_{kk}^{[\pi_{k}(l_{k})]}
  \!   - \!   \Bigl(\max \Bigl\{ \max_{l_{k}' < l_{k}} \{ r_{k}^{[\pi_{k}(l_{k}')]}  +  \alpha_{kk}^{[\pi_{k}(l_{k}')]} \}  ,\max_{j \neq k}
  \max_{l_{j}} \{ r_{j}^{[l_{j}]}  +  \alpha_{jk}^{[l_{j}]} \}   \Bigr\} \Bigr)^{+}
\end{equation}
where the outmost $\max\{0,\cdot\}$ in \eqref{eq:GDoF per user} is inactive, and hence removed, since $d_{k}^{[\pi_{k}(l_{k})]} > 0$ for such users.

Next, we define a new feasible power allocation tuple $\tilde{\mathbf{r}} \leq \mathbf{0}$ such that
\begin{equation}
\label{eq:general_TIN_proof_r_tilde}
\tilde{r}_{i}^{[l_{i}]} =
\begin{cases}
r_{i}^{[l_{i}]}, \ \; (l_{i},i) \in \mathcal{S} \\
- \infty, \ (l_{i},i) \in \overline{\mathcal{S}}.
\end{cases}
\end{equation}
With this power allocation, the TIN scheme of Section \ref{subsec:TIN_scheme}
achieves all GDoF tuples $\tilde{\mathbf{d}}$ that satisfy
\begin{align}
\label{eq:general_TIN_proof_GDoF_r_tilde_1}
\tilde{d}_{j}^{[\pi_{j}(l_{j})]}    & =  0, \; \big(\pi_{j}(l_{j}),j\big) \in \overline{\mathcal{S}} \\
\tilde{d}_{k}^{[\pi_{k}(l_{k})]} & \geq 0 , \; \big(\pi_{k}(l_{k}),k\big) \in \mathcal{S}   \\
\nonumber
\tilde{d}_{k}^{[\pi_{k}(l_{k})]} &   \leq \tilde{r}_{k}^{[\pi_{k}(l_{k})]}  +   \alpha_{kk}^{[\pi_{k}(l_{k})]}
  \ \ -
  \\
\label{eq:general_TIN_proof_GDoF_r_tilde_3}
 \Bigl(\max  \Bigl\{ & \max_{\big(\pi_{k}(l_{k}'),k\big) \in \mathcal{S}:l_{k}' < l_{k}} \{ \tilde{r}_{k}^{[\pi_{k}(l_{k}')]}  +  \alpha_{kk}^{[\pi_{k}(l_{k}')]} \}  ,
  \max_{\big(\pi_{j}(l_{j}),j\big) \in \mathcal{S}: j \neq k}
  \{ \tilde{r}_{j}^{[\pi_{j}(l_{j})]}  +  \alpha_{jk}^{[\pi_{j}(l_{j})]} \}   \Bigr\} \Bigr)^{+},  \big(\pi_{k}(l_{k}),k\big) \in \mathcal{S}
\end{align}
which follows by plugging $\tilde{\mathbf{r}}$, as defined in \eqref{eq:general_TIN_proof_r_tilde}, into \eqref{eq:GDoF per user}.
As in \eqref{eq:general_TIN_proof_d_ineq_1}, we note that the outmost
$\max\{0,\cdot\}$ of \eqref{eq:GDoF per user} has also been relaxed in \eqref{eq:general_TIN_proof_GDoF_r_tilde_3}.
This follows because the right-hand-side of the inequality in \eqref{eq:general_TIN_proof_GDoF_r_tilde_3} is no less than the
rightmost-side of the compound inequality \eqref{eq:general_TIN_proof_d_ineq_1}.
Consequently, the GDoF tuple $\tilde{\mathbf{d}} = \mathbf{d}$ is also achieved with the power allocation $\tilde{\mathbf{r}}$.

As a final step, we note that any GDoF tuple $\tilde{\mathbf{d}} = \mathbf{d}$ achieved with the power allocation
$\tilde{\mathbf{r}}$ is also in $\mathcal{P}_{\bm{\pi}} (\mathcal{S})$.
This holds as $\tilde{\mathbf{r}}$ is feasible and \eqref{eq:general_TIN_proof_GDoF_r_tilde_1}--\eqref{eq:general_TIN_proof_GDoF_r_tilde_3}
are the inequalities that define the polyhedral TIN region $\mathcal{P}_{\bm{\pi}} (\mathcal{S})$ (see Section \ref{subsec:polyhedral_TIN}).
Therefore, we have $\mathbf{d} \in \mathcal{P}_{\bm{\pi}} (\mathcal{S})$ which completes the proof.
\section{Proof of Convexity}
\label{sec:proof_polyhedrality}
In this section, we present a proof for Theorem \ref{theorem:polyhedrality}.
We assume that the TIN-convexity conditions in \eqref{eq:Poly_condition_2} and \eqref{eq:Poly_condition_1}
always hold throughout this section.
For ease of exposition, we divide the proof of Theorem \ref{theorem:polyhedrality} into three steps as follows:
\begin{itemize}
\item \emph{Step 1}: We show that under the conditions of Theorem \ref{theorem:polyhedrality} and for any $\mathcal{S} \subseteq \mathcal{K}$,
we have $\mathcal{P}_{\bm{\pi}}(\mathcal{S}) \subseteq \mathcal{P}_{\bm{\mathrm{id}}}(\mathcal{S})$ for all $\bm{\pi} \in \Pi$.
Hence, the general TIN-achievable region in \eqref{eg:general_TIN_region_thrm} becomes
\begin{equation}
\nonumber
\mathcal{P}^{\star} = \bigcup_{\mathcal{S} \subseteq \mathcal{K}}
\mathcal{P}_{\bm{\mathrm{id}}}(\mathcal{S}).
\end{equation}
\item \emph{Step 2}: We show that for any $\mathcal{S} \subseteq \mathcal{K}$, we have
$\mathcal{P}_{\bm{\mathrm{id}}}(\mathcal{S}) \subseteq \mathcal{P}_{\bm{\mathrm{id}}}\big( \cup_{i \in \mathcal{M}} \mathcal{K}_{i} \big)$
for some $\mathcal{M} \subseteq \langle K \rangle$. Hence, the general TIN-achievable region now becomes
\begin{equation}
\nonumber
\mathcal{P}^{\star} = \bigcup_{\mathcal{M} \subseteq \langle K \rangle}
\mathcal{P}_{\bm{\mathrm{id}}}\big( \cup_{i \in \mathcal{M}} \mathcal{K}_{i} \big).
\end{equation}
\item \emph{Step 3}: We show that for any $\mathcal{M} \subseteq \langle K \rangle$, we have
$ \mathcal{P}_{\bm{\mathrm{id}}}\big( \cup_{i \in \mathcal{M}} \mathcal{K}_{i} \big) \subseteq \mathcal{P}_{\bm{\mathrm{id}}}(\mathcal{K})$.
Hence, we have $\mathcal{P}^{\star} = \mathcal{P}_{\bm{\mathrm{id}}}(\mathcal{K})$ as stated in Theorem \ref{theorem:polyhedrality}.
\end{itemize}
Before we proceed, we introduce some definitions and notation that facilitate the proof.
\subsection{A Compact Representation of Polyhedral TIN-Achievable GDoF Regions}
\label{subsubsec:compact_rep_polyhedral}
For any  given decoding order $\bm{\pi} \in \Pi$, we define $\mathcal{F}_{\bm{\pi}}(\mathcal{K})$ as a family of subsets of $\mathcal{K}$
where each member of  $\mathcal{F}_{\bm{\pi}}(\mathcal{K})$, denoted by
$\mathcal{S}_{\bm{\pi}} $,   takes the
form\footnote{For the 2-cell, 3-user network from the running example of Section \ref{sec:main_results}, we have
 $\mathcal{F}_{\bm{\pi}} \big( \big\{  (1,1),(2,1),(1,2)  \big\} \big) =  \Big\{ \big\{ (\pi_{1}(1),1) \big\} , \big\{ (\pi_{1}(1),1) , (\pi_{1}(2),1) \big\} ,
 \big\{ (1,2) \big\} ,  \big\{ (\pi_{1}(1),1) , (1,2)\big\}
  , \big\{ (\pi_{1}(1),1) , (\pi_{1}(2),1) , (1,2) \big\}     \Big\}$.}
\begin{equation}
\label{eq:subset_users_sum_GDoF_pi}
\mathcal{S}_{\bm{\pi}} = \left\{ \big(\pi_{i}(s_{i}), i \big): s_{i} \in \langle l_{i} \rangle, \; i \in \mathcal{M} \right\},
\text{ for some } \mathcal{M} \subseteq \langle K \rangle \text{ and } l_{j} \in \langle L_{j} \rangle, j \in \mathcal{M}.
\end{equation}
Moreover, for a GDoF tuple $\mathbf{d}$, we use $\mathbf{d}(\mathcal{S}_{\bm{\pi}})$ to denote the sum-GDoF
$\sum_{(l_{i},i) \in \mathcal{S}_{\bm{\pi}}} d_{i}^{[l_{i}]}$, where $\mathbf{d}(\emptyset) = 0$.
The characterization of $\mathcal{P}_{\bm{\pi}}$ in  Theorem \ref{theorem:TIN_region} is given in terms of
sum-GDoF inequalities for all possible subsets of users $\mathcal{S}_{\bm{\pi}}$  such that  $\mathcal{S}_{\bm{\pi}} \in \mathcal{F}_{\bm{\pi}}(\mathcal{K})$.
Hence, $\mathcal{P}_{\bm{\pi}}$ in \eqref{eq:IMAC_TIN_region_1} and \eqref{eq:IMAC_TIN_region_3} can be represented compactly
(yet less informatively) by all GDoF tuples $\mathbf{d} \in \mathbb{R}_{+}^{|\mathcal{K}|}$
that satisfy
\begin{equation}
\label{eq:polyhedral_region_compact}
\mathbf{d}(\mathcal{S}_{\bm{\pi}})  \leq f_{\bm{\pi}}(\mathcal{S}_{\bm{\pi}}), \;  \mathcal{S}_{\bm{\pi}} \in \mathcal{F}_{\bm{\pi}}(\mathcal{K})
\end{equation}
where $f_{\bm{\pi}} : \mathcal{F}_{\bm{\pi}}(\mathcal{K}) \rightarrow \mathbb{R}_{+}$ is a normalized
 set function (i.e. $f_{\bm{\pi}} (\emptyset) = 0$)  given by
\begin{equation}
\label{eq:set_function}
f_{\bm{\pi}}(\mathcal{S}_{\bm{\pi}}) =
\begin{cases}
\alpha_{ii}^{[\pi_{i}(l_{i})]}, \; \mathcal{S}_{\bm{\pi}} \text{ with } |\mathcal{M}| = 1 \\
\displaystyle
\min_{(i_{1},\ldots,i_{|\mathcal{M}| }) \in \Sigma(\mathcal{M})}
\sum_{j =1 }^{|\mathcal{M}|} \alpha_{i_{j}i_{j}}^{[\pi_{i_{j}}(l_{i_{j}})]} - \alpha_{i_{j}i_{j-1}}^{[\pi_{i_{j}}(l_{i_{j}})]}
, \; \mathcal{S}_{\bm{\pi}} \text{ with } |\mathcal{M}| \geq 2.
\end{cases}
\end{equation}

More generally, we may define the family
$\mathcal{F}_{\bm{\pi}}(\mathcal{S})$ over any subnetwork $\mathcal{S} \subseteq \mathcal{K}$. In particular, suppose that $\mathcal{S} \subseteq \mathcal{K}$ is given by  $\mathcal{S} = \cup_{i \in \mathcal{M}} \mathcal{S}_{i}$ for some $\mathcal{M} \subseteq \langle K \rangle$ and $\mathcal{S}_{i} \subseteq \mathcal{K}_{i}$, $i \in \mathcal{M}$.
We define $\mathcal{F}_{\bm{\pi}}(\mathcal{S})$ as the family of all subsets of $\mathcal{S}$ which take the form
$\mathcal{S}_{\bm{\pi}} = \left\{ \big(\pi_{i}(s_{i}), i \big) \in \mathcal{S}: s_{i} \leq l_{i}, \; i \in \mathcal{M}' \right\}$, where
$\mathcal{M}' \subseteq \mathcal{M} $ and $l_{i} \in \langle |\mathcal{S}_{i}| \rangle, i \in \mathcal{M}'$.
This in turn allows us to have a similar compact representation for general polyhedral TIN regions
$\mathcal{P}_{\bm{\pi}}(\mathcal{S})$ for any $\mathcal{S} \subseteq \mathcal{K}$.
\subsection{The 3-Step Proof of Convexity}
Now we proceed to show the three steps stated above.

\emph{Step 1}:
First, we consider $\mathcal{S} = \mathcal{K}$ and we show that under the conditions of Theorem \ref{theorem:polyhedrality},
we have $\mathcal{P}_{\bm{\pi}} \subseteq \mathcal{P}_{\bm{\mathrm{id}}}$ for all $\bm{\pi} \in \Pi$.
Consider a GDoF tuple $\mathbf{d}' \in \mathcal{P}_{\bm{\pi}}$ for some $\bm{\pi} \in \Pi$.
Since $\mathbf{d}' \geq \mathbf{0}$, to prove that $\mathbf{d}'$ is also in $\mathcal{P}_{\bm{\mathrm{id}}}$, it is sufficient to
show that the set of inequalities given by
\begin{equation}
\nonumber
\mathbf{d}'(\mathcal{S}_{\bm{\mathrm{id}}})  \leq f_{\bm{\mathrm{id}}}(\mathcal{S}_{\bm{\mathrm{id}}}), \;  \mathcal{S}_{\bm{\mathrm{id}}}
\in \mathcal{F}_{\bm{\mathrm{id}}}(\mathcal{K}).
\end{equation}
Consider some subset of users $\mathcal{S}_{\bm{\mathrm{id}}} \in \mathcal{F}_{\bm{\mathrm{id}}}(\mathcal{K})$
where $\mathcal{S}_{\bm{\mathrm{id}}} = \left\{ (s_{i}, i ): s_{i} \in \langle l_{i}' \rangle, \; i \in \mathcal{M} \right\}$.
We define $\mathcal{S}_{\bm{\pi}}(\mathcal{S}_{\bm{\mathrm{id}}}) =
\left\{ \big(\pi_{i}(s_{i}), i \big): s_{i} \in \langle l_{i} \rangle, \; i \in \mathcal{M} \right\}$ as the smallest member
of the family $\mathcal{F}_{\bm{\pi}}(\mathcal{K}) $, in terms of cardinality, such that
$\mathcal{S}_{\bm{\mathrm{id}}} \subseteq \mathcal{S}_{\bm{\pi}}(\mathcal{S}_{\bm{\mathrm{id}}})$ holds.
This set has the following property.
\begin{remark}
\label{eq:smalles_set_cont_S_id}
For $\mathcal{S}_{\bm{\mathrm{id}}}$ and $\mathcal{S}_{\bm{\pi}}(\mathcal{S}_{\bm{\mathrm{id}}})$ as defined above, where
$\big(\pi_{i}(l_{i}), i \big)$ is the user of cell $i$ to be decoded first in $\mathcal{S}_{\bm{\pi}}(\mathcal{S}_{\bm{\mathrm{id}}})$,
we must have $ \big(\pi_{i}(l_{i}), i \big) \in \mathcal{S}_{\bm{\mathrm{id}}}$ for all $i \in \mathcal{M}$.
This holds as the contrary implies that we can choose a smaller $\mathcal{S}_{\bm{\pi}}(\mathcal{S}_{\bm{\mathrm{id}}})$
which satisfies $\mathcal{S}_{\bm{\mathrm{id}}} \subseteq \mathcal{S}_{\bm{\pi}}(\mathcal{S}_{\bm{\mathrm{id}}})$.
\hfill $\lozenge$
\end{remark}
The above observations lead directly to the following result.
\begin{lemma}
\label{lemma:polyhedrality_step_1}
For $\mathcal{S}_{\bm{\mathrm{id}}} \in \mathcal{F}_{\bm{\mathrm{id}}}(\mathcal{K})$ and
$\mathcal{S}_{\bm{\pi}}(\mathcal{S}_{\bm{\mathrm{id}}})$ as defined above, the following inequality holds
\begin{equation}
\nonumber
f_{\bm{\pi}} \big( \mathcal{S}_{\bm{\pi}}(\mathcal{S}_{\bm{\mathrm{id}}}) \big) \leq f_{\bm{\mathrm{id}}} (\mathcal{S}_{\bm{\mathrm{id}}}).
\end{equation}
\end{lemma}
\begin{proof}
For the case where $\mathcal{M} = \{i\}$ (i.e. $|\mathcal{M}| = 1$), we have the following
\begin{align}
\label{eq:set_function_UB_1}
f_{\bm{\pi}} \big( \mathcal{S}_{\bm{\pi}}(\mathcal{S}_{\bm{\mathrm{id}}}) \big)  =  \alpha_{ii}^{[\pi_{i}(l_{i})]} \leq
\alpha_{ii}^{[l_{i}']}  = f_{\bm{\mathrm{id}}} (\mathcal{S}_{\bm{\mathrm{id}}})
\end{align}
where \eqref{eq:set_function_UB_1} holds due to $ \big(\pi_{i}(l_{i}), i \big) \in \mathcal{S}_{\bm{\mathrm{id}}}$, as shown in
Remark \ref{eq:smalles_set_cont_S_id},
and the order of
direct link strength levels in \eqref{eq:strength_order}.
On the other hand, for $|\mathcal{M}| \geq 2$ we have
\begin{align}
\nonumber
f_{\bm{\pi}} \big( \mathcal{S}_{\bm{\pi}}(\mathcal{S}_{\bm{\mathrm{id}}}) \big) & =
\min_{(i_{1},\ldots,i_{|\mathcal{M}| }) \in \Sigma(\mathcal{M})}
\sum_{j =1 }^{|\mathcal{M}|} \alpha_{i_{j}i_{j}}^{[\pi_{i_{j}}(l_{i_{j}})]} - \alpha_{i_{j}i_{j-1}}^{[\pi_{i_{j}}(l_{i_{j}})]} \\
\label{eq:set_function_UB_2}
& \leq  \min_{(i_{1},\ldots,i_{|\mathcal{M}| }) \in \Sigma(\mathcal{M})}
\sum_{j =1 }^{|\mathcal{M}|} \alpha_{i_{j}i_{j}}^{[l_{i_{j}}']} - \alpha_{i_{j}i_{j-1}}^{[l_{i_{j}}']}  \\
\nonumber
& = f_{\bm{\mathrm{id}}} (\mathcal{S}_{\bm{\mathrm{id}}})
\end{align}
where the inequality in \eqref{eq:set_function_UB_2} follows from
$ \big(\pi_{i}(l_{i}), i \big) \in \mathcal{S}_{\bm{\mathrm{id}}}$ and the condition in \eqref{eq:Poly_condition_2}
of Theorem \ref{theorem:polyhedrality} (see also the equivalent representation in \eqref{eq:Poly_condition_2_2}).
\end{proof}
Equipped with Lemma \ref{lemma:polyhedrality_step_1}, we obtain the following inequalities
\begin{equation}
\nonumber
\mathbf{d}'(\mathcal{S}_{\bm{\mathrm{id}}}) \leq \mathbf{d}'\big( \mathcal{S}_{\bm{\pi}}(\mathcal{S}_{\bm{\mathrm{id}}}) \big)
\leq f_{\bm{\pi}} \big( \mathcal{S}_{\bm{\pi}}(\mathcal{S}_{\bm{\mathrm{id}}}) \big)
\leq f_{\bm{\mathrm{id}}} (\mathcal{S}_{\bm{\mathrm{id}}}).
\end{equation}
By applying the above to every $\mathcal{S}_{\bm{\mathrm{id}}} \in \mathcal{F}_{\bm{\mathrm{id}}}(\mathcal{K})$, we conclude that
$\mathbf{d}' \in \mathcal{P}_{\bm{\mathrm{id}}}$ and hence  $\mathcal{P}_{\bm{\pi}} \subseteq \mathcal{P}_{\bm{\mathrm{id}}}$.

Following the same steps above, it can be shown that under the conditions of Theorem \ref{theorem:polyhedrality}
and for any subnetwork $\mathcal{S} \subseteq \mathcal{K}$, we have
$\mathcal{P}_{\bm{\pi}}(\mathcal{S}) \subseteq \mathcal{P}_{\bm{\mathrm{id}}}(\mathcal{S})$, for all $\bm{\pi}$.
This completes this step.

\emph{Step 2}:
Consider an arbitrary subnetwork $\mathcal{S} \subseteq  \mathcal{K}$ and let
$\mathcal{S} = \cup_{i \in \mathcal{M}} \mathcal{S}_{i}$ for some $\mathcal{M} \subseteq \langle K \rangle$
and $\mathcal{S}_{i} \subseteq \mathcal{K}_{i}$, $i \in \mathcal{M}$.
Moreover, let $\tilde{\mathcal{S}} $ be obtained by augmenting $\mathcal{S}$ such that
$\tilde{\mathcal{S}} = \cup_{i \in \mathcal{M}} \mathcal{K}_{i}$.
Using the compact representation in Section \ref{subsubsec:compact_rep_polyhedral},
the corresponding polyhedral TIN regions are given by
\begin{align}
\label{eq:polyhedrality_region_S_alt}
\mathcal{P}_{\bm{\mathrm{id}}} ( \mathcal{S} ) & = \left\{  \mathbf{d} \in \mathbb{R}_{+}^{|\mathcal{K}|}:
\mathbf{d}(\mathcal{S}')  \leq f_{\bm{\mathrm{id}}}(\mathcal{S}'), \mathcal{S}' \in \mathcal{F}_{\bm{\mathrm{id}}}(\mathcal{S}),
\
\mathbf{d}(\mathbf{s}) = 0, \forall \mathbf{s} \in \mathcal{K} \setminus\mathcal{S}  \right\}
\\
\label{eq:polyhedrality_region_S_tilde_alt}
\mathcal{P}_{\bm{\mathrm{id}}} ( \tilde{\mathcal{S}} ) & = \left\{  \mathbf{d} \in \mathbb{R}_{+}^{|\mathcal{K}|}:
\mathbf{d}(\tilde{\mathcal{S}} ')  \leq f_{\bm{\mathrm{id}}}(\tilde{\mathcal{S}} '), \tilde{\mathcal{S}} ' \in \mathcal{F}_{\bm{\mathrm{id}}}(\tilde{\mathcal{S}}),
\
\mathbf{d}(\mathbf{s}) = 0, \forall \mathbf{s} \in \mathcal{K} \setminus \tilde{\mathcal{S}}   \right\}
\end{align}
To show that $\mathcal{P}_{\bm{\mathrm{id}}} ( \mathcal{S} )
\subseteq \mathcal{P}_{\bm{\mathrm{id}}} ( \tilde{\mathcal{S}})$ holds,
consider a GDoF tuple $\mathbf{d}' \in \mathcal{P}_{\bm{\mathrm{id}}} ( \mathcal{S} ) $.
It follows that $\mathbf{d}'  \geq \mathbf{0} $ and
$\mathbf{d}'(\mathbf{s}) = 0, \forall \mathbf{s} \in \mathcal{K} \setminus \tilde{\mathcal{S}}$, where the
latter holds due to $\{\mathcal{K} \setminus \tilde{\mathcal{S}}\} \subseteq \{\mathcal{K} \setminus \mathcal{S}\}$ and
the equalities in \eqref{eq:polyhedrality_region_S_alt}.
It remains to show that $\mathbf{d}'$ satisfies the rest of the inequalities in \eqref{eq:polyhedrality_region_S_tilde_alt}.

For any $\tilde{\mathcal{S}}' \in \mathcal{F}_{\bm{\mathrm{id}}}(\tilde{\mathcal{S}})$, let
$\mathcal{S}'(\tilde{\mathcal{S}}')$ be the largest set in $\mathcal{F}_{\bm{\mathrm{id}}}(\mathcal{S})$ such that
$\mathcal{S}'(\tilde{\mathcal{S}}')\subseteq \tilde{\mathcal{S}}'$.
Note that $\mathcal{S}'(\tilde{\mathcal{S}}')$ exists and is
non-empty as $\tilde{\mathcal{S}}$ is obtained by augmenting $\mathcal{S}$.
We denote $\mathcal{S}'(\tilde{\mathcal{S}}')$ by $\mathcal{S}'$ henceforth for brevity.
With these definitions in mind, we present the following lemma.
\begin{lemma}
\label{lemma:polyhedrality_step_2}
For any $\tilde{\mathcal{S}}' \in \mathcal{F}_{\bm{\mathrm{id}}}(\tilde{\mathcal{S}})$ and $\mathcal{S}'$ as defined above,
the following inequality holds
\begin{equation}
\nonumber
f_{\bm{\mathrm{id}}}(\mathcal{S}')
\leq f_{\bm{\mathrm{id}}}(\tilde{\mathcal{S}}').
\end{equation}
\end{lemma}
\begin{proof}
We start by expressing $\tilde{\mathcal{S}}'$ as
\begin{equation}
\nonumber
\tilde{\mathcal{S}}' = \big\{ (s_{i},i) : s_{i} \in \langle l_{i} \rangle, \; i \in \mathcal{M}' \big\},
\text{ for some }\mathcal{M}'  \subseteq \mathcal{M} \text{ and } l_{i} \in \langle L_{i} \rangle, i \in \mathcal{M}'.
\end{equation}
Since $\mathcal{S}' \subseteq \tilde{\mathcal{S}}' \cap \mathcal{S}$, it follows that $\mathcal{S}'$ can be expressed as
\begin{equation}
\nonumber
\mathcal{S}' = \big\{ (s_{i}',i)\in \mathcal{S} : s_{i}' \leq l_{i}', \; i \in \mathcal{M}' \big\},
\text{ for some } l_{i}' \leq l_{i}, i \in \mathcal{M}'.
\end{equation}
For $\mathcal{M}' = \{i\}$ (i.e. $|\mathcal{M}'| = 1$), we have
$f_{\bm{\mathrm{id}}}(\tilde{\mathcal{S}}') = \alpha_{ii}^{[l_{i}]} \geq
\alpha_{ii}^{[l_{i}']} = f_{\bm{\mathrm{id}}}(\mathcal{S}')$
which is due to $l_{i}' \leq l_{i}$ (see the order in \eqref{eq:strength_order}).
On the other hand, for $|\mathcal{M}'| \geq 2$, we have
\begin{align}
\nonumber
f_{\bm{\mathrm{id}}}(\tilde{\mathcal{S}}') & =
\min_{(i_{1},\ldots,i_{|\mathcal{M}'| }) \in \Sigma(\mathcal{M}')}
\sum_{j =1 }^{|\mathcal{M}'|} \alpha_{i_{j}i_{j}}^{[l_{i_{j}}]} - \alpha_{i_{j}i_{j-1}}^{[l_{i_{j}}]} \\
\label{eq:set_function_UB_2_2}
& \geq \min_{(i_{1},\ldots,i_{|\mathcal{M}'| }) \in \Sigma(\mathcal{M}')}
\sum_{j =1 }^{|\mathcal{M}'|} \alpha_{i_{j}i_{j}}^{[l_{i_{j}}']} - \alpha_{i_{j}i_{j-1}}^{[l_{i_{j}}']}  \\
\nonumber
& = f_{\bm{\mathrm{id}}}(\mathcal{S}')
\end{align}
where the inequality in \eqref{eq:set_function_UB_2_2} follows from $l_{i}' \leq l_{i}$  and the condition in \eqref{eq:Poly_condition_2}
of Theorem \ref{theorem:polyhedrality}.
\end{proof}
Next, we observe that if $\mathcal{S}' \subset \tilde{\mathcal{S}}'$ (strict inclusion), then any user in the non-empty set $\tilde{\mathcal{S}}' \setminus \mathcal{S}'$ is not in $\mathcal{S}$.
This holds as the contrary implies the existence of a set $\mathcal{S}''\in\mathcal{F}_{\bm{\mathrm{id}}}(\mathcal{S})$ such that
$\mathcal{S}' \subset \mathcal{S}'' \subseteq \tilde{\mathcal{S}}'$, hence contradicting the maximality of $\mathcal{S}'$.
It follows that $\mathbf{d}'(\tilde{\mathcal{S}}' \setminus \mathcal{S}') = 0$ (see the equalities in \eqref{eq:polyhedrality_region_S_alt}).
This observation together with Lemma \ref{lemma:polyhedrality_step_2} lead to
\begin{equation}
\nonumber
\mathbf{d}'(\tilde{\mathcal{S}}')  = \mathbf{d}'(\mathcal{S}') + \mathbf{d}'(\tilde{\mathcal{S}}' \setminus \mathcal{S}')
\leq f_{\bm{\mathrm{id}}}(\mathcal{S}')
\leq f_{\bm{\mathrm{id}}}(\tilde{\mathcal{S}}').
\end{equation}
The above holds for all $\tilde{\mathcal{S}}' \in \mathcal{F}_{\bm{\mathrm{id}}}(\tilde{\mathcal{S}})$
and therefore $\mathcal{P}_{\bm{\mathrm{id}}} (\mathcal{S}) \subseteq
\mathcal{P}_{\bm{\mathrm{id}}} (\tilde{\mathcal{S}})$, which completes this step.

\emph{Step 3}:
In this step we show that $ \mathcal{P}_{\bm{\mathrm{id}}}\big( \cup_{i \in \mathcal{M}} \mathcal{K}_{i} \big) \subseteq \mathcal{P}_{\bm{\mathrm{id}}}(\mathcal{K})$, for any $\mathcal{M} \subseteq \langle K \rangle$, by proving that
the set
$\mathcal{P}_{\bm{\mathrm{id}}}\big( \cup_{i \in \mathcal{M}} \mathcal{K}_{i} \big)$ is
monotonically increasing in $\mathcal{M}$, i.e. the following holds:
\begin{equation}
\label{eq:poly_region_id_monotonic}
\mathcal{P}_{\bm{\mathrm{id}}}\big( \cup_{i \in \mathcal{M}} \mathcal{K}_{i} \big) \subseteq
\mathcal{P}_{\bm{\mathrm{id}}}\big( \cup_{i \in \mathcal{M}\cup\{k\}} \mathcal{K}_{i} \big),
\text{ for any } \mathcal{M} \subseteq \langle K \rangle.
\end{equation}
We assume, without loss of generality, that $\mathcal{M} \subset \langle K \rangle$ and
$k \in \langle K \rangle \setminus \mathcal{M}$.
To demonstrate that \eqref{eq:poly_region_id_monotonic}  holds,
consider a GDoF tuple $\mathbf{d}' \in \mathcal{P}_{\bm{\mathrm{id}}}\big( \cup_{i \in \mathcal{M}} \mathcal{K}_{i} \big)$.
Since $\mathbf{d}' \geq \mathbf{0}$ and $\mathbf{d}'(\mathbf{s}) = 0$,
$ \forall \mathbf{s} \notin \cup_{i \in \mathcal{M} \cup \{k\}} \mathcal{K}_{i} $,
we show that $\mathbf{d}'$ satisfies the remaining inequalities that describe
$\mathcal{P}_{\bm{\mathrm{id}}}\big( \cup_{i \in \mathcal{M} \cup \{k\}} \mathcal{K}_{i} \big)$, i.e.
\begin{equation}
\label{polyhedralit_step_3_rem_inequality}
\mathbf{d}'\big( \mathcal{S}_{\bm{\mathrm{id}}} \big) \leq
f_{\bm{\mathrm{id}}} \big( \mathcal{S}_{\bm{\mathrm{id}}} \big), \ \forall
\mathcal{S}_{\bm{\mathrm{id}}} \in \mathcal{F}_{\bm{\mathrm{id}}} \big( \cup_{i \in \mathcal{M}\cup\{k\}} \mathcal{K}_{i} \big).
\end{equation}
To this end, we present the following useful lemma.
\begin{lemma}
\label{lemma:polyhedrality_step_3}
Consider $\mathcal{M} \subset \langle K \rangle$ and $k \in \langle K \rangle \setminus \mathcal{M}$.
Moreover, for each $i \in \mathcal{M} \cup \{k\}$, consider the set $\mathcal{S}_{i} = \big\{(s_{i},i) : s_{i} \in \langle l_{i} \rangle \big\}$,
where $l_{i} \in \langle L_{i} \rangle$.
The following inequality holds
\begin{equation}
\nonumber
f_{\bm{\mathrm{id}}} \big( \cup_{i \in \mathcal{M}} \mathcal{S}_{i} \big)   \leq
f_{\bm{\mathrm{id}}} \big( \cup_{i \in \mathcal{M}\cup\{k\}} \mathcal{S}_{i} \big).
\end{equation}
\end{lemma}
\begin{proof}
For the case where $|\mathcal{M}| = 1$, it is not difficult to show that the above inequality holds.
In particular, for any pair of cells $i,k \in \langle  K \rangle$, $i \neq k $, the condition in \eqref{eq:Poly_condition_1}
of Theorem \ref{theorem:polyhedrality} implies the following inequality:
$\alpha_{ii}^{[l_{i}]} \leq \alpha_{ii}^{[l_{i}]} + \alpha_{kk}^{[l_{k}]} - (\alpha_{ik}^{[l_{i}]} + \alpha_{ki}^{[l_{k}]})$.
Therefore, we focus on the case where  $|\mathcal{M}| \geq 2$ in what follows while implicitly assuming that $K \geq 3$.

Let $(i_{1}^{\star},\ldots,i_{|\mathcal{M}\cup\{k\}|}^{\star}) \in \Sigma \big( \mathcal{M}\cup\{k\} \big)$
be a cyclic sequence that attains the minimum in the definition of
$f_{\bm{\mathrm{id}}} \big( \cup_{i \in \mathcal{M}\cup\{k\}} \mathcal{S}_{i} \big)$, i.e. one that satisfies
\begin{equation}
\label{eq:polyhedrality_step3_f_k}
f_{\bm{\mathrm{id}}} \big( \cup_{i \in \mathcal{M}\cup\{k\}} \mathcal{S}_{i} \big) =
\sum_{j =1 }^{|\mathcal{M}\cup\{k\}|} \alpha_{i_{j}^{\star}i_{j}^{\star}}^{[l_{i_{j}^{\star}}]} -
\alpha_{i_{j}^{\star}i_{j-1}^{\star}}^{[l_{i_{j}^{\star}}]}.
\end{equation}
Due to the cyclic nature, we may assume without loss of generality that $i_{|\mathcal{M}\cup\{k\}|}^{\star} = k$.
Moreover, we denote the index of the preceding cell $i_{|\mathcal{M}|}$ and the following cell $i_{1}$ by $k'$ and $k''$, respectively.
Now consider the cyclic sequence given by
$(i_{1}^{\star},\ldots,i_{|\mathcal{M}|}^{\star}) = (k'',i_{2}^{\star},\ldots,i_{|\mathcal{M}|-1}^{\star},k')$.
This sequence is clearly in $\Sigma \big( \mathcal{M}\big)$, from which we obtain an upper bound on
$f_{\bm{\mathrm{id}}} \big( \cup_{i \in \mathcal{M}} \mathcal{S}_{i} \big) $ given by
\begin{equation}
\label{eq:polyhedrality_step3_f}
f_{\bm{\mathrm{id}}} \big( \cup_{i \in \mathcal{M}} \mathcal{S}_{i} \big) \leq
\bigg[
\sum_{j =1 }^{|\mathcal{M}|} \alpha_{i_{j}^{\star}i_{j}^{\star}}^{[l_{i_{j}^{\star}}]} \bigg]
- \alpha_{k''k'}^{[l_{k''}]}
-\bigg[
\sum_{j =2 }^{|\mathcal{M}|}
\alpha_{i_{j}^{\star}i_{j-1}^{\star}}^{[l_{i_{j}^{\star}}]} \bigg].
\end{equation}
From \eqref{eq:polyhedrality_step3_f_k} and \eqref{eq:polyhedrality_step3_f}, we obtain
\begin{align}
\nonumber
f_{\bm{\mathrm{id}}} \big( \cup_{i \in \mathcal{M}\cup\{k\}} \mathcal{S}_{i} \big) -
f_{\bm{\mathrm{id}}} \big( \cup_{i \in \mathcal{M}} \mathcal{S}_{i} \big)
& \geq \alpha_{kk}^{[l_{k}]} - \bigg[ \alpha_{kk'}^{[l_{k}]} + \alpha_{k''k}^{[l_{k''}]} - \alpha_{k''k'}^{[l_{k''}]}\bigg] \\
\nonumber
& \geq \alpha_{kk}^{[l_{k}]} - \max_{k',(l_{k''},k''):k \neq k''}\bigg\{ \alpha_{kk'}^{[l_{k}]} + \alpha_{k''k}^{[l_{k''}]} -
\alpha_{k''k'}^{[l_{k''}]} \mathbbm{1} \big( k'' \neq k'\big) \bigg\} \\
\label{eq:polyhedrality_step3_f_f_k}
& \geq 0
\end{align}
where the inequality in \eqref{eq:polyhedrality_step3_f_f_k} follows from the condition in \eqref{eq:Poly_condition_1}
of Theorem \ref{theorem:polyhedrality}.
\end{proof}
Now to show \eqref{polyhedralit_step_3_rem_inequality}, we observe that any
$\mathcal{S}_{\bm{\mathrm{id}}} \in \mathcal{F}_{\bm{\mathrm{id}}} \big( \cup_{i \in \mathcal{M}\cup \{k\}} \mathcal{K}_{i} \big)$
can be expressed as
\begin{equation}
\nonumber
\mathcal{S}_{\bm{\mathrm{id}}} = \mathcal{S}_{\bm{\mathrm{id}}}' \cup \mathcal{S}_{\bm{\mathrm{id}}}'', \
\text{for some} \ \mathcal{S}_{\bm{\mathrm{id}}}' \in  \mathcal{F}_{\bm{\mathrm{id}}} \big( \cup_{i \in \mathcal{M}} \mathcal{K}_{i} \big)
\ \text{and} \
\mathcal{S}_{\bm{\mathrm{id}}}''
\in  \mathcal{F}_{\bm{\mathrm{id}}} \big( \mathcal{K}_{k} \big),
\end{equation}
where we highlight that $\mathcal{S}_{\bm{\mathrm{id}}}' $ or $\mathcal{S}_{\bm{\mathrm{id}}}''$
may be equal to $\emptyset$.
Since  $\mathbf{d}' \in \mathcal{P}_{\bm{\mathrm{id}}}\big( \cup_{i \in \mathcal{M}} \mathcal{K}_{i} \big)$,
we have $ \mathbf{d}'(\mathcal{S}_{\bm{\mathrm{id}}}'') = 0$.
Combining this with Lemma \ref{lemma:polyhedrality_step_3}, we obtain
\begin{align}
\nonumber
\mathbf{d}'(\mathcal{S}_{\bm{\mathrm{id}}}) & =  \mathbf{d}'(\mathcal{S}_{\bm{\mathrm{id}}}')
\leq f_{\bm{\mathrm{id}}} \big( \mathcal{S}_{\bm{\mathrm{id}}}' \big)
\leq f_{\bm{\mathrm{id}}} \big( \mathcal{S}_{\bm{\mathrm{id}}}  \big),
\end{align}
which in turn proves \eqref{eq:poly_region_id_monotonic}.
Therefore, we have $\mathcal{P}_{\bm{\mathrm{id}}}\big( \cup_{i \in \mathcal{M}} \mathcal{K}_{i} \big) \subseteq
\mathcal{P}_{\bm{\mathrm{id}}}\big( \mathcal{K} )$, for any $\mathcal{M} \subseteq \langle K \rangle$,
which completes this step and with it the proof of Theorem \ref{theorem:polyhedrality}.
\section{Proof of Optimality}
\label{sec:TIN-Optimality}
The TIN-optimality result in Theorem \ref{theorem:TIN_optimality} follows directly from the following outer bound.
\begin{theorem}
\label{theorem:outer_bound}
For the IMAC with input-output relationship in \eqref{eq:system model 2}, if the TIN-optimality conditions in \eqref{eq:TIN_condition_2}
and \eqref{eq:TIN_condition_1} hold, then the capacity region $\mathcal{C}$ is included in the set of rate tuples satisfying
\begin{align}
\label{eq:capacity_outer_1}
\sum_{s_{i}\in \langle l_{i} \rangle}R_{i}^{[s_{i}]} & \leq \log\left(1 + l_{i}P^{\alpha_{ii}^{[l_{i}]}} \right),
\; l_{i} \in \langle L_{i} \rangle, \forall i \in \langle K \rangle\\
\nonumber
\sum_{j \in \langle m \rangle } \sum_{s_{i_{j}} \in \langle l_{i_{j}} \rangle} R_{i_{j}}^{[s_{i_{j}}]} & \leq
m (l_{i_{j}} - 1) \log(l_{i_{j}}) + \sum_{j \in \langle m \rangle } \log\left( 1 +
(l_{i_{j+1}} + l_{i_{j}}) P^{\alpha_{i_{j}i_{j}}^{[l_{i_{j}}]}-\alpha_{i_{j}i_{j-1}}^{[l_{i_{j}}]}}  \right), \\
\label{eq:capacity_outer_2}
\forall l_{i_{j}} \in \langle L_{i_{j}} \rangle, \; & (i_{1},\ldots,i_{m}) \in \Sigma\big(\langle K \rangle\big), m \in \langle 2:K \rangle.
\end{align}
\end{theorem}
As noted in Remark \ref{remark:constant_gap}, the above result leads to a constant-gap characterization of
the capacity region when the TIN-optimality conditions hold.
The remainder of this section is dedicated to proving Theorem
\ref{theorem:outer_bound}.
To this end, we start by presenting two instrumental lemmas.
\subsection{Useful Lemmas}
The first lemma is a generalization of \cite[Lem. 8]{Gherekhloo2016} to an arbitrary number of input sequences.
\begin{lemma}
\label{lemma:entropy_diff}
Let $X_{1}^{n},\ldots,X_{l}^{n}$ be $l$ independent random sequences (input sequences) of length $n$ each, where
$X_{i}^{n} = X_{i}(1),\ldots,X_{i}(n)$, $i \in \langle l \rangle$, satisfies the power constraint
$\frac{1}{n}\sum_{t=1}^{n}\E \big[|X_{i}(t)|^{2}\big] \leq P_{i}$.
Moreover, let $Y_{a}^{n}$ and $Y_{a}^{n}$ be noisy output sequences given by
\begin{align}
Y_{a}(t) & = a_{1}X_{1}(t) + a_{2}X_{2}(t) + \cdots +  a_{l}X_{l}(t) + Z_{a}(t)  \\
Y_{b}(t) & = b_{1}X_{1}(t) + b_{2}X_{2}(t) + \cdots +  b_{l}X_{l}(t) + Z_{b}(t)
\end{align}
where $a_{i},b_{i} \in \mathbb{C}$, $\forall i \in \langle l \rangle$, are constants and $Z_{a}(t),Z_{b}(t) \sim\mathcal{N}_{\mathbb{C}}(0,1)$ are AWGN terms.
Given that
\begin{equation}
\label{eq:lemma_entropy_diff_cond_1}
1 \leq P_{i}|a_{i}|^{2} \leq \frac{ P_{i}|b_{i}|^{2} }{ P_{i+1}|b_{i+1}|^{2}  }, \; \forall i \in \langle l \rangle
\end{equation}
where $P_{l+1}|b_{l+1}|^{2} = 1$, the difference between the output differential entropies is bounded as
\begin{equation}
\label{eq:lemma_entropy_diff}
h(Y_{a}^{n}) - h(Y_{b}^{n})  \leq n(l-1) \log(l).
\end{equation}
\end{lemma}
The next lemma gives a variant of the TIN-optimality condition in
\eqref{eq:TIN_condition_2}.
\begin{lemma}
\label{lemma:user_partition}
Consider $i,j \in \langle K \rangle$, such that $i \neq j$, and the set $\langle l_{i} \rangle$, where $l_{i} \in \langle L_{i} \rangle$.
Moreover, consider the partition of $\langle l_{i} \rangle$ given by
\begin{align}
\label{eq:partition_prime_2}
\langle l_{i} \rangle''_{j} & \triangleq \Big\{ s_{i}'' \in \langle l_{i} - 1\rangle :
\alpha_{ii}^{[l_{i}]} - \alpha_{ij}^{[l_{i}]} \geq \alpha_{i i}^{[s_{i}'']}  \Big\} \\
\label{eq:partition_prime_1}
\langle l_{i} \rangle'_{j} & \triangleq \langle l_{i} \rangle \setminus \langle l_{i} \rangle''_{j}.
\end{align}
Given that the TIN condition in \eqref{eq:TIN_condition_2} holds, then we have
\begin{equation}
\alpha_{ii}^{[l_{i}]} - \alpha_{ij}^{[l_{i}]} \geq  \alpha_{ii}^{[s_{i}']} - \alpha_{ij}^{[s_{i}']} + \alpha_{ij}^{[l_{i}']}, \; \forall
s_{i}',l_{i}' \in \langle l_{i} \rangle'_{j}\setminus\{l_{i}\}, \; s_{i}' < l_{i}'.
\end{equation}
\end{lemma}
The proofs of Lemma \ref{lemma:entropy_diff}
and Lemma \ref{lemma:user_partition}
are given in Appendix \ref{appendix:proof_lemma:entropy_diff}
and Appendix  \ref{appendix:proof_lemma:user_partition}, respectively.
\subsection{Proof of Theorem \ref{theorem:outer_bound}}
\label{subsec:proof_theorem_outer_bound}
In the following, we use the notation of the channel in \eqref{eq:system model}
and the channel in \eqref{eq:system model 2} interchangeably for convenience.
While doing so, we assume that
\begin{equation}
\label{eq:converse_ch_gains}
1 \leq P_{k}^{[l_{k}]} \big|h_{ki}^{[l_{k}]}\big|^{2}  = P^{\alpha_{ki}^{[l_{k}]}}, \; \forall (l_{k},k) \in \mathcal{K}, i \in \langle K \rangle.
\end{equation}
For each cell $i$, \eqref{eq:capacity_outer_1} is a cut-set upper bound which
follows from the MAC capacity region \cite{Cover2012} and the order of channel strength levels in \eqref{eq:strength_order}. Hence, we focus on the cyclic bounds in \eqref{eq:capacity_outer_2}.
A modulo operation is implicitly used on receiver indices such that $i_{0} = i_{m}$ and $i_{m+1} = i_{1}$.

An arbitrary bound in \eqref{eq:capacity_outer_2} is identified by the two sequences
$(i_{1},\ldots,i_{m}) \in \Sigma\big(\langle K \rangle\big)$ and $(l_{i_{1}},\ldots,l_{i_{m}}) \in \langle L_{i_{1}} \rangle \times \cdots \times
\langle L_{i_{m}} \rangle $, describing the cyclic order of cells and the number of participating users from each cell, respectively.
The corresponding set of participating users is given by $\left\{(s_{i_{j}},i_{j}): s_{i_{j}} \in
\langle l_{i_{j}} \rangle,j \in \langle m \rangle \right\}$.
For every $j \in \langle m \rangle$, we partition $\langle l_{i_{j}} \rangle$ as in Lemma \ref{lemma:user_partition} into
\begin{align}
\label{eq:l_i_j_prime_subset_2}
\langle l_{i_{j}} \rangle'' & \triangleq \Big\{ s_{i_{j}}'' \in \langle l_{i_{j}} - 1 \rangle : \alpha_{i_{j}i_{j}}^{[l_{i_{j}}]} - \alpha_{i_{j}i_{j-1}}^{[l_{i_{j}}]} \geq  \alpha_{i_{j} i_{j}}^{[s_{i_{j}}'']}  \Big\}\\
\label{eq:l_i_j_prime_subset}
\langle l_{i_{j}} \rangle' & \triangleq \langle l_{i_{j}} \rangle \setminus \langle l_{i_{j}} \rangle''.
\end{align}
where the subscript $i_{j-1}$ is omitted from the subsets $\langle l_{i_{j}} \rangle''$ and $\langle l_{i_{j}} \rangle'$  for notational brevity.
Next, we go through the following steps:
\begin{itemize}
\item Eliminate all non-participating transmitters $(l_{i},i) \in \mathcal{K} \setminus \left\{(s_{i_{j}},i_{j}): s_{i_{j}} \in
\langle l_{i_{j}} \rangle,j \in \langle m \rangle \right\}$, all non-participating receivers
$i \in \langle K \rangle \setminus \{i_{1},\ldots,i_{m} \}$ and the corresponding messages.
\item For the remaining network, eliminate all interfering links except for links from
Tx-$(s_{i_{j}}',i_{j})$ to Rx-$i_{j-1}$, for every $ j \in \langle m \rangle$ and $s_{i_{j}}' \in \langle l_{i_{j}} \rangle'$.
\end{itemize}
We end up with a partially connected cyclic IMAC with input-output relationship given by
\begin{equation}
\label{eq:system model_cyclic}
Y_{i_{j}}(t) = \sum_{s_{i_{j}} \in \langle l_{i_{j}} \rangle} h_{i_{j}i_{j}}^{[s_{i_{j}}]} \tilde{X}_{i_{j}}^{[s_{i_{j}}]}(t) +
\overbrace{
\sum_{s_{i_{j+1}} \in  \langle l_{i_{j+1}} \rangle'}  \! \! \!  h_{i_{j+1}i_{j}}^{[s_{i_{j+1}}]} \tilde{X}_{i_{j+1}}^{[s_{i_{j+1}}]}(t) + Z_{i_{j}}(t)}^{U_{i_{j+1}}(t)}
\end{equation}
where $U_{i_{j+1}}(t)$ denotes the interference plus noise term at Rx-$i_{j}$.
Since none of the above two steps hurts the rates of the remaining messages, the channel in \eqref{eq:system model_cyclic}
is used for the outer bound.

Next, we define the side information intended to Rx-$i_{j}$ as
\begin{equation}
\label{eq:side_information_cyclic}
S_{i_{j}}(t) = g_{i_{j}} \sum_{s_{i_{j}} \in \langle l_{i_{j}} \rangle'}    h_{i_{j}i_{j}}^{[s_{i_{j}}]} \tilde{X}_{i_{j}}^{[s_{i_{j}}]}(t)
+ Z_{i_{j-1}}(t)
\end{equation}
where the gain $g_{i_{j}}$ is given by
\begin{equation}
g_{i_{j}} = \frac{h_{i_{j}i_{j-1}}^{[l_{i_{j}}]}}{h_{i_{j}i_{j}}^{[l_{i_{j}}]}}.
\end{equation}
The side information sequence $S_{i_{j}}^{n}$ is given to Rx-$i_{j}$ through a genie, which cannot hurt the rates.
Using Fano's inequality, we bound the sum rate of participating users associated with cell $i_{j}$ as
\begin{align}
\nonumber
n \sum_{s_{i_{j}} \in \langle l_{i_{j}} \rangle} R_{i_{j}}^{[s_{i_{j}}]} -  n \epsilon & \leq I
\big(W_{i_{j}}^{[1:l_{i_{j}}]} ; Y_{i_{j}}^{n}, S_{i_{j}}^{n}   \big) \\
\nonumber
&  = I\big(W_{i_{j}}^{[1:l_{i_{j}}]} ; S_{i_{j}}^{n}   \big) +
I\big(W_{i_{j}}^{[1:l_{i_{j}}]} ; Y_{i_{j}}^{n} | S_{i_{j}}^{n}   \big) \\
\nonumber
& = h\big(S_{i_{j}}^{n}\big) - h\big(S_{i_{j}}^{n}  | W_{i_{j}}^{[1:l_{i_{j}}]}\big) +
h\big(Y_{i_{j}}^{n}  | S_{i_{j}}^{n}\big) - h\big( Y_{i_{j}}^{n} |  S_{i_{j}}^{n} , W_{i_{j}}^{[1:l_{i_{j}}]} \big) \\
\label{eq:Fano_single_cell}
& = h\big(S_{i_{j}}^{n}\big) - h\big(Z_{i_{j-1}}^{n} \big) +
h\big(Y_{i_{j}}^{n}  | S_{i_{j}}^{n}\big) - h\big( U_{i_{j+1}}^{n} \big)
\end{align}
where $W_{i_{j}}^{[1:l_{i_{j}}]} \triangleq W_{i_{j}}^{[1]},\ldots,W_{i_{j}}^{[l_{i_{j}}]}$.
Taking the sum of bounds in  \eqref{eq:Fano_single_cell} for all $j \in \langle m \rangle$, we obtain a bound on the sum rate of all participating users as
\begin{align}
\nonumber
n \sum_{j \in \langle m \rangle}\sum_{s_{i_{j}} \in \langle l_{i_{j}} \rangle} R_{i_{j}}^{[s_{i_{j}}]} - mn \epsilon  & \leq
 \sum_{j \in \langle m \rangle} \left[ h\big(S_{i_{j}}^{n}\big)  - h\big(U_{i_{j}}^{n}\big) +
  h\big(Y_{i_{j}}^{n}  | S_{i_{j}}^{n}\big) - h\big(Z_{i_{j}}^{n} \big) \right] \\
 \label{eq:sum_rate_bound_1}
  & \leq m n (l_{i_{j}} - 1) \log(l_{i_{j}}) + \sum_{j \in \langle m \rangle} \left[ h\big(Y_{i_{j}}^{n}  | S_{i_{j}}^{n}\big) - h\big(Z_{i_{j}}^{n} \big) \right]
\end{align}
where \eqref{eq:sum_rate_bound_1} follows by bounding each $h\big(S_{i_{j}}^{n}\big)  - h\big(U_{i_{j}}^{n}\big)$ as explained next.
For all $i_{j}$ with $\langle l_{i_{j}} \rangle' = \{ l_{i_{j}} \}$, it is readily seen that $h\big(S_{i_{j}}^{n}\big)  - h\big(U_{i_{j}}^{n}\big) = 0$
as $S_{i_{j}}^{n} = U_{i_{j}}^{n}$.
Otherwise, for $i_{j}$ such that $\langle l_{i_{j}} \rangle'\setminus\{ l_{i_{j}} \} \neq \emptyset$, we apply Lemma \ref{lemma:entropy_diff}
by taking $S_{i_{j}}^{n}$ and $U_{i_{j}}^{n}$ as the corresponding output sequences.
It remains to verify that the condition in \eqref{eq:lemma_entropy_diff_cond_1} holds.
From Lemma \ref{lemma:user_partition}, the following condition holds
\begin{equation}
\label{eq:condition_TIN_S3_1}
\alpha_{i_{j}i_{j}}^{[l_{i_{j}}]} - \alpha_{i_{j}i_{j-1}}^{[l_{i_{j}}]} \geq  \alpha_{i_{j}i_{j}}^{[s_{i_{j}}]} - \alpha_{i_{j}i_{j-1}}^{[s_{i_{j}}]} + \alpha_{i_{j}i_{j-1}}^{[l_{i_{j}}']}, \; \forall s_{i_{j}},l_{i_{j}}' \in \langle l_{i_{j}} \rangle', \; s_{i_{j}} < l_{i_{j}}'.
\end{equation}
Moreover, from the definition of the partition in \eqref{eq:l_i_j_prime_subset_2} and \eqref{eq:l_i_j_prime_subset}, we have
\begin{equation}
\label{eq:condition_TIN_S3_1_2}
\alpha_{i_{j}i_{j}}^{[l_{i_{j}}]} - \alpha_{i_{j}i_{j-1}}^{[l_{i_{j}}]} <  \alpha_{i_{j}i_{j}}^{[s_{i_{j}}]},
\; \forall s_{i_{j}} \in \langle l_{i_{j}} \rangle' \setminus \{l_{i_{j}}\}.
\end{equation}
From \eqref{eq:converse_ch_gains}, the conditions in \eqref{eq:condition_TIN_S3_1} and \eqref{eq:condition_TIN_S3_1_2} can
be rewritten as
\begin{equation}
\label{eq:condition_TIN_S3}
0 <
|g_{i_{j}} |^{2}
P_{i_{j}}^{[s_{i_{j}}]}\big|h_{i_{j}i_{j}}^{[s_{i_{j}}]} \big|^{2}
 \leq
\frac{P_{i_{j}}^{[s_{i_{j}}]}\big|h_{i_{j}i_{j-1}}^{[s_{i_{j}}]} \big|^{2}}{P_{i_{j}}^{[l_{i_{j}}']}\big|h_{i_{j}i_{j-1}}^{[l_{i_{j}}']} \big|^{2}}
, \;  \forall s_{i_{j}},l_{i_{j}}' \in \langle l_{i_{j}} \rangle', \; s_{i_{j}} < l_{i_{j}}'.
\end{equation}
Note that \eqref{eq:condition_TIN_S3} implies the condition in \eqref{eq:lemma_entropy_diff_cond_1} of  Lemma \ref{lemma:entropy_diff},
from which we obtain the upper bound $h\big(S_{i_{j}}^{n}\big)  - h\big(U_{i_{j}}^{n}\big) \leq
n (l_{i_{j}} - 1) \log(l_{i_{j}})$, which holds for all $i_{j}$, $j \in \langle m \rangle$.

Now we turn our attention to $h\big(Y_{i_{j}}^{n}  | S_{i_{j}}^{n}\big) - h\big(Z_{i_{j}}^{n} \big)$ in \eqref{eq:sum_rate_bound_1}. For any $j \in \langle m \rangle$, we have
\begin{align}
\nonumber
h\big(Y_{i_{j}}^{n}  | S_{i_{j}}^{n}\big) - h\big(Z_{i_{j}}^{n} \big) & \leq
\sum_{t \in \langle n \rangle} \big[ h\big(Y_{i_{j}}(t)  | S_{i_{j}}(t)\big) - h\big(Z_{i_{j}}(t) \big)  \big] \\
\label{eq:conditional_entropy_Gaussian}
& \leq n h\big(Y_{i_{j}}^{\mathrm{G}}  | S_{i_{j}}^{\mathrm{G}}\big) - nh\big(Z_{i_{j}}\big) \\
\label{eq:conditional_entropy_Gaussian_2}
& = n \log \big( \sigma^{2}_{Y_{i_{j}}^{\mathrm{G}}  | S_{i_{j}}^{\mathrm{G}}} \big)
\end{align}
where $Y_{i_{j}}^{\mathrm{G}}$ and $S_{i_{j}}^{\mathrm{G}}$ are the outputs in \eqref{eq:system model_cyclic} and \eqref{eq:side_information_cyclic} respectively for a single use of the channel when the corresponding inputs are drawn from independent Gaussian distributions as
$\tilde{X}_{i}^{[l_{i}]}\sim \mathcal{N}_{\mathbb{C}}\big(0,P_{i}^{[l_{i}]}\big)$.
The inequality in \eqref{eq:conditional_entropy_Gaussian} follows by employing \cite[Lem. 1]{Annapureddy2009}
and the i.i.d. noise assumption, where $t$ is omitted from $Z_{i_{j}}$ for brevity.
The variance in \eqref{eq:conditional_entropy_Gaussian_2} is given by
\begin{equation}
\label{eq:cond_var}
\sigma^{2}_{Y_{i_{j}}^{\mathrm{G}}  | S_{i_{j}}^{\mathrm{G}}} \triangleq
\E \big[ |Y_{i_{j}}^{\mathrm{G}}|^{2} \big] - \E \big[ Y_{i_{j}}^{\mathrm{G}}S_{i_{j}}^{\mathrm{G}\ast} \big]
\big(\E \big[ |S_{i_{j}}^{\mathrm{G}}|^{2} \big]\big)^{-1}
\E \big[ S_{i_{j}}^{\mathrm{G}}Y_{i_{j}}^{\mathrm{G}\ast} \big].
\end{equation}
Next, we calculate each term in \eqref{eq:cond_var}.
We have
\begin{equation}
\E \big[ |Y_{i_{j}}^{\mathrm{G}}|^{2} \big]  =
\sum_{s_{i_{j}}' \in \langle l_{i_{j}} \rangle'}  \big| h_{i_{j}i_{j}}^{[s_{i_{j}}']} \big|^{2} P_{i_{j}}^{[s_{i_{j}}']} + \!
\sum_{s_{i_{j}}'' \in \langle l_{i_{j}} \rangle''}  \big| h_{i_{j}i_{j}}^{[s_{i_{j}}'']} \big|^{2} P_{i_{j}}^{[s_{i_{j}}'']} + \! \!
\sum_{s_{i_{j+1}} \in  \langle l_{i_{j+1}} \rangle'}   \big| h_{i_{j+1}i_{j}}^{[s_{i_{j+1}}]} \big|^{2} P_{i_{j+1}}^{[s_{i_{j+1}}]} + 1
\end{equation}
\begin{equation}
\E \big[ Y_{i_{j}}^{\mathrm{G}}S_{i_{j}}^{\mathrm{G}\ast} \big] = \Big( \E \big[ S_{i_{j}}^{\mathrm{G}}Y_{i_{j}}^{\mathrm{G}\ast} \big] \Big)^{\ast} =  g_{i_{j}}^{\ast}   \sum_{s_{i_{j}} \in \langle l_{i_{j}} \rangle' } \big| h_{i_{j}i_{j}}^{[s_{i_{j}}]} \big|^{2} P_{i_{j}}^{[s_{i_{j}}]}
\end{equation}
\begin{equation}
\E \big[ |S_{i_{j}}^{\mathrm{G}}|^{2} \big] = | g_{i_{j}} |^{2} \sum_{s_{i_{j}} \in \langle l_{i_{j}} \rangle' } \big| h_{i_{j}i_{j}}^{[s_{i_{j}}]} \big|^{2} P_{i_{j}}^{[s_{i_{j}}]} + 1.
\end{equation}
From the above, we obtain an upper bound for $\sigma^{2}_{Y_{i_{j}}^{\mathrm{G}}  | S_{i_{j}}^{\mathrm{G}}}$ as
\begin{align}
\nonumber
\sigma^{2}_{Y_{i_{j}}^{\mathrm{G}}  | S_{i_{j}}^{\mathrm{G}}} & =
1 + \sum_{s_{i_{j}}'' \in \langle l_{i_{j}} \rangle''} \! \! \big| h_{i_{j}i_{j}}^{[s_{i_{j}}'']} \big|^{2} P_{i_{j}}^{[s_{i_{j}}'']} + \! \! \! \! \! \!\sum_{s_{i_{j+1}} \in  \langle l_{i_{j+1}} \rangle'}  \! \!  \big| h_{i_{j+1}i_{j}}^{[s_{i_{j+1}}]} \big|^{2} P_{i_{j+1}}^{[s_{i_{j+1}}]}
+ \frac{\sum_{s_{i_{j}} \in \langle l_{i_{j}} \rangle' } \big| h_{i_{j}i_{j}}^{[s_{i_{j}}]} \big|^{2} P_{i_{j}}^{[s_{i_{j}}]}}{
| g_{i_{j}} |^{2}
\sum_{s_{i_{j}} \in \langle l_{i_{j}} \rangle' } \big| h_{i_{j}i_{j}}^{[s_{i_{j}}]} \big|^{2} P_{i_{j}}^{[s_{i_{j}}]} + 1} \\
\nonumber
& \leq 1 + \sum_{s_{i_{j}}'' \in \langle l_{i_{j}} \rangle''} \! \! \big| h_{i_{j}i_{j}}^{[s_{i_{j}}'']} \big|^{2} P_{i_{j}}^{[s_{i_{j}}'']} + \! \! \! \! \! \!\sum_{s_{i_{j+1}} \in  \langle l_{i_{j+1}} \rangle'}  \! \!  \big| h_{i_{j+1}i_{j}}^{[s_{i_{j+1}}]} \big|^{2} P_{i_{j+1}}^{[s_{i_{j+1}}]}
+ \big| \langle l_{i_{j}} \rangle' \big| \frac{\big| h_{i_{j}i_{j}}^{[l_{i_{j}}]} \big|^{2} P_{i_{j}}^{[l_{i_{j}}]}}
{\big| h_{i_{j}i_{j-1}}^{[l_{i_{j}}]} \big|^{2} P_{i_{j}}^{[l_{i_{j}}]}} \\
\label{eq:cond_var_UB}
& \leq 1 + (l_{i_{j}} + l_{i_{j+1}}) \frac{\big| h_{i_{j}i_{j}}^{[l_{i_{j}}]} \big|^{2} P_{i_{j}}^{[l_{i_{j}}]}}
{\big| h_{i_{j}i_{j-1}}^{[l_{i_{j}}]} \big|^{2} P_{i_{j}}^{[l_{i_{j}}]}}
\end{align}
where the inequality in \eqref{eq:cond_var_UB} follows by employing the TIN-optimality conditions in \eqref{eq:TIN_condition_2} and \eqref{eq:TIN_condition_1}.
In particular, the condition in \eqref{eq:TIN_condition_1} gives us
\begin{equation}
\alpha_{i_{j}i_{j}}^{[l_{i_{j}}]} \geq   \alpha_{i_{j}i_{j-1}}^{[l_{i_{j}}]}  +  \alpha_{i_{j+1}i_{j}}^{[s_{i_{j+1}}]}
\; \Leftrightarrow \;
\frac{\big| h_{i_{j}i_{j}}^{[l_{i_{j}}]} \big|^{2} P_{i_{j}}^{[l_{i_{j}}]}}
{\big| h_{i_{j}i_{j-1}}^{[l_{i_{j}}]} \big|^{2} P_{i_{j}}^{[l_{i_{j}}]}} \geq
\big| h_{i_{j+1}i_{j}}^{[s_{i_{j+1}}]} \big|^{2} P_{i_{j+1}}^{[s_{i_{j+1}}]}
\end{equation}
while from the condition in \eqref{eq:TIN_condition_2}, combined with the partition in \eqref{eq:l_i_j_prime_subset_2}, we obtain
\begin{equation}
\alpha_{i_{j}i_{j}}^{[l_{i_{j}}]}  - \alpha_{i_{j}i_{j-1}}^{[l_{i_{j}}]} \geq   \alpha_{i_{j}i_{j}}^{[s_{i_{j}}'']}
\; \Leftrightarrow \;
\frac{\big| h_{i_{j}i_{j}}^{[l_{i_{j}}]} \big|^{2} P_{i_{j}}^{[l_{i_{j}}]}}
{\big| h_{i_{j}i_{j-1}}^{[l_{i_{j}}]} \big|^{2} P_{i_{j}}^{[l_{i_{j}}]}} \geq
\big| h_{i_{j}i_{j}}^{[s_{i_{j}}'']} \big|^{2} P_{i_{j}}^{[s_{i_{j}}'']}, \; \forall s_{i_{j}}'' \in \langle l_{i_{j}} \rangle''.
\end{equation}
Plugging the upper bound in \eqref{eq:cond_var_UB} into \eqref{eq:conditional_entropy_Gaussian_2}, which in turn, is plugged into \eqref{eq:sum_rate_bound_1}, the bound in \eqref{eq:capacity_outer_2} is obtained. This completes the proof.
\section{Conclusion and Future Directions}
In this work, we considered the problem of TIN-optimality in the Gaussian IMAC,  motivated by uplink scenarios in cellular networks.
We proposed an adequate definition of TIN for cellular networks in which each cell carries out a power-controlled version of its
capacity achieving  strategy while treating all inter-cell interference as noise.
According to this definition, we derived a TIN-achievable GDoF region for the IMAC
through a novel application of the potential graph approach.
Then we identified two regimes of interest: 1) a TIN-convexity regime in which the proposed
TIN-achievable GDoF region is convex without the need for time-sharing, and
2) a TIN-optimality regime, contained in the TIN-convexity regime, in which the TIN-achievable GDoF region is optimal and
leads to a constant-gap characterization of the capacity region.

An interesting future direction is to employ the identified conditions and GDoF
characterizations to design efficient scheduling and power control  algorithms.
TIN-inspired scheduling algorithms for device-to-device communications,
 modeled by the regular $K$-user IC,  were proposed in \cite{Yi2016,Naderializadeh2014}.
Moreover,  a number of GDoF-based, low-complexity power allocation algorithms
for TIN in the regular IC were proposed in \cite{Geng2016,Yi2016,Geng2018}.
Finding similar efficient scheduling and power allocation algorithms for the IMAC (and cellular scenarios in general) is of
great practical importance.


Another interesting direction following this work is to consider downlink scenarios.
Such scenarios are modeled by the Gaussian interfering broadcast channel (IBC).
The TIN definition proposed here for cellular settings extends to the IBC, where superposition coding and successive decoding
can be employed in each cell while treating inter-cell interference as noise.
It is of interest to investigate the relationship between the IMAC and IBC under TIN and whether a form of uplink-downlink duality holds,
from which solving one problem leads directly to a solution for the other.
Some progress along these lines was recently reported in \cite{Joudeh2019}.
\appendix
\section{Numerical Evaluations in Simple Cellular Models}
\label{appendix:numerical_evaluations}
\begin{figure}
\centering
\includegraphics[width = 1.0\textwidth,trim={0.9cm 0cm 0cm 9.5cm},clip]{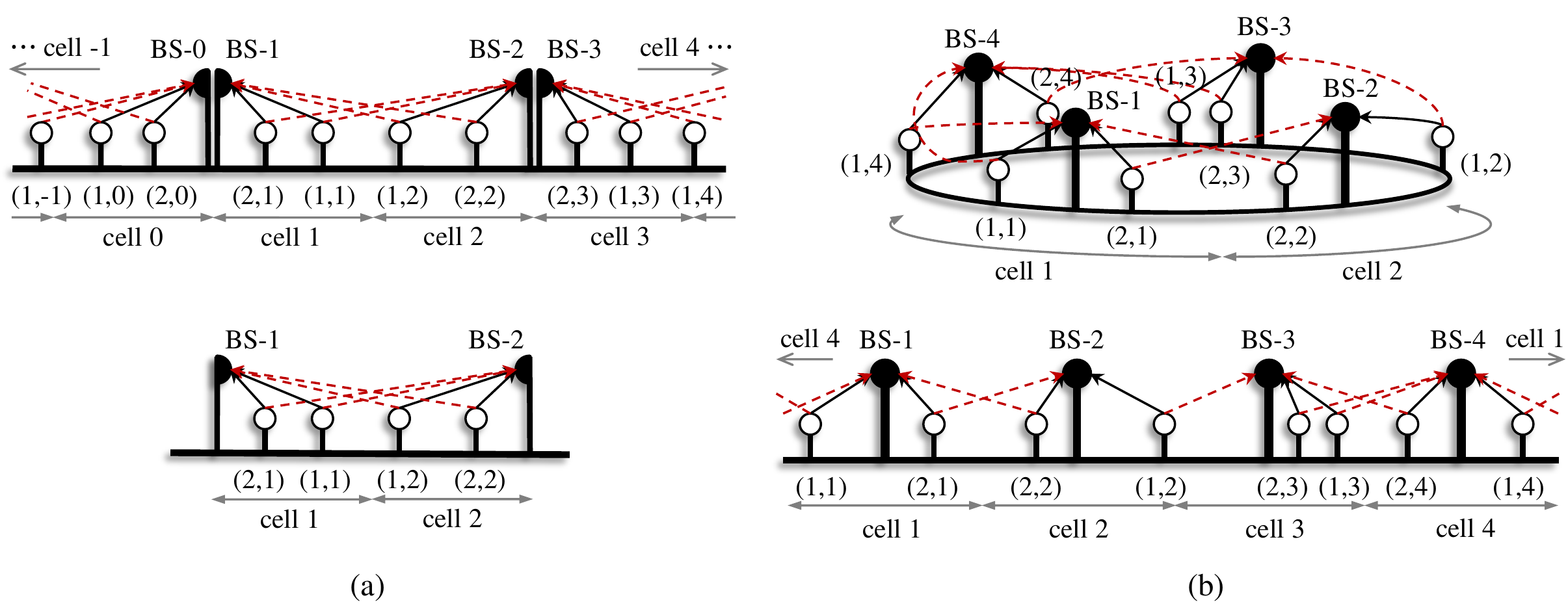}
\caption{\small Two simple cellular arrangements: a sectorized linear cell-array in (a) and a
circular cell-array in (b). Original arrangements are given in top figures, while bottom figures show
equivalent arrangements sufficient to test for TIN conditions. Users are denoted by their index tuples for brevity.}
\label{fig:cell_networks}
\end{figure}
In this appendix, we evaluate the probability that the TIN conditions, identified in Theorem \ref{theorem:polyhedrality}
and Theorem \ref{theorem:TIN_optimality}, are satisfied in simple cellular scenarios with fixed base station locations and random user locations.
For simplicity, we restrict our attention to the influence of distance-dependent path loss while neglecting shadowing and small-scale fading effects.
We consider the two following cellular arrangements, which are essentially variants of the classical and modified Wyner models \cite{Wyner1994,Somekh2007}:
\begin{enumerate}[(a)]
\item \emph{Sectorized linear cell-array}: In this model, cellular \emph{sites} are uniformly ordered in a linear array, where each site covers a segment of length $2r$ and is placed at the center of such segment. We assume sectorization where each site consists of two base stations, each with a directional antenna pointing in a distinct direction (left or right). Hence, a base station covers a cell (or sector) of length $r$ in its corresponding direction, e.g. Fig. \ref{fig:cell_networks}(a)-top.
\item \emph{Circular cell-array}: In this model, $K$ sites are uniformly ordered in a circular array. As in the above model, each site is
placed at the center of a $2r$ long segment. Unlike the above model however, here we do not assume sectorization. Instead, each site consists of one base station with an omnidirectional antenna covering a cell of length $2r$, e.g. Fig. \ref{fig:cell_networks}(b)-top.
\end{enumerate}
For both arrangements, $L$ users are randomly and independently placed in each cell
with locations drawn from a uniform distribution over the cell segment, while excluding a
segment of length $2r_{0}$ about the center of each site.
Under distance-dependent path loss, strength levels of different links are determined by the corresponding distances.
Therefore, the above arrangements enjoy the property that desired links are stronger than interfering links.
Moreover, it can be easily verified that for the purpose of checking GDoF-based TIN conditions in the above settings, there is no loss of generality in making the common assumption that interference is limited to adjacent cells.

In our numerical evaluations, we focus on $2$ cells for the sectorized linear model (see Fig. \ref{fig:cell_networks}(a)-bottom), as each pair of interfering cells can be treated as an independent network. For the circular model, we consider a network with $4$ cells as shown in Fig. \ref{fig:cell_networks}(b).
The distance-dependent path loss is modeled as $\mathrm{PL}(d) = 148.1 + 37.6\log_{10}(d)$ in dB,
where $d$ is the distance in kilometers.
Each user has a transmit power of $23$ dBm, while the base station noise floor is given by $-102$ dBm
(i.e. noise power spectral density: $-172$ dBm/Hz, receiver noise figure: $2$ dB, and transmission bandwidth: $10$ MHz).
We set $r_{0}$ to $35$ meters, while $r$ and $L$ are varied.
The results of our numerical evaluations are shown in Fig. \ref{fig:TIN_prob},
where each probability value is calculated from $10^{4}$ random user placements.
\begin{figure}
\centering
\includegraphics[width = 0.9\textwidth,trim={1.8cm 0cm 1.8cm 0cm},clip]{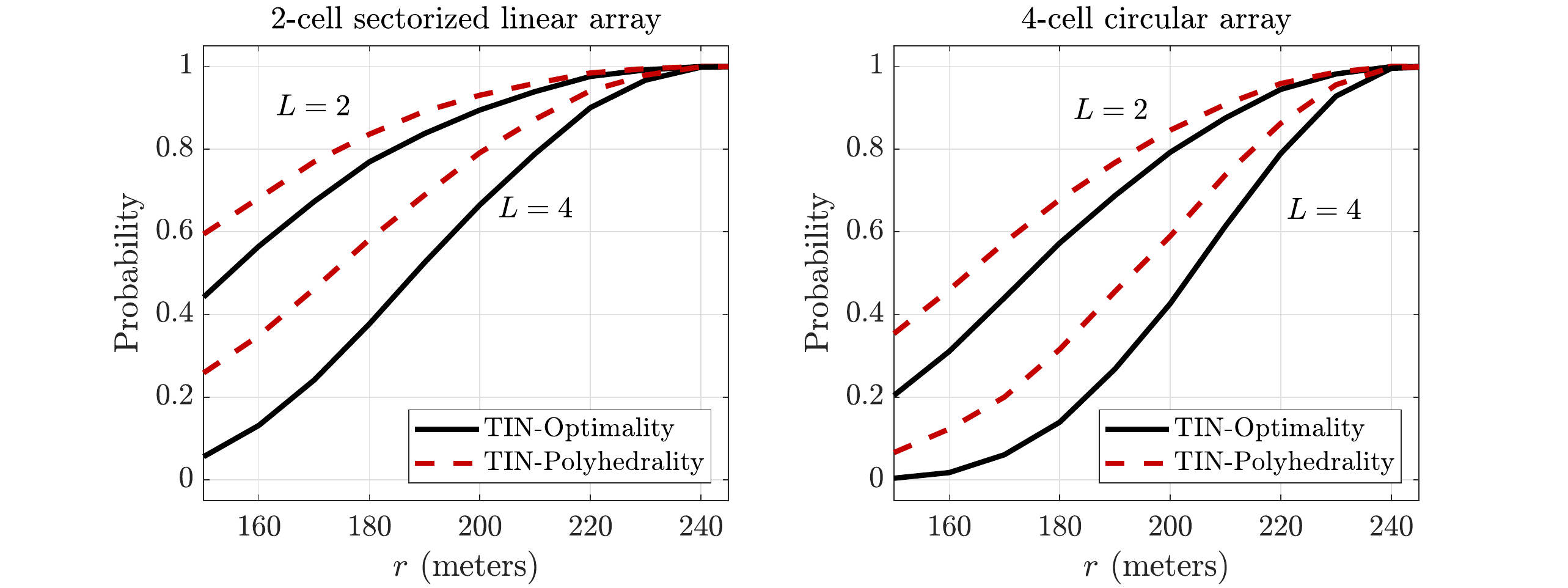}
\caption{\small Influence of cell size (determined by $r$) and number of users per cell (denoted by $L$) on the probabilities that the TIN-convexity conditions (Theorem \ref{theorem:polyhedrality}) and the TIN-optimality conditions (Theorem \ref{theorem:TIN_optimality}) hold
in a sectorized linear cell-array with 2 cells (left) and a circular cell-array with 4 cells (right).}
\label{fig:TIN_prob}
\end{figure}

As a direct consequence of Remark \ref{remark:TIN_poly_includes_TIN_opt}, the probability that the TIN-convexity conditions of Theorem \ref{theorem:polyhedrality} hold in a given setting is no less than the probability that the TIN-optimality conditions of Theorem \ref{theorem:TIN_optimality} hold in the same setting, which is clearly exhibited in Fig. \ref{fig:TIN_prob}.
Such probabilities decrease with an increased number of users per cell (i.e. $L$), which is not surprising as more users
induce more conditions to be satisfied.
We also observe that all probabilities increase with cell size (determined by site radius $r$).
This is due to the fact that as the distance between adjacent cells increases, the effects of inter-cell interference become less
pronounced, making the TIN-convexity and optimality conditions more likely to hold (e.g. set cross link strengths to small values in
\eqref{eq:TIN_condition_2} and \eqref{eq:TIN_condition_1}).
For example, under the adopted system parameters, the cell-edge SNR is about $0$ dB for
$r=243$ meters, enabling both sets of TIN conditions to hold with
probability $1$ as inter-cell interference remains below noise level.
The results in Fig. \ref{fig:TIN_prob}, albeit restricted to simple cellular models, show the potential broadness of the regimes for which the TIN conditions identified in Theorem \ref{theorem:polyhedrality} and Theorem \ref{theorem:TIN_optimality} will hold in more realistic cellular settings.
\section{Proof of Lemma \ref{lemma:redundant_circuits}}
\label{appendix:proof_lemma:redundant_circuits}
First, we observe that  for any multi-user circuit $\mathbf{c}(\mathbf{s}^{n})$,
where $\mathbf{s}^{n} = \big(\mathbf{s}_{1}^{n_{1}},\ldots,\mathbf{s}_{m}^{n_{m}} \big) \in \Sigma(\mathcal{K})$
and $n\geq2$,  the corresponding GDoF inequality obtained from the
non-negative length condition of Lemma \ref{lemma:non_negative_circuit_length}
is expressed in terms of the single-cell partition as
\begin{equation}
\label{eq:GDoF_ineq_without_u}
\sum_{j = 1}^{m} \sum_{s_{j} = 1}^{n_{j}} d(e_{j}^{s_{j}}) \leq
\sum_{j = 1}^{m} \sum_{s_{j} = 1}^{n_{j}} \Big[ \alpha(e_{j}^{s_{j}}) - w(e_{j}^{s_{j}}) \Big].
\end{equation}
Moreover, it is useful to observe that for intra-cell directed edges, i.e.
edges connecting pairs of users belonging to the same cell, we have
\begin{equation}
\label{eq:intra_cell_edge_w}
w(e_{j}^{s_{j}}) = \alpha(e_{j}^{s_{j}+1}) \mathbbm{1}_{\overline{\mathcal{E}_{1}'}}(e_{j}^{s_{j}}), \; \forall s_{j} \in \langle n_{j} - 1 \rangle, \; j \in \langle m \rangle
\end{equation}
which follows from \eqref{eq:edges_in_circuits}, \eqref{eq:edges_alpha} and \eqref{eq:edges_w}.

\emph{Necessity of \ref{cond:redundant_2}}:
To show this, consider a directed circuit $\mathbf{c}(\mathbf{s}^{n})$,
as expressed in \eqref{eq:circuit_fn_1}, and suppose that it  violates \ref{cond:redundant_2}.
For this to hold, we must have $n \geq 4$.
Moreover, we assume without loss of generality that $i_{1} = i_{k}$, for some $2 < k < n$, and that $l_{k}^{n_{k}} > l_{1}^{n_{1}}$ (otherwise we rename the indices).
The resulting GDoF inequality obtained from the non-negative length condition
for this circuit is given by \eqref{eq:GDoF_ineq_without_u}.
We show that the same set of users traversed by $\mathbf{c}(\mathbf{s}^{n})$ can be used to construct two smaller directed circuits with
GDoF inequalities that imply \eqref{eq:GDoF_ineq_without_u}.
Let us define
\begin{equation}
n_{1}^{\star} \triangleq \max \big\{s_{1} \in \langle n_{1} \rangle : l_{k}^{n_{k}} > l_{1}^{s_{1}} \big\}
\end{equation}
which exists since $l_{k}^{n_{k}} > l_{1}^{n_{1}}$.
The first constructed directed circuit is given by
\begin{align}
\label{eq:distinct_cells_sub_circuit_1}
\mathbf{c}' = \big(e_{1}^{n_{1}^{\star}},\ldots,e_{1}^{n_{1}},\ldots,e_{k}^{1},\ldots,e_{k}^{n_{k} - 1}, \tilde{e}_{k}^{n_{k}}\big)
\end{align}
where $\tilde{e}_{k}^{n_{k}} = \big((l_{k}^{n_{k}},i_{k}),(l_{1}^{n_{1}^{\star}},i_{1})\big)$.
$\mathbf{c}'$ is a valid circuit of $\mathcal{G}_{\mathrm{p}}$ which yields the inequality given by
\begin{multline}
\label{eq:distinct_cells_sub_GDoF_1}
\sum_{s_{1} = n_{1}^{\star}}^{n_{1}} d(e_{1}^{s_{1}}) +
\sum_{j = 2}^{k} \sum_{s_{j} = 1}^{n_{j}} d(e_{j}^{s_{j}}) \leq  \sum_{j = 2}^{k-1} \sum_{s_{j} = 1}^{n_{j}} \Big[ \alpha(e_{j}^{s_{j}}) - w(e_{j}^{s_{j}}) \Big] +  \sum_{s_{1} = n_{1}^{\star}}^{n_{1}} \Big[ \alpha(e_{1}^{s_{1}}) - w(e_{1}^{s_{1}}) \Big]
\\+\sum_{s_{k} = 1}^{n_{k} - 1} \Big[ \alpha(e_{k}^{s_{k}}) - w(e_{k}^{s_{k}}) \Big] + \Big[ \alpha(e_{k}^{n_{k}}) - \alpha(e_{1}^{n_{1}^{\star}}). \Big]
\end{multline}
In the above inequality, we have used $d(\tilde{e}_{k}^{n_{k}}) = d(e_{k}^{n_{k}})$ and
$\alpha(\tilde{e}_{k}^{n_{k}}) = \alpha(e_{k}^{n_{k}})$, where $e_{k}^{n_{k}}$ is traversed by the original directed circuit $\mathbf{c}(\mathbf{s}^{n})$, in addition to
$w(\tilde{e}_{k}^{n_{k}}) = \alpha(e_{1}^{n_{1}^{\star}}) $ which follows from $i_{k} = i_{1}$, $l_{k}^{n_{k}} > l_{1}^{n_{1}^{\star}}$ and
\eqref{eq:intra_cell_edge_w}.
The second directed circuit is given by
\begin{align}
\label{eq:distinct_cells_sub_circuit_2}
\mathbf{c}'' = \big(e_{1}^{1},\ldots,e_{1}^{n_{1}^{\star}-1},\tilde{e}_{1}^{n_{1}^{\star}},e_{k+1}^{1},\ldots,e_{k+1}^{n_{k+1}},\ldots,e_{m}^{1},\ldots,e_{m}^{n_{m}}\big)
\end{align}
where $\tilde{e}_{1}^{n_{1}^{\star}} = \big((l_{1}^{n_{1}^{\star}},i_{1}),(l_{k+1}^{1},i_{k+1}) \big)$.
This is also a valid directed circuit of $\mathcal{G}_{\mathrm{p}}$
and its corresponding GDoF inequality is given by
\begin{equation}
\label{eq:distinct_cells_sub_GDoF_2}
\sum_{s_{1} = 1}^{n_{1}^{\star}} d(e_{1}^{s_{1}}) + \sum_{j = k+1}^{m} \sum_{s_{j} = 1}^{n_{j}} d(e_{j}^{s_{j}}) \leq   \!
\Big[ \alpha(e_{1}^{n_{1}^{\star}}) - w(e_{k}^{n_{k}}) \Big] +
\sum_{s_{1} = 1}^{n_{1}^{\star}-1} \! \Big[ \alpha(e_{1}^{s_{1}}) - w(e_{1}^{s_{1}}) \Big]
+\sum_{j = k+1}^{m} \sum_{s_{j} = 1}^{n_{j}} \! \Big[ \alpha(e_{j}^{s_{j}}) - w(e_{j}^{s_{j}}) \Big]
\end{equation}
where we have used $d(\tilde{e}_{1}^{n_{1}^{\star}}) = d(e_{1}^{n_{1}^{\star}})$ and $\alpha(\tilde{e}_{1}^{n_{1}^{\star}}) = \alpha(e_{1}^{n_{1}^{\star}})$, in addition to $w(\tilde{e}_{1}^{n_{1}^{\star}}) = w(e_{k}^{n_{k}})$ which follows from $i_{k} = i_{1}$ and \eqref{eq:edges_w}.
By adding the inequalities in \eqref{eq:distinct_cells_sub_GDoF_1} and \eqref{eq:distinct_cells_sub_GDoF_2}, we obtain
\begin{equation}
\label{eq:distinct_cells_GDoF}
d(e_{1}^{n_{1}^{\star}}) + \sum_{j = 1}^{m} \sum_{s_{j} = 1}^{n_{j}} d(e_{j}^{s_{j}}) \leq
\sum_{j = 1}^{m} \sum_{s_{j} = 1}^{n_{j}} \Big[ \alpha(e_{j}^{s_{j}}) - w(e_{j}^{s_{j}}) \Big].
\end{equation}
Since $d(e_{1}^{n_{1}^{\star}}) \geq 0$, the inequality in \eqref{eq:distinct_cells_GDoF} implies the inequality in \eqref{eq:GDoF_ineq_without_u}, and hence $\mathbf{c}(\mathbf{s}^{n})$ is redundant compared to $\mathbf{c}'$ and $\mathbf{c}''$.
Note that users associated with cells $i_{1}$ and $i_{k}$ are now cyclicly adjacent in $\mathbf{c}'$ and constitute one single-cell subsequence, while  $i_{k}$ does not appear in $\mathbf{c}''$.
If any of $\mathbf{c}'$ or $\mathbf{c}''$ still violates \ref{cond:redundant_2}, we apply the above argument recursively until all resulting circuits satisfy  \ref{cond:redundant_2}.

\emph{Necessity of \ref{cond:redundant_3}}:
Now we proceed to show the necessity of \ref{cond:redundant_3} while assuming that the condition in \ref{cond:redundant_2} is satisfied.
Consider an arbitrary subset of users $\mathcal{S} \subseteq \mathcal{K}$, where $|\mathcal{S}| = n \geq 2$.
Each directed circuit $\mathbf{c}(\mathbf{s}^{n})$, induced by a cyclic sequence $\mathbf{s}^{n} \in \Sigma(\mathcal{S})$ spanning all users in $\mathcal{S}$,
gives a different inequality for the same sum-GDoF $\sum_{(l,i) \in \mathcal{S}}d_{i}^{[l]}$. Such inequalities take the form in \eqref{eq:GDoF_ineq_without_u}.
As a first step, we show that a necessary condition for the non-redundancy of $\mathbf{c}(\mathbf{s}^{n})$ is
\begin{equation}
\label{eq:circuit_order_condition_2_n}
l_{j}^{2} > l_{j}^{3} > \cdots > l_{j}^{n_{j}}, \ \forall j \in \langle m \rangle.
\end{equation}
That is, apart from the first user in each single-cell subsequence $\mathbf{s}^{n_{j}}_{j}$, all following users should be ordered in a descending manner.
Considering the right-hand-side of \eqref{eq:GDoF_ineq_without_u}, we have
\begin{align}
\label{eq:GDoF_conditions_order_1}
\sum_{j=1}^{m} \sum_{s_{j} = 1}^{n_{j}} \!  \Big [ \alpha(e_{j}^{s_{j}}) - w(e_{j}^{s_{j}}) \Big]
& = \sum_{j=1}^{m} \! \Big[ \alpha(e_{j}^{1}) - w(e_{j}^{n_{j}}) +  \alpha(e_{j}^{2})\mathbbm{1}_{\mathcal{E}_{1}'}(e_{j}^{1}) \Big]
 \! \! +  \! \! \sum_{j=1}^{m} \! \sum_{s_{j} = 2}^{n_{j} - 1}  \! \Big[ \alpha(e_{j}^{s_{j} + 1}) - w(e_{j}^{s_{j}}) \Big] \\
\label{eq:GDoF_conditions_order_2}
& \geq \sum_{j=1}^{m} \Big[ \alpha(e_{j}^{1}) - w(e_{j}^{n_{j}}) + \alpha(e_{j}^{2})\mathbbm{1}_{\mathcal{E}_{1}'}(e_{j}^{1}) \Big].
\end{align}
The equality in \eqref{eq:GDoF_conditions_order_1} uses $\alpha(e_{j}^{2}) - w(e_{j}^{1}) = \alpha(e_{j}^{2})\mathbbm{1}_{\mathcal{E}_{1}'}(e_{j}^{1})$, which is obtained from \eqref{eq:intra_cell_edge_w}.
Note that if $n_{j} = 1$ for some $j \in \langle m \rangle$, then $\alpha(e_{j}^{2})\mathbbm{1}_{\mathcal{E}_{1}'}(e_{j}^{1}) = 0$ by definition of $\mathbbm{1}_{\mathcal{E}_{1}'}(\cdot)$, and $j$ does not contribute to
the double summation on the right-hand side of \eqref{eq:GDoF_conditions_order_1}.
The inequality in \eqref{eq:GDoF_conditions_order_2} follows from $ 0\leq  w(e_{j}^{s_{j}}) \leq \alpha(e_{j}^{s_{j} + 1})$, $s_{j} \leq n_{j}-1$, as seen from \eqref{eq:intra_cell_edge_w}.
Note that  \eqref{eq:GDoF_conditions_order_2} holds with equality when \eqref{eq:circuit_order_condition_2_n} is satisfied, yielding a tighter GDoF inequality compared to when \eqref{eq:circuit_order_condition_2_n} is violated.

We proceed by focusing on cyclic sequences $\mathbf{s}^{n} \in \Sigma(\mathcal{S})$ that satisfy both \ref{cond:redundant_2} and  \eqref{eq:circuit_order_condition_2_n}.
The next step is to show that for any such sequence, if $l_{j}^{1} < l_{j}^{2}$ for some $j \in \langle m \rangle$, then the corresponding GDoF inequality is redundant.
Suppose, without loss of generality, that we have $\mathbf{s}^{n}  = (\mathbf{s}_{1}^{n_{1}},\ldots,\mathbf{s}_{m}^{n_{m}})$ with
$n_{1} \geq 2$ and $l_{1}^{1} < l_{1}^{2}$.
The GDoF inequality obtained from $\mathbf{c}(\mathbf{s}^{n})$ is given by
\begin{equation}
\label{eq:GDoF_ordered_cells_0}
\sum_{j=1}^{m} \sum_{s_{j} = 1}^{n_{j}} d(e_{j}^{s_{j}}) \leq
\Big[ \alpha(e_{1}^{1}) - w(e_{1}^{n_{1}})
+ \alpha(e_{1}^{2}) \Big] +
\sum_{j=2}^{m} \Big[ \alpha(e_{j}^{1}) - w(e_{j}^{n_{j}})
+ \alpha(e_{j}^{2})\mathbbm{1}_{\mathcal{E}_{1}'}(e_{j}^{1}) \Big]
\end{equation}
where we have used \eqref{eq:GDoF_conditions_order_2} in addition to  $l_{1}^{1} < l_{1}^{2}$.
We construct two smaller directed circuits from the users traversed by
$\mathbf{c}(\mathbf{s}^{n})$ and show that their corresponding GDoF inequalities imply \eqref{eq:GDoF_ordered_cells_0}.
Consider the directed circuit given by
\begin{align}
\label{eq:ordered_cells_sub_circuit_1}
\tilde{\mathbf{c}}' = \big(\tilde{e}_{1}^{1},e_{2}^{1},\ldots,e_{2}^{n_{2}},\ldots,e_{m}^{1},\ldots,e_{m}^{n_{m}}\big)
\end{align}
where $\tilde{e}_{1}^{1} = \big((l_{1}^{1},i_{1}),(l_{2}^{1},i_{2})\big)$.
This directed circuit is valid for $\mathcal{G}_{\mathrm{p}}$, satisfies \ref{cond:redundant_2} and  \eqref{eq:circuit_order_condition_2_n},
and yields the GDoF inequality given by
\begin{equation}
\label{eq:GDoF_ordered_cells_1}
d(e_{1}^{1}) + \sum_{j=2}^{m} \sum_{s_{j} = 1}^{n_{j}} d(e_{j}^{s_{j}}) \leq \Big[ \alpha(e_{1}^{1}) - w(e_{1}^{n_{1}}) \Big] + \sum_{j=2}^{m} \Big[ \alpha(e_{j}^{1}) - w(e_{j}^{n_{j}})
+ \alpha(e_{j}^{2})\mathbbm{1}_{\mathcal{E}_{1}'}(e_{j}^{1}) \Big].
\end{equation}
where $d(\tilde{e}_{1}^{1}) = d(e_{1}^{1})$, $\alpha(\tilde{e}_{1}^{1}) = \alpha(e_{1}^{1})$
and  $w(\tilde{e}_{1}^{1}) = w(e_{1}^{n_{1}})$ are used in \eqref{eq:GDoF_ordered_cells_1}.
Now consider a second directed circuits  given by
\begin{align}
\label{eq:ordered_cells_sub_circuit_2}
\tilde{\mathbf{c}}'' = \big(e_{1}^{2},\ldots,e_{1}^{n_{1} - 1},\tilde{e}_{1}^{n_{1}}\big)
\end{align}
where $\tilde{e}_{1}^{n_{1}} = \big((l_{1}^{n_{1}},i_{1}),(l_{1}^{2},i_{1})\big)$.
This is a single-cell circuit with users ordered in a descending manner. The resulting GDoF inequality is given by
\begin{equation}
\label{eq:GDoF_ordered_cells_2}
\sum_{s_{1} = 2}^{n_{1}} d(e_{1}^{s_{1}})  \leq \alpha(e_{1}^{2}).
\end{equation}
It is readily seen that the inequality in \eqref{eq:GDoF_ordered_cells_0} is retrieved by adding the inequalities in \eqref{eq:GDoF_ordered_cells_1} and \eqref{eq:GDoF_ordered_cells_2}, hence $\mathbf{c}(\mathbf{s}^{n})$ is redundant compared to $\tilde{\mathbf{c}}'$ and $\tilde{\mathbf{c}}''$.
If $l_{j}^{1} < l_{j}^{2}$ for some $j \in \langle 2:m \rangle$ in $\tilde{\mathbf{c}}'$, we apply the same steps above recursively,
hence showing that non-redundancy necessitates
\begin{equation}
\label{eq:circuit_order_condition_1_n}
l_{j}^{1}  > l_{j}^{2} >  \cdots > l_{j}^{n_{j}}, \ \forall j \in \langle m \rangle.
\end{equation}
We are left with directed circuits $\mathbf{c}(\mathbf{s}^{n}) \in \Sigma(\mathcal{S})$ that satisfy  \ref{cond:redundant_2} and \eqref{eq:circuit_order_condition_1_n}, for which the corresponding GDoF inequalities take the form
\begin{equation}
\label{eq:GDoF_ordered_cells_3}
\sum_{j=1}^{m} \sum_{s_{j} = 1}^{n_{j}} d(e_{j}^{s_{j}}) \leq \sum_{j=1}^{m} \Big[ \alpha(e_{j}^{1}) - w(e_{j}^{n_{j}})\Big].
\end{equation}
The final step is to show that by including all users in
$\tilde{\mathcal{S}} = \big\{ (l_{j},i_{j}) : l_{j} \in \langle l_{j}^{1} \rangle \setminus \{l_{j}^{1}, \ldots,l_{j}^{n_{j}} \} , j \in \langle m \rangle \big\}$, we obtain a GDoF inequality that implies \eqref{eq:GDoF_ordered_cells_3}.
In particular, consider the cyclic sequence $\tilde{\mathbf{s}} = \big(\tilde{\mathbf{s}}_{1}^{n_{1}},\ldots,\tilde{\mathbf{s}}_{m}^{n_{m}} \big) \in \Sigma(\mathcal{S} \cup \tilde{\mathcal{S}})$, obtained by augmenting each single-cell subsequence $\mathbf{s}_{j}^{n_{j}}$ in $\mathbf{s}^{n}$
as $\tilde{\mathbf{s}}_{j}^{n_{j}} = \big((l_{j}^{1},i_{j}),(l_{j}^{1} - 1,i_{j}),\ldots,(1,i_{j}) \big)$, $\forall j \in \langle m \rangle$.
The corresponding directed circuit is given by
\begin{equation}
\mathbf{c}(\tilde{\mathbf{s}}) = (\tilde{e}_{1}^{1},\ldots,\tilde{e}_{1}^{l_{1}^{1}},\ldots,\tilde{e}_{j}^{1},\ldots,\tilde{e}_{j}^{l_{j}^{1}})
\end{equation}
where edges are defined as in \eqref{eq:edges_in_circuits}, but with respect to the cyclic sequence $\tilde{\mathbf{s}}$.
From the non-negative circuit length condition,  $\mathbf{c}(\tilde{\mathbf{s}})$ yields the GDoF inequality given by
\begin{equation}
\label{eq:GDoF_ordered_cells_4}
\sum_{j=1}^{m} \sum_{s_{j} = 1}^{l_{j}^{1}} d(\tilde{e}_{j}^{s_{j}})\leq \sum_{j=1}^{m} \Big[ \alpha(\tilde{e}_{j}^{1})-
w(\tilde{e}_{j}^{l_{j}^{1}})\Big].
\end{equation}
Note that every user traversed by $\mathbf{c}(\mathbf{s}^{n})$ is also traversed by $\mathbf{c}(\tilde{\mathbf{s}})$,
which may also traverse additional users.
On the other hand, we have
$\alpha(\tilde{e}_{j}^{1})- w(\tilde{e}_{j}^{l_{j}^{1}}) = \alpha(e_{j}^{1})- w(e_{j}^{n_{j}})$, $\forall j \in \langle m \rangle$.
Therefore, \eqref{eq:GDoF_ordered_cells_4} implies \eqref{eq:GDoF_ordered_cells_3}, hence showing the necessity of \ref{cond:redundant_3}.
\section{Proof of Lemma \ref{lemma:entropy_diff}}
\label{appendix:proof_lemma:entropy_diff}
We start by finding an upper bound for $h(Y_{a}^{n}) - h(Y_{b}^{n}) $ through the following steps
\begin{align}
\nonumber
h(Y_{a}^{n}) - h(Y_{b}^{n}) & = h(Y_{a}^{n}) - h(Y_{b}^{n}) - h(Z_{a}^{n})  + h(Z_{b}^{n}) \\
\nonumber
& = I\big( X_{1}^{n},\ldots,X_{l}^{n} ; Y_{a}^{n} \big)  - I\big( X_{1}^{n},\ldots,X_{l}^{n} ; Y_{b}^{n} \big)  \\
\nonumber
& = \sum_{i = 1}^{l} \Big[ I\big( X_{i}^{n} ; Y_{a}^{n} | X_{1}^{n},\ldots,X_{i-1}^{n} \big) -
I\big( X_{i}^{n} ; Y_{b}^{n} | X_{1}^{n},\ldots,X_{i-1}^{n} \big) \Big] \\
\label{eq:lemma_entropy_diff_ub_2}
& = \sum_{i = 1}^{l} \Big[ I\big( X_{i}^{n} ; a_{i}X_{i}^{n} + \cdots + a_{l}X_{l}^{n} + Z_{a}^{n} \big) -
I\big( X_{i}^{n} ; b_{i}X_{i}^{n} + \cdots + b_{l}X_{l}^{n} + Z_{b}^{n}\big) \Big] \\
\label{eq:lemma_entropy_diff_ub_3}
& \leq \sum_{i = 1}^{l} \Big[ I\big( X_{i}^{n} ; a_{i}X_{i}^{n} + Z_{a}^{n} \big) -
I\big( X_{i}^{n} ; b_{i}X_{i}^{n} + \cdots + b_{l}X_{l}^{n} + Z_{b}^{n}\big) \Big]
\end{align}
where \eqref{eq:lemma_entropy_diff_ub_2} is due to the independence of all input sequences and noise, and \eqref{eq:lemma_entropy_diff_ub_3} follows from the data processing inequality \cite{Cover2012}.
Now we focus on the difference between the mutual information terms in \eqref{eq:lemma_entropy_diff_ub_3} for a given $i \in \langle l \rangle$.
Defining $\tilde{b}_{i+1} \triangleq \sqrt{P_{i+1}}b_{i+1}$, the mutual information term with the negative sign is bounded below as
\begin{align}
\label{eq:lemma_entropy_diff_ub_2_1}
I\big( X_{i}^{n} ; & b_{i}X_{i}^{n} + \cdots + b_{l}X_{l}^{n} + Z_{b}^{n}\big)  =
I\Big( X_{i}^{n} ; \frac{b_{i}}{\tilde{b}_{i+1}}X_{i}^{n} + \cdots + \frac{b_{l}}{{\tilde{b}_{i+1}}}X_{l}^{n} + \frac{1}{{\tilde{b}_{i+1}}}Z_{b}^{n}\Big)   \\
\label{eq:lemma_entropy_diff_ub_2_2}
&  \geq I\Big( X_{i}^{n} ; \frac{b_{i}}{\tilde{b}_{i+1}}X_{i}^{n} + \cdots + \frac{b_{l}}{{\tilde{b}_{i+1}}}X_{l}^{n} + Z_{b}^{n}\Big) \\
\nonumber
&  = h\Big(\frac{b_{i}}{\tilde{b}_{i+1}}X_{i}^{n} + \cdots + \frac{b_{l}}{{\tilde{b}_{i+1}}}X_{l}^{n} + Z_{b}^{n}\Big)
- h\Big(\frac{b_{i+1}}{\tilde{b}_{i+1}}X_{i+1}^{n} + \cdots + \frac{b_{l}}{{\tilde{b}_{i+1}}}X_{l}^{n} + Z_{b}^{n}\Big) \\
\label{eq:lemma_entropy_diff_ub_2_3}
& \geq
h\Big(\frac{b_{i}}{{\tilde{b}_{i+1}}}X_{i}^{n} + Z_{b}^{n}\Big)
- h\Big(\frac{b_{i+1}}{\tilde{b}_{i+1}}X_{i+1}^{n} + \cdots + \frac{b_{l}}{{\tilde{b}_{i+1}}}X_{l}^{n} + Z_{b}^{n}\Big) \\
\nonumber
& =
I\Big(X_{i}^{n} ; \frac{b_{i}}{{\tilde{b}_{i+1}}}X_{i}^{n} + Z_{b}^{n}\Big)
- I\Big(X_{i+1}^{n},\ldots,X_{l}^{n} ; \frac{b_{i+1}}{\tilde{b}_{i+1}}X_{i+1}^{n} + \cdots + \frac{b_{l}}{{\tilde{b}_{i+1}}}X_{l}^{n} + Z_{b}^{n}\Big) \\
\label{eq:lemma_entropy_diff_ub_2_4}
& \geq I\Big(X_{i}^{n} ; a_{i}X_{i}^{n} + Z_{b}^{n}\Big)
- I\Big(X_{i+1}^{n},\ldots,X_{l}^{n} ; \frac{b_{i+1}}{\tilde{b}_{i+1}}X_{i+1}^{n} + \cdots + \frac{b_{l}}{{\tilde{b}_{i+1}}}X_{l}^{n} + Z_{b}^{n}\Big).
\end{align}
The inequality in $\eqref{eq:lemma_entropy_diff_ub_2_2}$  follows from $|\tilde{b}_{i+1}|^{2} \geq 1$ (see \eqref{eq:lemma_entropy_diff_cond_1} and $P_{l+1}|b_{l+1}|^{2} = 1$), which makes  the output in \eqref{eq:lemma_entropy_diff_ub_2_1} less noisy compared to
the output in \eqref{eq:lemma_entropy_diff_ub_2_2}.
The inequality in \eqref{eq:lemma_entropy_diff_ub_2_3} follows by conditioning the differential entropy with the positive sign and the independence of input sequences and noise.
\eqref{eq:lemma_entropy_diff_ub_2_4} follows from $|a_{i}|^{2} \leq \frac{|b_{i}|^{2}}{|\tilde{b}_{i+1}|^{2}}$ in \eqref{eq:lemma_entropy_diff_cond_1};  this is similar to a Gaussian degraded broadcast channel with input $X_{i}^{n} $ and outputs $a_{i}X_{i}^{n} + Z_{b}^{n}$ and $\frac{b_{i}}{{\tilde{b}_{i+1}}}X_{i}^{n} + Z_{b}^{n}$ \cite{Cover2012}.

By combining the bounds in \eqref{eq:lemma_entropy_diff_ub_3} and \eqref{eq:lemma_entropy_diff_ub_2_4}, we proceed as follows
\begin{align}
\nonumber
h(Y_{a}^{n}) - h(Y_{b}^{n}) & \leq  \sum_{i =1}^{l} I\Big(X_{i+1}^{n},\ldots,X_{l}^{n} ; \frac{b_{i+1}}{\tilde{b}_{i+1}}X_{i+1}^{n} + \cdots + \frac{b_{l}}{{\tilde{b}_{i+1}}}X_{l}^{n} + Z_{b}^{n}\Big) \\
\nonumber
& = \sum_{i =1}^{l} \bigg[ h\bigg(\frac{b_{i+1}}{\tilde{b}_{i+1}}X_{i+1}^{n} + \cdots + \frac{b_{l}}{{\tilde{b}_{i+1}}}X_{l}^{n} + Z_{b}^{n}\bigg) -
h\big( Z_{b}^{n}\big) \bigg] \\
\label{eq:lemma_entropy_diff_ub_3_1}
& \leq  n\sum_{i =1}^{l} \log\bigg(1 + \sum_{j=i+1}^{l} \frac{P_{j}|b_{j}|^{2}}{P_{i+1}|b_{i+1}|^{2}} \bigg) \\
\label{eq:lemma_entropy_diff_ub_3_2}
& \leq  n\sum_{i=1}^{l}\log(i) .
\end{align}
where \eqref{eq:lemma_entropy_diff_ub_3_1} follows by a direct application of the inequality in \cite[(2.8)]{ElGamal2011}
and \eqref{eq:lemma_entropy_diff_ub_3_2} holds because $P_{j}|b_{j}|^{2} \leq P_{k}|b_{k}|^{2}, \forall j \geq k$ (see \eqref{eq:lemma_entropy_diff_cond_1}).
Finally, \eqref{eq:lemma_entropy_diff} follows from \eqref{eq:lemma_entropy_diff_ub_3_2}, which completes the proof.
\section{Proof of Lemma \ref{lemma:user_partition}}
\label{appendix:proof_lemma:user_partition}
From the TIN condition in \eqref{eq:TIN_condition_2} and the definition of the partition in \eqref{eq:partition_prime_2}  and \eqref{eq:partition_prime_1}, since $\alpha_{ii}^{[l_{i}]} - \alpha_{ij}^{[l_{i}]} < \alpha_{i i}^{[s_{i}']}$, $\forall
s_{i} \in \langle l_{i} \rangle'_{j}\setminus\{l_{i}\}$, then we must have
\begin{equation}
\label{eq:partition_lemma_1}
\alpha_{ii}^{[l_{i}]} - \alpha_{ij}^{[l_{i}]} \geq  \alpha_{ii}^{[s_{i}]} - \alpha_{ij}^{[s_{i}]} + \alpha_{ij}^{[l_{i}]}, \; \forall
s_{i} \in \langle l_{i} \rangle'_{j}\setminus\{l_{i}\}.
\end{equation}
As a first step of the proof, we show that \eqref{eq:partition_lemma_1} holds in a more general sense, such that
\begin{equation}
\label{eq:partition_lemma_2}
\alpha_{ii}^{[l_{i}']} - \alpha_{ij}^{[l_{i}']} \geq  \alpha_{ii}^{[s_{i}']} - \alpha_{ij}^{[s_{i}']} + \alpha_{ij}^{[l_{i}']}, \; \forall
s_{i}',l_{i}' \in \langle l_{i} \rangle'_{j} , \; s_{i}'<l_{i}'.
\end{equation}
Suppose that \eqref{eq:partition_lemma_2} does not hold for some $s_{i}' < l_{i}' < l_{i}$.
The TIN conditions in \eqref{eq:TIN_condition_2} dictates that we must have
$\alpha_{ii}^{[l_{i}']} - \alpha_{ij}^{[l_{i}']} \geq \alpha_{ii}^{[s_{i}']} $ instead.
Combining this with \eqref{eq:partition_lemma_1}, we obtain
\begin{align}
\nonumber
\alpha_{ii}^{[l_{i}]} - \alpha_{ij}^{[l_{i}]} & \geq \alpha_{ii}^{[l_{i}']} - \alpha_{ij}^{[l_{i}']} + \alpha_{ij}^{[l_{i}]} \\
\nonumber
& \geq \alpha_{ii}^{[s_{i}']}  + \alpha_{ij}^{[l_{i}]} \\
\label{eq:partition_lemma_3}
& \geq \alpha_{ii}^{[s_{i}']}
\end{align}
which yields a contradiction since $s_{i}' \notin \langle l_{i} \rangle_{j}''$, and hence \eqref{eq:partition_lemma_3} must not hold.
Therefore, \eqref{eq:partition_lemma_2} must hold and we have
\begin{align}
\nonumber
\alpha_{ii}^{[l_{i}]} - \alpha_{ij}^{[l_{i}]} & \geq \alpha_{ii}^{[l_{i}']} - \alpha_{ij}^{[l_{i}']} + \alpha_{ij}^{[l_{i}]} \\
\nonumber
& \geq \alpha_{ii}^{[s_{i}']} - \alpha_{ij}^{[s_{i}']} + \alpha_{ij}^{[l_{i}']} + \alpha_{ij}^{[l_{i}]} \\
\nonumber
& \geq \alpha_{ii}^{[s_{i}']} - \alpha_{ij}^{[s_{i}']} + \alpha_{ij}^{[l_{i}']}
\end{align}
which completes the proof.
\section*{Acknowledgement}
The authors wish to thank the anonymous reviewers for their valuable and timely comments.
The authors also wish to thank Mr. Enrico Piovano for his careful reading of an earlier draft of this paper.
H. Joudeh gratefully acknowledges helpful discussions with Prof. Syed A. Jafar regarding Lemma \ref{lemma:entropy_diff}.
\bibliographystyle{IEEEtran}
\bibliography{References}
\end{document}